\documentclass[12pt]{article}
\usepackage[pass,paperwidth=8.5in,paperheight=11in]{geometry}

\usepackage{setspace} 
\usepackage{color}
\usepackage{ifthen}

\usepackage{amsfonts,amsmath,amssymb,amsthm,enumerate,CJK,array}
\usepackage{bm,bbm,dsfont,multicol}
\usepackage{graphicx,graphics,psfrag,booktabs}
\usepackage{natbib}
\usepackage[hypertexnames=false]{hyperref}  
\usepackage{caption}
\usepackage{appendix}
\usepackage{titlesec}
\usepackage{rotating,threeparttable,booktabs}
\usepackage{longtable, lscape,rotating}
\usepackage[]{algorithm2e}
\usepackage{chngcntr} 
\usepackage{subfigure}
\usepackage[section]{placeins}
\usepackage{adjustbox}
\usepackage{blindtext}
\usepackage{pdflscape}
\usepackage{verbatim}
\usepackage[utf8]{inputenc}
\usepackage{algorithmicx}
\usepackage{algpseudocode}

\newcommand{\R}{\mathbb{R}}
\newcommand{\N}{\mathbb{N}}
\newcommand{\tr}{\mathop{\rm{tr}}}

\newcommand{\rank}{\mathop{\rm{rank}}}
\newcommand{\var}{\mathop{\mbox{Var}}}

\newcommand{\E}{\mbox{\sf E}}     
\renewcommand{\P}{\mathrm{P}}            

\newcommand{\IF}{\boldsymbol{1}}    
\def\defeq{\stackrel{\mathrm{def}}{=}}  

\def\B{\mbox{\boldmath $B$}}

\def\K{\mbox{$\mathcal K$}}
\def\p{\mbox{$\mathcal P$}}

\def\vep{\varepsilon}
\def\prox{S_{\lambda/M}}

\newcommand{\bSigma}{\bm{\Sigma}}
\newcommand{\hSigma}{\widehat{\bm{\Sigma}}}
\def\Gamah{\widehat{\bm{\Gamma}}_\tau}

\def\Gamat{\bm{\Gamma}_{\tau,T}}
\def\Gamai{\bm{\Gamma}_{\tau,\infty}}
\def\tG{\widehat{\bm{\Gamma}}_{\tau,\delta}}
\newcommand{\bGamma}{{\bm{\Gamma}}}
\newcommand{\bOmega}{\bm{\Omega}}
\newcommand{\bTheta}{\bm{\Theta}}
\newcommand{\bDelta}{\bm{\Delta}}
\newcommand{\hDelta}{\widehat{\bm{\Delta}}}

\def\tdel{\widehat{\bm{\Delta}}}

\newcommand{\ab}{\mathbf{a}}

\newcommand{\ub}{\mathbf{u}}
\newcommand{\vb}{\mathbf{v}}

\newcommand{\Ab}{\mathbf{A}}
\newcommand{\Bb}{\mathbf{B}}

\newcommand{\Db}{\mathbf{D}}

\newcommand{\Ib}{\mathbf{I}}

\newcommand{\Pb}{\mathbf{P}}

\newcommand{\Sb}{\mathbf{S}}

\newcommand{\Ub}{\mathbf{U}}
\newcommand{\Vb}{\mathbf{V}}
\newcommand{\Wb}{\mathbf{W}}
\newcommand{\Xb}{\mathbf{X}}
\newcommand{\Yb}{\mathbf{Y}}

\newcommand{\bW}{\bm{W}}
\newcommand{\bX}{\bm{X}}
\newcommand{\bY}{\bm{Y}}
\newcommand{\bZ}{\bm{Z}}

\newcommand{\bb}{\bm{b}}

\newcommand{\cA}{\mathcal{A}}

\newcommand{\cP}{\mathcal{P}}

\newcommand{\cK}{\mathcal{K}}
\newcommand{\cM}{\mathcal{M}}

\newcommand{\cS}{\mathcal{S}}
\newcommand{\cV}{\mathcal{V}}
\newcommand{\cU}{\mathcal{U}}
\newcommand{\fline}{\underline{f}^\tau}
\newcommand{\hQ}{\widehat{Q}}

\newcommand{\GG}{\mathbb{G}}

\newcommand{\PP}{\mathbb{P}}

\def\bvep{{\mbox{\boldmath $\varepsilon$}}}

\def\proofname{PROOF}

\def\ba{\bm{a}}

\def\bx{\bm{x}}

\def\bu{{\mbox{\boldmath $u$}}}
\def\bv{{\mbox{\boldmath $v$}}}

\def\bvep{{\mbox{\boldmath $\varepsilon$}}}
\def\sbvep{{\mbox{\scriptsize \boldmath $\varepsilon$}}}
\def\sbx{{\mbox{\scriptsize \boldmath $x$}}}

\newcommand{\IO}{\boldsymbol{0}} 
\renewcommand{\hat}{\widehat}
\renewcommand{\tilde}{\widetilde}

\newtheorem{theorem}{Theorem}[section]

\newtheorem{corollary}[theorem]{Corollary}
\newtheorem{lemma}[theorem]{Lemma}
\newtheorem{remark}[theorem]{Remark}
\newtheorem{defin}[theorem]{Definition}

\numberwithin{equation}{section}
\numberwithin{theorem}{section}
\counterwithin{figure}{section}  
\counterwithin{table}{section}

\author{Shih-Kang Chao\thanks{Department
		of Statistics, University of Missouri-Columbia, Columbia, MO 65211. E-mail:
		\texttt{chaosh@missouri.edu}. Tel: +1-573-882-1577. Fax: +1-573-415-8075.
		} \and Wolfgang K. H\"{a}rdle\thanks{Ladislaus von Bortkiewicz Chair of Statistics, C.A.S.E. - Center for applied Statistics and Economics, Humboldt-Universit\"{a}t zu Berlin, Unter den Linden 6, 10099 Berlin, Germany. Email: \texttt{haerdle@wiwi.hu-berlin.de.} Sim Kee Boon Institute for Financial Economics, Singapore Management University, 50 Stamford Road, Singapore 178899, Singapore.} \and Ming Yuan\thanks{Department of Statistics, Columbia University, 1255 Amsterdam Avenue, New York, NY 10027, U.S.A. Email: \texttt{ming.yuan@columbia.edu}.}} 

\titleformat{\section}{\Large}{\textbf\thesection.}{.5em}{\textbf}
\titlespacing{\section}{0pt}{*3}{*2}
\titleformat{\subsection}{\large}{\textbf\thesubsection.}{.5em}{\textbf}
\titlespacing{\subsection} {0pt}{*3}{*2}
\titleformat{\subsubsection}{\normalfont}{\textbf\thesubsubsection.}{.5em}{\textbf}
\titlespacing{\subsubsection} {0pt}{*3}{*2}

\counterwithin{figure}{section} 
\renewenvironment{proof}[1][\proofname]{{\noindent\bfseries #1.\ }}{\qed\medskip}

\begin{document}
\title{Factorisable Multitask Quantile Regression}
\maketitle
\begin{abstract}\singlespacing

A multivariate quantile regression model with a factor structure is proposed to study data with many responses of interest. The factor structure is allowed to vary with the quantile levels, which makes our framework more flexible than the classical factor models. The model is estimated with the nuclear norm regularization in order to accommodate the high dimensionality of data, but the incurred optimization problem can only be efficiently solved in an approximate manner by off-the-shelf optimization methods. Such a scenario is often seen when the empirical risk is non-smooth or the numerical procedure involves expensive subroutines such as singular value decomposition. To ensure that the approximate estimator accurately estimates the model, non-asymptotic bounds on error of the the approximate estimator is established. For implementation, a numerical procedure that provably marginalizes the approximate error is proposed. The merits of our model and the proposed numerical procedures are demonstrated through Monte Carlo experiments and an application to finance involving a large pool of asset returns. 
\end{abstract}

\noindent {\bf KEY WORDS:} Factor model; quantile regression; non-asymptotic analysis; multivariate regression; nuclear norm regularization.\\ 
\noindent {\bf JEL:} C13, C38, C61, G17.

\newpage
\doublespacing
\section{Introduction} 

In a variety of applications in economics, the interest is in the conditional quantiles of response variable \citep{KH01}. Quantile regression \citep{KB:1978} is arguably one of the most popular methods for estimating the quantile of a response variable. However, in the situation of multivariate responses with common predictors, equation-by-equation quantile regression fails to capture the latent common structure. In econometrics literature, factor models are used for quantile regression with multiple responses \citep{AT11,CDG15}, but the factors are usually invariant to the quantile level, or do not include the information of exogenous predictors. This seems to contradict with the reality that the upper and lower quantiles are usually interpreted differently. 

To fill this gap, we consider quantile level dependent factors $f_k^\tau(\bX)$ formed by covariates $\bX\in\R^p$, where $\tau\in(0,1)$ is the quantile level. We assume that the conditional quantile $q_j(\tau|\bX)$ of $Y_j$, the $j$th component in the response vector $\bY$, satisfies
\begin{align}
q_j(\tau|\bX) &= \sum_{k=1}^{r_\tau} \Psi_{\tau,kj} f_k^\tau(\bX),\quad j=1,...,m=\dim(\bY), \label{eq:fintro}
\end{align}
where $\Psi_{\tau,kj}\in\R$ is the factor loading, and $r_\tau\in\N$ is assumed small and quantile level dependent. At first glance, Model (1.1) appears to be a factor-augmented regression model (FAR) of \cite{SW02}, but in fact they are drastically different. The predictors and the response variables are both generated by factors in the FAR model, but (1.1) assumes that the factors are functions of predictors, where the functions are unknown but non-random. Besides the factor models, \cite{FXZ:13} consider transnormal models to allow for ultrahigh dimensional covariates.

From a practical perspective, the factors $f_k^\tau(\bX)$ in \eqref{eq:fintro} are unobservable, so computing the parameters $\Psi_{\tau,kj}$ can be challenging. To tackle this challenge, instead of pre-estimating the factors, we adopt a \emph{one-shot} approach that simultaneously estimates the factors and the loadings. In particular, assume additionally that $f_k^\tau(\bX_i)$ is linear in $\bX\in\R^p$, that is, $f_k^\tau(\bX_i) \defeq \Phi_{\tau,*k}^\top\bX_i$, where $\Phi_{\tau,*k}\in\R^p$. The model \eqref{eq:fintro} can be written as
\begin{align}
q_j(\tau|\bX_i) = \bGamma_{\tau,*j}^\top \bX_i, \quad i=1,...,n, \quad j=1,...,m=\dim(\bY),\label{eq:mintro}
\end{align}
where $\bGamma_{\tau,*j} = \Phi_\tau\Psi_{\tau,*j}$ is the $j$th column of $\bGamma_\tau$, which is a $p\times m$ coefficient matrix. Note that this implies $\bGamma_\tau = \Phi_\tau\Psi_\tau$, where $\Phi_\tau\in\R^{p\times r_\tau}$ linearly transforms $\bX$ to factors $\Phi_\tau^\top \bX_i=(f_1^\tau(\bX_i),f_2^\tau(\bX_i),...,f_{r_\tau}^\tau(\bX_i))^\top$, and $\Psi_\tau\in\R^{r_\tau\times m}$ are the loadings defined in \eqref{eq:fintro}, with $j$th column $\Psi_{\tau,*j}$ corresponding to output $Y_j$. 
Here, we normalize $\|\Psi_{\tau,k*}\|_2=1$ so that the model is identifiable. If the matrix $\bGamma_\tau$ in \eqref{eq:mintro} is available, a factorization of $\bGamma_\tau$ gives factors and loadings simultaneously; see Section \ref{sec.ex}. 


We observe that the rank of $\bGamma_\tau$ is the number of factors $r_\tau$, so the rank of $\bGamma_\tau$ is small when $r_\tau$ is small. Therefore, estimation of $\bGamma_\tau$ should exploit this low rankness property, so that the estimation can be stable even when both $m$ and $p$ are large relative to $n$. Recall that the rank corresponds to the number of nonzero singular values. Therefore, with a tuning parameter $\lambda>0$, an estimation with singular values regularization is considered
\begin{align}
\Gamah &\defeq \operatorname{arg}\,\underset{\Sb \in \R^{p \times m}}{\operatorname{min}}\big\{L_\tau(\Sb) \defeq \widehat Q_\tau(\Sb)+\lambda \|\Sb\|_*\big\},  \label{peml}
\end{align}
where $\|\Sb\|_*$ denotes the nuclear norm, which is the sum of singular values, and
\begin{align}
\widehat Q_\tau(\Sb)\defeq (mn)^{-1} \sum_{i=1}^n \sum_{j=1}^{m} \rho_\tau\big(Y_{ij}-\bX_i^\top\Sb_{\ast j}\big),	\label{eq:Qhat}
\end{align}
in which $\rho_\tau(u) = u(\tau-\IF\{u \leq 0\})$ is the "check function" \citep{KB:1978}. Note that $\widehat Q_\tau$ is similar to the loss function used in \cite{KP:90} for the low dimensional case, i.e. $m$ and $p$ do not diverge with $n$. For mean regression, the nuclear norm penalty has been considered by many authors \citep{YELM:2007,BSW:2011,NW:11,KLT:2011,MP13,MPR16}. 

In \eqref{peml}, $\Gamah$ is the minimizer of a convex empirical risk. Its theoretical guarantee has been well studied \citep{K:11} and its convergence rate is the best of all admissible estimators. Unfortunately, in reality, many off-the-shelf optimization methods only solves the optimization problem \eqref{peml} \emph{approximately} with an error $\delta>0$. The optimization error $\delta$ may be nonzero for various reasons. For example, $\tG$ may be computed with a surrogate loss function of $L_\tau$ that is easier to optimize, or $\tG$ is the outcome of an iterative procedure, in which each iteration involves costly subroutines such as singular value decomposition. A question arises naturally: \emph{can $\tG$ estimates $\bGamma_\tau$, as accurately as $\Gamah$ in some statistical sense?}  

The main goal of this paper is to provide an affirmative answer to the above question. In particular, we prove a non-asymptotic bound for $\tG-\bGamma_\tau$, that has similar convergence rate as that of $\Gamah-\bGamma_\tau$, under the condition that $\delta$ is smaller than an explicit upper bound depending on $n$, $m$ and $p$. This result is then be applied to prove that the estimator from a simple algorithm has the same convergence rate as $\Gamah$, because its optimization error can be marginalized if parameters of the algorithm are well-chosen. Theoretical tools developed in this paper may be potentially useful for other convex problems where finding an exact optimizer is expensive and unrealistic.

In addition to the theoretical guarantees, we experiment the numerical procedure using Monte Carlo simulations with i.i.d. and dependent design. The outputs variables are generated from a two-piece normal distribution \citep{W14}, which has been used for the inference of inflation rate by central banks \citep{W99}. The results show that our numerical procedure can correctly identify the number of factors. For an empirical illustration of our method, we estimate the market systemic risk from a large pool of assets, and compute the exposure of each asset to the systemic risk. As our method is scalable to high dimensional data, we are able to overcome the computational barriers inherent in the existing studies \citep{AB16,WKM10caviar2}.  

Lastly, we remark that multitask linear models have been considered in applications where the tails of distribution are the focus of interest. A recent paper \cite{CHH18} explores the functional magnetic resonance imaging (fMRI) data with multitask \emph{expectile} regression. They propose an iterative shrinkage algorithm, and show finite-sample convergence rate of the estimator while taking the optimization risk into account. The task undertaken in the current paper is much more challenging than \cite{CHH18} from both computational and theoretical aspects, because of the non-smoothness of the quantile regression loss function. The model in this paper is more appropriate when response variables are heavy-tailed as moment conditions are not required here.

\bigskip
The rest of this paper is organized as follows. Section \ref{sec.al} discusses the numerical procedure for estimating the coefficient matrix $\bGamma_\tau$ and the factors and loadings. The selection of $\lambda$ is also presented. Section \ref{sec:ora} provides non-asymptotic analysis for $\tG$ and characterizes the sufficient condition on $\delta$. The estimator for factors and loadings are also investigated. Results on Monte Carlo experiments are presented in Section \ref{sec.si}, and an application on financial systemic risk is shown in Section \ref{sec.app}. Section \ref{sec:conc} concludes this paper. Appendix contains the detailed development of the algorithm in Section \ref{sec.al} and the proof of key theoretical results. Other proofs and technical details are shifted to the supplementary materials.
\vspace{0.3cm}

\noindent\textbf{Notations.} 
Notations associated with matrices will be used extensively in this paper. For a matrix $\Ab = (A_{ij}) \in \R^{p \times m}$, denote the singular values of $\Ab$: $\sigma_1(\Ab) \geq \sigma_2(\Ab) \geq ... \geq \sigma_{p \wedge m}(\Ab)$. $\sigma_{\max}(\Ab)$ and $\sigma_{\min}(\Ab)$ for the largest and smallest singular values of $\Ab$. Let $\|\Ab\|=\sigma_{\max}(\Ab)$ be the spectral norm, $\|\Ab\|_*$ be the nuclear norm and $\|\Ab\|_{\rm F}$ be the Frobenius norm. Denote $\Ab_{\ast j}$ and $\Ab_{i \ast}$ as the $j$th column vector and the $i$th row vector of $\Ab$. $\Ib_p$ denotes the $p \times p$ identity matrix. For vectors $\ba_1,...,\ba_m$ in $\R^p$, denote $[\ba_1\ \ba_2\ ...\ \ba_m] \in \R^{p\times m}$ a matrix with $\ba_j$ being its $j$th column. $\IF(S)$ is the indicator function, which is one when the statement $S$ is true. 

\section{Approximate Estimator and Estimation}\label{sec.al}

In this paper, the approximate estimator $\tG$ is assumed to satisfy
	\begin{align}
		0\leq L_\tau(\tG)- L_\tau(\Gamah) \leq \delta. \label{eq:appest}
	\end{align}
	for some $\delta\geq 0$, where $L_\tau$ is the empirical risk in \eqref{peml}.

Section \ref{sec:coef} presents an algorithm that computes an $\tG$, and its optimization error $\delta$ will be characterized. Given $\tG$, Section \ref{sec.ex} describe the ways to estimate factors and loadings from $\tG$. Section \ref{sec:tune} discusses the choice of tuning parameter $\lambda$.

\subsection{Coefficient Matrix}\label{sec:coef}

The proposed estimation procedure combines the Fast Iterative Shrinkage-Thresholding Algorithm (FISTA) of \cite{BT:09} and the smoothing technique of \cite{N:05}. Similar approach is previously applied to estimate regression models with penalties that induce complex structural sparsity \citep{CLK:2012}. Comparing with other existing methods, this approach is more scalable to higher dimension than the semidefinite programming (SDP, \cite{FHB01,SRJ05}), and is more stable than the non-convex reformulations \citep{RS05,WKLS08,WKS08}. See \cite{CSP17} for a recent account on the latter issue. 

Specifically, the first step is to smooth $\hat Q_\tau$ in \eqref{eq:Qhat} with a surrogate function $\hat Q_{\tau,\kappa}$ parameterized by a smoothing parameter $\kappa$ \citep{N:05}. The surrogate loss function converges to the original loss function as $\kappa\to 0$. $\hat Q_{\tau,\kappa}$ has a good property that $\nabla \hat Q_{\tau,\kappa}$ is globally Lipschitz with constant $Lip=(\kappa m^2 n^2)^{-1}\|\Xb\|^2$, where $\Xb=[\bX_1 \,...\,\bX_n]^\top\in\R^{n\times p}$ is the design matrix. Next, FISTA \citep{BT:09} is applied on the modified loss function with step size $Lip^{-1}$. The procedure is summarized below, and the details are shifted to Section \ref{sec:alg_detail}.
\begin{enumerate}
	\item[Step 1:] Given $\kappa>0$, $\hat Q_{\tau,\kappa} = \texttt{SMOOTH}(\hat Q_\tau)$;
	\item[Step 2:] Given $\lambda>0$ and an initial estimator $\bGamma_{\tau,0}$, for each $t=1,2,..,T$, apply \texttt{FISTA} step on the minimization problem $\min_{\Sb\in\R^{p\times m}}\{\tilde L_\tau(\Sb)\defeq\hat Q_{\tau,\kappa}(\Sb)+\lambda\|\Sb\|_*\}$ with step size $\kappa m^2 n^2\|\Xb\|^{-2}$. Return the last iterate $\Gamat$.
\end{enumerate}
The quantity $\kappa$ is the key that controls both the smoothing quality and the step size of FISTA. Small $\kappa$ leads to smaller smoothing error, but slows down the convergence. Therefore, there exists a tradeoff between smoothing error and the speed of convergence. In our simulation and data application, we typically set $\kappa$ between $10^{-4}$ and $10^{-7}$, and 3000 to 4000 iterations are usually sufficient for convergence.

\bigskip
The next theorem shows that $\Gamat$ is an approximate estimator in the sense of \eqref{eq:appest}.

\begin{theorem}\label{thm.qr} 
Recall that $\Gamah$ is the optimal solution for minimizing \eqref{peml} and let $\Gamai\defeq\lim_{T\to\infty}\Gamat=\arg\,\min_{\Sb}\{\tilde L_\tau(\Sb)=\widehat Q_{\tau,\kappa}(\Sb)+\lambda \|\Sb\|_*\}$. 
Then for any $T$, $\tG=\bGamma_{\tau,T}$ satisfies \eqref{eq:appest} with
\begin{align}
	 &\delta = \delta(T,n,\kappa,\tau,\Xb)= \frac{3\kappa mn (\tau \vee \{1-\tau\})^2}{2} + \frac{2 \|\bGamma_{\tau,0}-\bGamma_{\tau,\infty}\|_{\rm F}^2}{(T+1)^2} \frac{\|\Xb\|^2}{\kappa m^2n^2}. \label{qr.bound}
\end{align}	
\end{theorem}
See Section \ref{proof.thmqr} in the supplementary material for a proof for Theorem \ref{thm.qr}. This theorem shows that the proposed numerical procedure provides an approximate optimizer in the sense of \eqref{eq:appest}. The first term on the right-hand side of \eqref{qr.bound} is related to the smoothing error $\hat Q_{\tau,\kappa}$ (Step 1), while the second term is related to the FISTA algorithm (Step 2). The quantile level $\tau$ enters \eqref{qr.bound} by the term $\tau \vee \{1-\tau\}$, which increases when $\tau$ approaches the boundaries of the interval $[0,1]$. 

\subsection{Factors number, factors and Loadings}\label{sec.ex}

Factorizing the true coefficient matrix $\bGamma_\tau = \Phi_\tau \Psi_\tau$ allows to compute the factors $f_k^\tau(\bX)=\Phi_{\tau,*k}^\top\bX_i$ and loadings $\Psi_{\tau,kj}$ for $j=1,...,m$ and $k=1,...,r_\tau$ at one shot. However, a potential problem here is that the decomposition $\bGamma_\tau = \Phi_\tau \Psi_\tau$ is not unique. In particular, for any invertible matrix $\Pb \in \R^{r\times r}$, we have $\Phi_\tau \Psi_\tau = \Phi_\tau \Pb\Pb^{-1}\Psi_\tau$. Therefore, to fix such a matrix $\Pb$, we apply the constraint in equation (2.14) on page 28 of \cite{RV:98}: let the singular value decomposition $\bGamma_\tau=\Ub_\tau \Db_\tau \Vb_\tau^\top$, where the singular vectors associated with the zero singular values are not included in the expression, and $\Ub_\tau^\top\Ub_\tau = \Vb_\tau^\top\Vb_\tau = \Ib_{r_\tau}$. We set
\begin{align}
	\Psi_\tau = \Vb_\tau \mbox{ and }\Phi_\tau = \Db_\tau^\top \Ub_\tau^\top. \label{eq:faciden}
\end{align} 
In practice, using singular value decomposition $\tG=\tilde\Ub_\tau \tilde\Db_\tau \tilde\Vb_\tau^\top$, the factors and loadings can be estimated similarly as \eqref{eq:faciden}:
\begin{align}\label{eq:facidenem}
	\begin{split}
\hat f_k^\tau(\bX_i)&=\tilde\sigma_k \tilde\Ub_{\tau,*k}^\top \bX_i,\\
\hat\Psi_\tau &= \tilde\Vb_{\tau}, 
	\end{split}
\end{align} 
where $\tilde\sigma_k$ is the $k$th largest singular value of $\tG$. 

The nuclear norm penalty in our loss function \eqref{peml} shrinks most singular values to 0, so typically $\tilde r_\tau = \max\{k:\tilde\sigma_k>0\}$ is small relatively to $n$, $p$ and $m$. Note that when $r_\tau$ is not small, the proposed method can still provably work, if the importance (measured by the singular values) of latter factors is small; see Remark \ref{rmk:notlowrank} for details. In practice, the number of factors $r_\tau$ is estimated by $\hat r_\tau = \max\{k:\tilde\sigma_k>\epsilon\}$ for, e.g. $\epsilon = 10^{-10}$, as adopted in this paper.

\subsection{Tuning}\label{sec:tune}

It is crucial to appropriately select $\lambda$ for the problem \eqref{peml}. We propose two ways to select $\lambda$. The first method is based on simulation. In particular, define the random variable 
\begin{align}
	\Lambda_\tau \defeq (nm)^{-1} \|\Xb^\top \widetilde\Wb_\tau\|, \label{pivotal}
\end{align}
where $(\widetilde W_\tau)_{ij} = \IF(U_{ij} \leq \tau)-\tau$, and $\{U_{ij}\}$ are i.i.d. uniform (0,1) random variables for $i=1,...,n$ and $j=1,...,m$, independent from $\bX_1,...,\bX_n$. The random variable $\Lambda_\tau$ is pivotal conditioning on design $\Xb$, as it does not depend on the unknown $\bGamma_\tau$. The formula in \eqref{pivotal} arises from the subgradient $\nabla \widehat Q_\tau(\bGamma_\tau)=(nm)^{-1}\Xb^\top \widetilde\Wb_\tau$. Set 
\begin{align}
	\lambda = \lambda_\tau = 2 \cdot \Lambda_\tau(1-\eta|\Xb), \label{qpivotal}
\end{align}
where $\Lambda_\tau(1-\eta|\Xb) \defeq (1-\eta)$-quantile of $\Lambda_\tau$ conditional on $\Xb$, for $0<\eta<1$ close to 0, for instance $\eta=0.1$. The constant 2 in \eqref{qpivotal} is mainly for the convenience of theoretical development, and it can be replaced by any constant greater than 1. 
In practice, when $n$ is large enough, the constant has little effect on the estimated number of factors, which is shown in our empirical study; see the left panel of Figure \ref{pcts}.

For estimating the number of factors, simulation study in Section \ref{sec.si} suggests that the tuning parameter given by \eqref{qpivotal} sometimes leads to too small $r_\tau$. Alternatively, we propose to choose $\lambda$ by minimizing the penalized testing error:
\begin{align}
	\min_{\lambda}\E\big[\widehat Q_\tau(\tG^\lambda)\big] + \lambda \|\tG^\lambda\|_*, \label{eq:cv}
\end{align}
where we include the superscript $\lambda$ to $\tG$ to emphasize its dependence on the tuning parameter $\lambda$, and $\hat Q(\cdot)$ is defined in \eqref{eq:Qhat}. The expectation in \eqref{eq:cv} can be approximated by an empirical average with a testing data set, in which the estimator $\tG^\lambda$ is computed with a training data set that has no overlap with the testing data, and the minimizer can be found by grid search. To select the grid points, the value in \eqref{qpivotal} can be used as the upper bound of the grid points. The performance of \eqref{eq:cv} will be evaluated by simulation in Section \ref{sec:facno}.

\section{Theory}\label{sec:ora}

Section \ref{sec:ora2} develops high probability error bound for the approximate optimizer $\tG$. The bound can be applied with Theorem \ref{thm.qr} to derive a bound for the estimator $\Gamat$ proposed in Section \ref{sec:coef}. Section \ref{sec:rfacloa} characterizes the risk of the factors and loadings estimator.

\subsection{Stochastic Risk of the Approximate Estimator \texorpdfstring{$\tG$}{$\Gamma_{\tau,t}$}}\label{sec:ora2}

The following assumptions are introduced.
\begin{itemize}
	\item[(A1)]\label{A1} (Sampling setting) Samples $(\bX_1,\bY_1),...,(\bX_n,\bY_n)$ are i.i.d. copies of $(\bX,\bY)$ random vectors in $\R^{p+m}$ with $p,m\geq 3$. $F_{Y|\bX}^{-1}(\tau|\bx) = \bx^\top\bGamma_{\tau,*j}$, where $Y_j$ is the $j$th variable in $\bY$. Moreover, for each $i$, $\bW_{\tau,i\ast}\defeq \{(\IF(Y_{ij}-\bX_i^\top \bGamma_{\tau,\ast j}\leq 0)-\tau)\}_{1\leq j\leq m}$ are i.i.d.  
	\item[(A2)]\label{A2} (Covariates) $\bX$ is centered with covariance matrix $\bSigma_X$. Assume the density function of $\bX$ exists. Suppose $0<\sigma_{\min}(\bSigma_X)<\sigma_{\max}(\bSigma_X)<\infty$, and there exist constants $B_p\geq 1$, $c_1,c_2>0$ such that $\|\bX_i\|$ and the sample covariance matrix $\hSigma_{\bX} = \frac{1}{n}\Xb^\top\Xb$ satisfies
	\begin{align}
		\P\big\{\sigma_{\min}(\hSigma_X)\geq c_1 \sigma_{\min}(\bSigma_X), \sigma_{\max}(\hSigma_X)\leq c_2 \sigma_{\max}(\bSigma_X), \|\bX_i\|\leq B_p \big\} \geq 1-\gamma_n, \label{cond.cov}
	\end{align}
	for a sequence $\gamma_n\to 0$.
	\item[(A3)]\label{A3} (Conditional densities) There exist constants $\bar f>0$, $\fline >0$ and $\bar f'<\infty$ such that 
	\begin{align*}
	&\max_{j\leq m}\sup_{\sbx,y}\big|f_{Y_{j}|\bX}(y|\bx)\big|\leq \bar f,\  \max_{j\leq m}\sup_{\sbx,y}\bigg|\frac{\partial}{\partial y_j}f_{Y_{j}|\bX}(y|\bx)\bigg| \leq \bar f',\ 
	\min_{j \leq m}\inf_{\sbx} f_{Y_{j}|\bX}(\bx^\top\bGamma_{\tau,*j}|\bx) \geq \fline, 
\end{align*}
where $f_{Y_{j}|\bX}$ is the conditional density function of $Y_{j}$ on $\bX$.
\end{itemize}

The i.i.d. condition in Assumption \hyperref[A1]{(A1)} allows to bound some tail probability with sharp random matrix theory (see Remark \ref{rem:noniid}). This may be replaced by $m$-dependent or weak dependent conditions, but the theory will be more complicated, which is left for future research. In Assumption \hyperref[A2]{(A2)}, $\bX$ is centered. $B_p$ can be assumed bounded by a constant (for example, p.2 of \cite{MP13} and Theorem 1 of \cite{AKLMY16}), but generally $B_p\asymp \sqrt{p}$ if each component of $\bX$ is bounded almost surely. Eigenvalue bounds in \eqref{cond.cov} hold when the components in $\bX$ have light tail; see, for example, \cite{V12}. \hyperref[A3]{(A3)} is standard in quantile regression literature \citep{BC:2011}. Note that $\fline$ decreases when $\tau$ approaches 0 or 1.

The next lemma gives the bound for $n^{-1}\|\Xb^\top \Wb_\tau\|$, where $\Wb_\tau = \{W_{\tau,ij}\}_{ij}= \{(\IF(Y_{ij}-\bX_i^\top \bGamma_{\tau,\ast j}\leq 0)-\tau)\}_{1\leq i\leq n,1\leq j\leq m}$ is an $n\times m$ matrix. The detailed proof can be found in the supplementary material. 

\begin{lemma}\label{lem.rate}
Assume \hyperref[A1]{(A1)} and \hyperref[A2]{(A2)} hold.
\begin{enumerate}
	\item For arbitrary $u>1$, with probability greater than $1-3e^{-(u-1)(p+m)\log 8}-\gamma_n$,
	\begin{align}
	\frac{1}{n}\|\Xb^\top \Wb_\tau\| \leq C^* \sqrt{u \sigma_{\max}(\bSigma_X)K(\tau)}\sqrt{\frac{p+m}{n}}, \label{rate.2norm}
	\end{align}
	where $C^* = 4\sqrt{\frac{c_2}{C'}\log 8} \vee 1$, $C'$ and $c_2$ are absolute constants given by Lemma \ref{lem.hoef} in the supplementary material and Assumption \hyperref[A2]{(A2)}, and
	\begin{align}\label{eq:Ktau}
		K(\tau) \defeq \left\{\begin{array}{ll}
					0, &\ \tau=0,1;\\
					\frac{2\tau-1}{2\{\log \tau-\log (1-\tau)\}}, &\ \tau \in (0,1)\backslash\{1/2\};\\
					1/4, &\ \tau=1/2,
					\end{array}
				\right.
	\end{align}
	\item With probability 1, for any $0<\eta<1$,
	\begin{align}
		\Lambda_\tau(1-\eta|\Xb) &\leq \bar\lambda := \frac{C^*}{m} \sqrt{\Big(1-\frac{\eta-\gamma_n}{3 (p+m)\log 8}\Big)\sigma_{\max}(\bSigma_X)K(\tau)}\sqrt{\frac{p+m}{n}} \label{eq:bddquant}
	\end{align}
	where $\gamma_n\to 0$ is defined in Assumption \hyperref[A2]{(A2)}.
\end{enumerate} 
\end{lemma}
See Section \ref{sec:proof_lem_rate} for a proof of Lemma \ref{lem.rate}. The constant $K(\tau)$ is the sub-Gaussian norm of the binary random variable $W_{\tau,ij}$ \citep[Theorem 3.1]{BM13}. Particularly, $K(\tau)$ is a concave function of $\tau\in[0,1]$ and is symmetric about $\tau=1/2$. The maximum of $K(\tau)$ is 1/4 at $\tau=1/2$. In addition, $K(\tau)\geq \tau(1-\tau)$ [Eqn. (9) on p.36 of \cite{BM13}]. See Lemma 2.1 of \cite{BM13} for more on $K(\tau)$.

The next result presents a non-asymptotic risk bound of the approximate optimizer $\tG$, when the optimization $\delta$ is well controlled. The key ingredient in its proof is the convexity arguments and a new tail probability bound for the empirical process $\GG_n\{\hQ_\tau(\bGamma_\tau+\bDelta)-\hQ_\tau(\bGamma_\tau)\}$, which builds on a sharp bound for the spectral norm of a partial sum of random matrices \citep{MP13,T:11}. Define
	\begin{align}
		\nu_\tau(\delta) &\defeq \frac{3}{8}\frac{\fline}{\bar f'} \inf_{\bDelta \in \cK(\bGamma_\tau;\delta)\atop\bDelta \neq 0} \frac{\big(\sum_{j=1}^m \E[|\bX_i^\top\bDelta_{*j}|^2]\big)^{3/2}}{\sum_{j=1}^m \E[|\bX_i^\top\bDelta_{*j}|^3]},\label{eq:nut}
\end{align}
where $\cK(\bGamma_\tau;\delta)$ is a "star-shaped" set of matrices defined in \eqref{eq:conea} [see more details there]. Note that $\cK(\bGamma_\tau;\delta_1)\subset\cK(\bGamma_\tau;\delta_2)$ for all $0\leq\delta_1\leq\delta_2$.
	
\begin{theorem}\label{thm:rec}
	Assume that \hyperref[A1]{(A1)}-\hyperref[A3]{(A3)} hold, and $\lambda$ satisfies 
	\begin{align}
		2 \Lambda_\tau(1-\eta|\Xb) \leq \lambda \leq 2 \bar\lambda, \label{lambda}
	\end{align}
	where $\bar\lambda$ is defined in \eqref{eq:bddquant}. Let $\delta \leq C \lambda m^{1/2}n^{-1/2}$ for some constant $C>0$, where $\delta$ is the upper bound of the optimization error in \eqref{eq:appest}. For some $u>1$, assume that $r=\rank(\bGamma_\tau)$ satisfies
	{\small
	\begin{align}
		u \epsilon_{n,\tau,r} \defeq u\frac{(384\sqrt{2}+96C^*)}{\fline\wedge 1}\sqrt{\frac{\sigma_{\max}(\bSigma_X) \vee 1}{\sigma_{\min}(\bSigma_X)\wedge 1}}  \sqrt{\frac{r(m+p \vee B_p)(\log p+\log m)}{m n}} < \nu_\tau(2C m^{1/2} n^{-1/2}),\label{eq:gr}
	\end{align}}
	where $C^*$ is a large constant defined in \eqref{rate.2norm}. Then, 
	\begin{align}
		\|\tG-\bGamma_\tau\|_{L_2(P_X)} &\leq u \epsilon_{n,\tau,r}\label{opred} 
	\end{align}
	with probability at least $1-\eta-\gamma_n-16(pm)^{1-u^2}-3\exp\{-(p+m)\log 8\}$, where $\|\cdot\|_{L_2(P_X)}^2 \defeq m^{-1}\E_{P_X}\|\cdot^\top \bX_i\|_2^2$ is the prediction error; in addition,
	\begin{align}
		\|\tG-\bGamma_\tau\|_{\rm F} \leq u\bigg(\frac{m}{\sigma_{\min}(\bSigma_X)}\bigg)^{1/2} \epsilon_{n,\tau,r}.\label{eq:L2bdd}
	\end{align}
\end{theorem}

See Section \ref{sec:proof_lem_ora} for a proof of Theorem \ref{thm:rec}. A sufficient condition for the bounds \eqref{opred} and \eqref{eq:L2bdd} of Theorem \ref{thm:rec} is $\delta \leq C \lambda m^{1/2}n^{-1/2}$.  We note that the conclusion of Theorem \ref{thm:rec} holds regardless of the algorithm that computes $\tG$, as long as the optimization error $\delta$ satisfies the bound. Corollary \ref{co:rec} provides an application of Theorem \ref{thm:rec} on the estimator $\bGamma_{\tau,T}$ proposed in Section \ref{sec:coef}.

 The error bound $\epsilon_{n,\tau,r}$ defined in \eqref{eq:gr} can be regarded as the stochastic error, which is not related to the optimization error $\delta$. If $p$ and $m$ are fixed with respect to $n$ (low dimensional setting), $\epsilon_{n,\tau,r} = O(n^{-1/2})$. The quantity $r(p+m)$ can be viewed as the actual number of unknown parameters, which has to be much smaller than $n$. The covariates can influence $\epsilon_{n,\tau,r}$ through the condition number $\sigma_{\max}(\bSigma_X)/\sigma_{\min}(\bSigma_X)$ of the covariance matrix $\bSigma_X$ and $B_p$. 
The estimation at $\tau$ close to 0 or 1 is challenging as $\fline$ in the denominator decreases when $\tau$ approaches 0 or 1. 

The rate in \eqref{eq:L2bdd} achieves the same (up to a constant) convergence rate in $p,m$ and $n$ to the multivariate regression for mean \citep{NW:11,KLT:2011}\footnote{Their regression problem (in mean) is analogous to our multivariate quantile regression setting by adjusting their $n$ to $mn$. See Example 1 on page 1075 of \cite{NW:11}. Note also that $\|\bDelta\|_{\rm F}\leq (m/\sigma_{\min}(\bSigma_X))^{1/2} \|\bDelta\|_{L_2(P_X)}$ for a well-behaved $\bDelta$ in our setting.}, which is shown to be unimprovable up to a logarithmic factor in the minimax sense \citep{KLT:2011}. Hence, we conjecture that $\epsilon_{n,\tau,r}$ is \emph{unimprovable} up to a logarithmic factor under the multivariate quantile regression setting, as the convergence rate of quantile regression is typically the same (except for some constants) as the mean regression. If it were true, then $\tG$ is as good as $\Gamah$, because the rate of both are unimprovable. 

\begin{remark}[Comment on the growth condition \eqref{eq:gr}]\label{rmk:nu}
In Theorem \ref{thm:rec}, the growth condition \eqref{eq:gr} guarantees the difference of population quantile loss $Q_{\tau}(\Sb+\bDelta)-Q_{\tau}(\Sb)$ is minorized by a quadratic function for $\bDelta$ inside a well-behaved set. Moreover, it can be easily seen that $\nu_\tau(\delta_1)<\nu_\tau(\delta_2)$ for all $\delta_1>\delta_2$ from the definition of $\nu_\tau(\delta)$ in \eqref{eq:nut}. Section \ref{sec:nu} discusses the details of the growth condition \eqref{eq:gr}. 
\end{remark}

\bigskip
The selection of smoothing parameter $\kappa>0$ has a significant impact on the algorithm in Section \ref{sec:coef}. Indeed, Corollary \ref{co:rec} below proves that estimator $\Gamat$ in Section \ref{sec:coef} achieves the bound in Theorem \ref{thm:rec}, provided that the smoothing parameter $\kappa$ and the number of iterations $T$ satisfy certain conditions.

\begin{corollary}\label{co:rec}
	Assume the conditions of Theorem \ref{thm:rec} and $C=1$ in \eqref{eq:gr}. Suppose in addition that the true coefficient matrix $\bGamma_\tau$ satisfies $\|\bGamma_\tau\|_{\rm F}^2 \leq C_\tau pm$ for some constant $C_\tau>0$, $B_p r(p+m)(\log p+\log m)=o(n)$. Let the initial estimator be $\bGamma_{\tau,t=0}=0$ in the algorithm in Section \ref{sec:coef}, and that
	\begin{align}
		\kappa &\leq \frac{\lambda}{3 m^{1/2} n^{3/2} \{\tau \vee (1-\tau)\}^2}, \quad\mbox{(smoothing parameter)} \label{eq:kas}\\
		T = T(\kappa)&\geq \frac{4 C_\tau^{1/2} p^{1/2}m^{1/2} \|\Xb\|}{\lambda^{1/2}\kappa^{1/2}m^{5/4}n^{3/2}}-1,
		 \quad\mbox{(number of iterations)}\label{eq:Ts}
	\end{align}
	then for $u>1$, \eqref{opred} and \eqref{eq:L2bdd} hold with $\tG=\Gamat$ with probability at least $1-2(\eta+\gamma_n+16(pm)^{1-u^2}+3\exp\{-(p+m)\log 8\})$, where the last bound in \eqref{eq:Ts} uses \eqref{eq:kas}.	
\end{corollary}

See Section \ref{sec:pfcor} for a proof of Corollary \ref{co:rec}. The key component of the proof is to verify that the optimization error of $\Gamat$ is less than $\lambda m^{1/2}n^{-1/2}$. The constants in both \eqref{eq:kas} and \eqref{eq:Ts} can be improved, and we adopt the current form for transparent exposition. In practice, the $\kappa$ based on \eqref{eq:kas} may be too small, and larger $\kappa$ usually performs better, as observed in the simulation analysis in Section \ref{sec:facno}. 
 
\begin{remark}[Many factors]\label{rmk:notlowrank}
When $\bGamma_\tau$ is not exactly small in rank, i.e. when $r_\tau$ is not a fixed number and possibly $r_\tau=\min\{p,m\}$. In this case, we may characterize the recovery performance of $\tG$ using the mathod of \cite{NRWY:12}. See Section \ref{sec:notlowrank} for more details. 
\end{remark}

\subsection{Realistic Bounds for Factors and Loadings}\label{sec:rfacloa}

The estimation error for the estimators for factors and loadings, defined in \eqref{eq:facidenem}, will be stated in terms of the Frobenius error $\|\tG-\bGamma_\tau\|_{\rm F}$. Theorem \ref{thm:rec} can be applied to find the explicit rate for the factors and loadings. 

First we observe that by Mirsky's theorem (see, for example, Theorem 4.11 on page 204 of \cite{SS90}), the singular values can be consistently estimated.
\begin{lemma}\label{lem:mirsky}
Let $\tG$ satisfy \eqref{eq:appest}, then
\begin{align}
	\sum_{j=1}^{p \wedge m}\big\{\sigma_j(\tG)-\sigma_j(\bGamma_\tau)\big\}^2 \leq \|\tG-\bGamma_\tau\|_{\rm F}^2.
\end{align}
\end{lemma}

Next, the error bounds for the factors and loadings are presented.

\begin{theorem}\label{th:facloa}
If the nonzero singular values of matrix $\bGamma_\tau$ are distinct, then with the choice of $\hat\Psi_\tau$ and $\hat f_k^\tau(\bX_i)$ in \eqref{eq:facidenem}, 
\begin{align}
	1-|(\hat\Psi_\tau)_{*j}^\top (\Psi_\tau)_{*j}| \leq \frac{2(2\|\bGamma_\tau\|+\|\tG-\bGamma_\tau\|_{\rm F})\|\tG-\bGamma_\tau\|_{\rm F}}{\min\{\sigma_{j-1}^2(\bGamma_\tau)-\sigma_j^2(\bGamma_\tau),\sigma_{j}^2(\bGamma_\tau)-\sigma_{j+1}^2(\bGamma_\tau)\}}\label{eq:vbound}	
\end{align}
If, in addition, let the SVDs $\bGamma_\tau=\Ub_\tau \Db_\tau \Vb_\tau^\top$ and $\tG=\tilde\Ub_\tau \tilde\Db_\tau \tilde\Vb_\tau^\top$, suppose $(\tilde \Ub_\tau)_{*j}^\top (\Ub_\tau)_{*j}\geq 0$, then
\begin{align}
	&\big|\hat f_k^\tau(\bX_i)-f_k^\tau(\bX_i)\big|\notag\\
	&\leq \|\bX_i\|\bigg(\|\tG-\bGamma_\tau\|_{\rm F}+2\sigma_k(\bGamma_\tau)\sqrt{\frac{(2\|\bGamma_\tau\|+\|\tG-\bGamma_\tau\|_{\rm F})\|\tG-\bGamma_\tau\|_{\rm F}}{\min\{\sigma_{k-1}^2(\bGamma_\tau)-\sigma_k^2(\bGamma_\tau),\sigma_{k}^2(\bGamma_\tau)-\sigma_{k+1}^2(\bGamma_\tau)\}}}\bigg)\label{eq:bfac}
\end{align}
\end{theorem}
See Section \ref{sec:facloa} for a proof for Theorem \ref{th:facloa}. The proof is based on a new Davis-Kahan type inequality of \cite{YWS15}. The inequalities in Theorem \ref{thm:rec} can be applied to find the exact rate for the loadings and factors.

\begin{remark}[Repeated singular values]\label{rmk:angle}
	Theorem \ref{th:facloa} is under the condition that the singular values for $\bGamma_\tau$ are distinct. If there are repeated singular values, then the corresponding singular vectors are not uniquely defined, and we can only obtain a bound for the "canonical angle" (see, for example, \cite{YWS15}) of the subspaces generated by the singular vectors associated with the repeated singular values. 
\end{remark}
\section{Simulation}\label{sec.si}
The performance of the numerical procedure in Section \ref{sec.al} on the factor quantile models will be checked via Monte Carlo experiments in this section. Section \ref{sec:esterr} presents the results on the Frobenius error. Section \ref{sec:facno} presents the performance on estimating $r_\tau$ with i.i.d. data, while Section \ref{sec:facno_ts} focuses on time series data.

\subsection{Estimation Error}\label{sec:esterr}
Given two distinct matrices $\Sb_1,\Sb_2$ with \emph{nonnegative} entries, $\rank(\Sb_1)=r_1$ and $\rank(\Sb_2)=r_2$, the ouputs are simulated from the two-piece normal model \citep{W14}:
\begin{align}
	Y_{ij} = G_\sigma^{-1}(U_{ij})\bX_i^\top&\big((\Sb_1)_{*j}\IF\{U_{ij} \leq 0.5\}+(\Sb_2)_{*j} \IF\{U_{ij} > 0.5\}\big),\label{het.model} \\ 
	&\hspace{4cm} i=1,...,n=500;\ j=1,...,m=300,\notag
\end{align}
where $U_{ij}$ are i.i.d. $U(0,1)$ independent of $\bX_i$; $G_\sigma$ is the cdf of $\mathcal N(0,\sigma^2)$. $\bX_i \in \R^p$ follows a multivariate $U([0,1])$ distribution with covariance matrix $\bSigma$ in which $\bSigma_{ij} = 0.1*0.8^{|i-j|}$ for $j=1,...,p=300$. Simulation of $\bX_i$ follows by the method of \cite{F99}. The number of simulation repetitions is 500.

Because the elements in $\bX_i$ are non-negative, the conditional quantile function $q_j(\tau|\bx)$ of $Y_{ij}$ on $\bx$ for the distribution of $Y_{ij}$ is  
\begin{align}\label{eq:trueq}
q_j^l(\tau|\bx) = \bx^\top G_\sigma^{-1}(\tau) \big(\Sb_1\IF\{\tau \leq 0.5\}+\Sb_2\IF\{\tau > 0.5\}\big)\defeq \bx^\top\bGamma_{\tau,*j},
\end{align}
It follows that $\bGamma_\tau = G_\sigma^{-1}(\tau)\Sb_1$ for $\tau\leq 1/2$ and $\bGamma_\tau = G_\sigma^{-1}(\tau)\Sb_2$ for $\tau>1/2$. In particular, $\bGamma_{1/2}=0$ and the $1/2$-quantile of all responses are equal to 0, because $G_\sigma^{-1}(1/2)=0$ for all $\sigma>0$. Hence, $\tau=0.5$ is the least interesting case, so it is not investigated.

To select $\Sb_1$ and $\Sb_2$, we first fix $s_1=\rank(\Sb_1)=2$ and for $\Sb_2$: 
\begin{itemize}
	\item[I.] Symmetric model: $\Sb_2^{Sym}$ with $r_2^{Sym}=\rank(\Sb_2^{Sym}) = 2$;
	\item[II.] Asymmetrical model: $\Sb_2^{Asym}$ with $r_2^{Asym}=\rank(\Sb_2^{Asym}) = 6$.
\end{itemize} 
The entries of $\Sb_1$, $\Sb_2^{Sym}$ and $\Sb_2^{Asym}$ are selected randomly and fixed for all simulations, and their singular values are distinct. We shift the details on selecting these matrices to Section \ref{sec:s1s2}.  

The algorithm in Section \ref{sec:coef} (specifics in Algorithm \ref{alg_qr}) is applied with $\tau$=5\%, 10\%, 20\%, 80\%, 90\% and 95\% to compute the estimator $\tG^l$ for $\bGamma_\tau^l$, where $l\in\{Sym,Asym\}$. We set $\kappa=10^{-4}$ and stop the algorithm when the change in the loss function is less than $10^{-6}$. The tuning parameter $\lambda$ is selected by the simulation procedure \eqref{qpivotal}. We compare $\tG^l$ with an oracle estimator, which is computed under the knowledge of the true $r_\tau$. 
The performance of $\tG^l$ and the oracle estimator is measured by the Frobenius error to the true coefficient $\bGamma_\tau$. 

The results are reported in Table \ref{tab:fro}. The oracle estimator errors are generally smaller for all $\tau$, and their standard deviation is also lower. When the model variance is larger ($\sigma=1$), the estimation of $\tG^l$ has greater error. The error of $\tG^l$ varies with $\tau$: the error for $\tau=0.05$ or 0.95 is almost twice as large as those for $\tau=0.2$ and 0.8. For the two models, the errors of $\tG^l$ are similar when $\tau$ is less than 0.5. However, when $\tau$ is greater than 0.5, the errors of the asymmetric model is around $\sqrt{r_2^{Asym}/r_2^{Sym}}=\sqrt{6/2} \approx 1.732$ times of that of the symmetric model. The oracle estimator also shows a similar pattern. The outcomes here are consistent with our theoretical analysis in Theorem \ref{thm:rec}, which predicts that the models with a larger rank and with $\tau$ closer to 0 or 1 have greater estimation errors. The prediction errors have similar pattern as the Frobenius error and are omitted for brevity. 


%
\begin{table}[htb]
        \begin{center}
        \begin{tabular}{lrrrrrr}
\hline\hline
$\tau$&$0.05$ &$0.1$ &$0.2$ &$0.8$ &$0.9$ &$0.95$\\
\hline
&\multicolumn{6}{c}{\underline{$\sigma = 0.5$}}    \\[0.7ex]
Symmetric & 60.995 & 48.746 & 34.302 & 33.973 & 48.375 & 60.604 \\
 &{\sl\small(0.253)} & {\sl\small(0.227)} & {\sl\small(0.209)} & {\sl\small(0.202)} & {\sl\small(0.217)} & {\sl\small(0.247)} \\ 
Symmetric Or. & 57.261 & 44.926 & 30.006 & 29.853 & 44.735 & 57.007 \\
			& {\sl\small(0.191)} & {\sl\small(0.152)} & {\sl\small(0.116)} & {\sl\small(0.118)} & {\sl\small(0.152)} & {\sl\small(0.184)} \\
Asymmetric & 60.978 & 48.724 & 34.289 & 60.487 & 85.997 & 108.310 \\
& {\sl\small(0.263)} & {\sl\small(0.220)} & {\sl\small(0.207)} & {\sl\small(0.539)} & {\sl\small(0.567)} & {\sl\small(0.820)} \\
Asymmetric Or. & 57.239 & 44.911 & 30.002 & 54.922 & 80.583 & 102.663 \\
& {\sl\small(0.202)} & {\sl\small(0.164)} & {\sl\small(0.120)} & {\sl\small(0.744)} & {\sl\small(0.464)} & {\sl\small(0.572)} \\
\hline
&\multicolumn{6}{c}{\underline{$\sigma = 1$}}    \\[0.7ex]
Symmetric & 118.245 & 93.419 & 64.289 & 63.634 & 92.519 & 117.365 \\
& {\sl\small(0.570)} & {\sl\small(0.420)} & {\sl\small(0.387)} & {\sl\small(0.382)} & {\sl\small(0.372)} & {\sl\small(0.438)} \\
Symmetric Or.& 113.636 & 88.781 & 58.913 & 58.593 & 88.365 & 113.099 \\
			& {\sl\small(0.427)} & {\sl\small(0.338)} & {\sl\small(0.238)} & {\sl\small(0.221)} & {\sl\small(0.301)} & {\sl\small(0.378)} \\
Asymmetric & 118.259 & 93.434 & 64.291 & 120.338 & 170.904 & 217.185 \\
& {\sl\small(0.530)} & {\sl\small(0.412)} & {\sl\small(0.380)} & {\sl\small(1.151)} & {\sl\small(1.273)} & {\sl\small(1.547)} \\
Asymmetric Or.& 113.647 & 88.788 & 58.911 & 108.754 & 161.303 & 205.371 \\
& {\sl\small(0.387)} & {\sl\small(0.308)} & {\sl\small(0.224)} & {\sl\small(0.711)} & {\sl\small(0.929)} & {\sl\small(1.188)} \\
\hline\hline
        \end{tabular}
        \end{center}
		  \caption{Averaged Frobenius errors with standard deviations. "Or." denotes the oracle estimator, which is estimated under the knowledge of true rank. The numbers in parentheses are standard deviations of the errors. }\label{tab:fro}
        \end{table}

\subsection{Estimating the Number of Factors for i.i.d. Data}\label{sec:facno}

Using the same data generating process \eqref{het.model}, this section shows the performance of estimating the number of factors. We will focus on the asymmetric case, as it is more challenging. Estimation performance for two representative quantile levels $\tau=0.2$ and $0.8$ are shown. Note the the number of factors are $r_{\tau=0.2}=2$ and $r_{\tau=0.8}=6$.

The tuning parameter $\lambda$ is selected by minimizing the \emph{penalized testing error} in \eqref{eq:cv} through grid search. The value in \eqref{qpivotal} is used as the upper bound of the grid points. To compute the penalized testing error \eqref{eq:cv} at a given $\lambda$, we independently generate data 150 times and compute 150 $\tG$. For each $\tG$, a testing error is computed with a testing data set of size $3000$ independent of the training data. Finally, we take an average of the 150 testing errors. The algorithm in Section \ref{sec:coef} (Algorithm \ref{alg_qr}) is applied with $T=4000$ iterations. The estimated number of factor is the number of singular values greater than $10^{-10}$.

Figure \ref{facno_iid} shows the relative frequency of the estimated number of factors and the estimated penalized testing error. The penalized testing error shows a quadratic shape, and there exists a minimum of the penalized testing error as a function of $\lambda$ for both $\tau=0.2$ and 0.8. The $\lambda$ that reaches the minimum for the two quantiles are essentially the same because they are symmetric about 0.5. For $\tau=0.2$ with the number of factor $r_{\tau=0.2}=2$, the number of factors is correctly estimated over 50\% of the time, and in the worst case scenario, the number of factors is misestimated by one. For $\tau=0.8$ with $r_{\tau=0.8}=6$, which is greater than $r_{\tau=0.2}=2$ of $\tau=0.2$. In this case, to ensure accurate estimation, the optimization error $\delta>0$ has to be small. For this purpose, \eqref{qr.bound} in Theorem \ref{thm.qr} implies that one should select smaller $\kappa$ and greater $T$. Therefore, a slightly smaller $\kappa = 4\times 10^{-7}$ is selected for $\tau=0.8$, but $T=4000$ appears to still be sufficient here. The bottom left panel of Figure \ref{facno_iid} shows that the correct number of factors is identified almost 60\% of the time, and at worst it is misestimated by two. We note that the $\lambda$ tuned by \eqref{qpivotal} is 0.0189 which leads to a models with only one factor as observed by unpublished simulations.

\begin{figure}[!h]
	\centering
\includegraphics[width=7cm, height = 7cm]{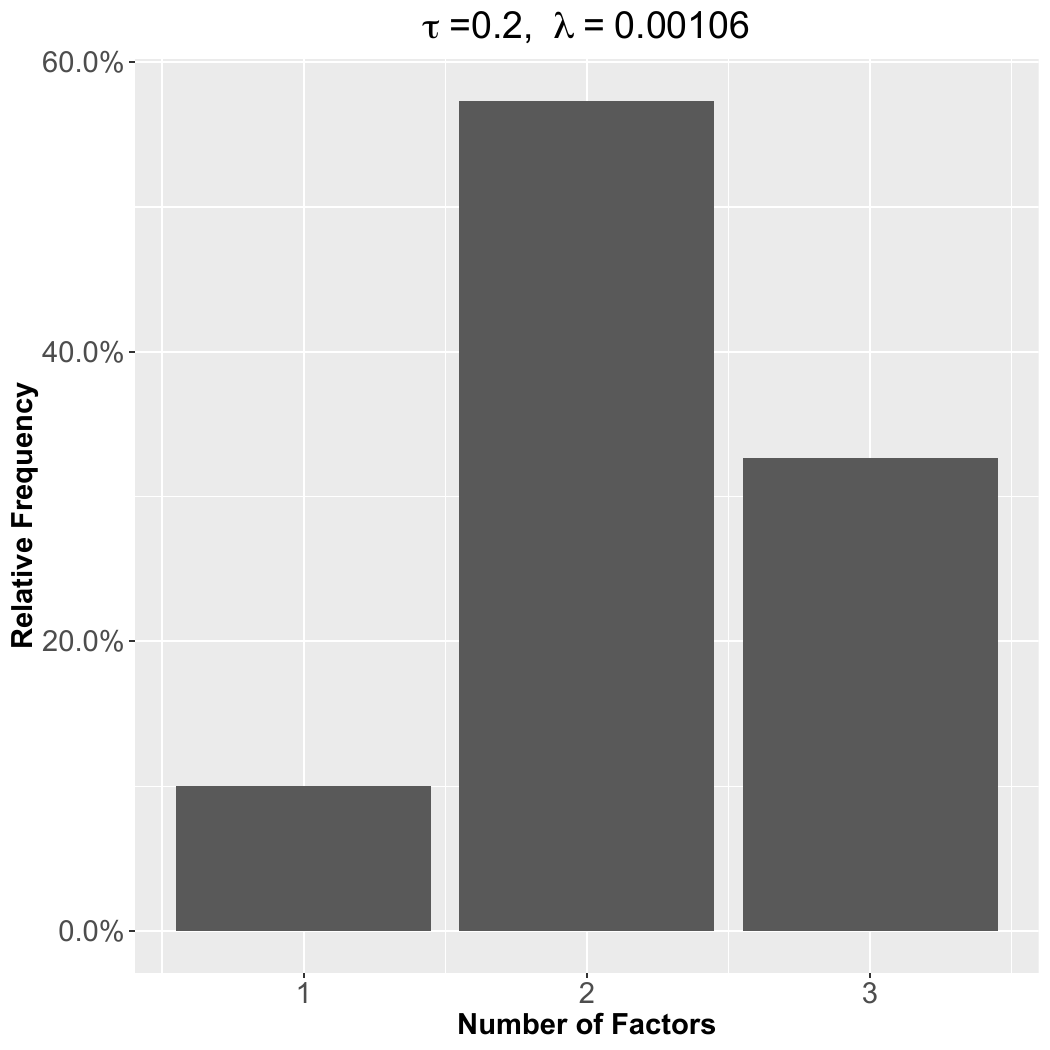}
\includegraphics[width=7cm, height = 7cm]{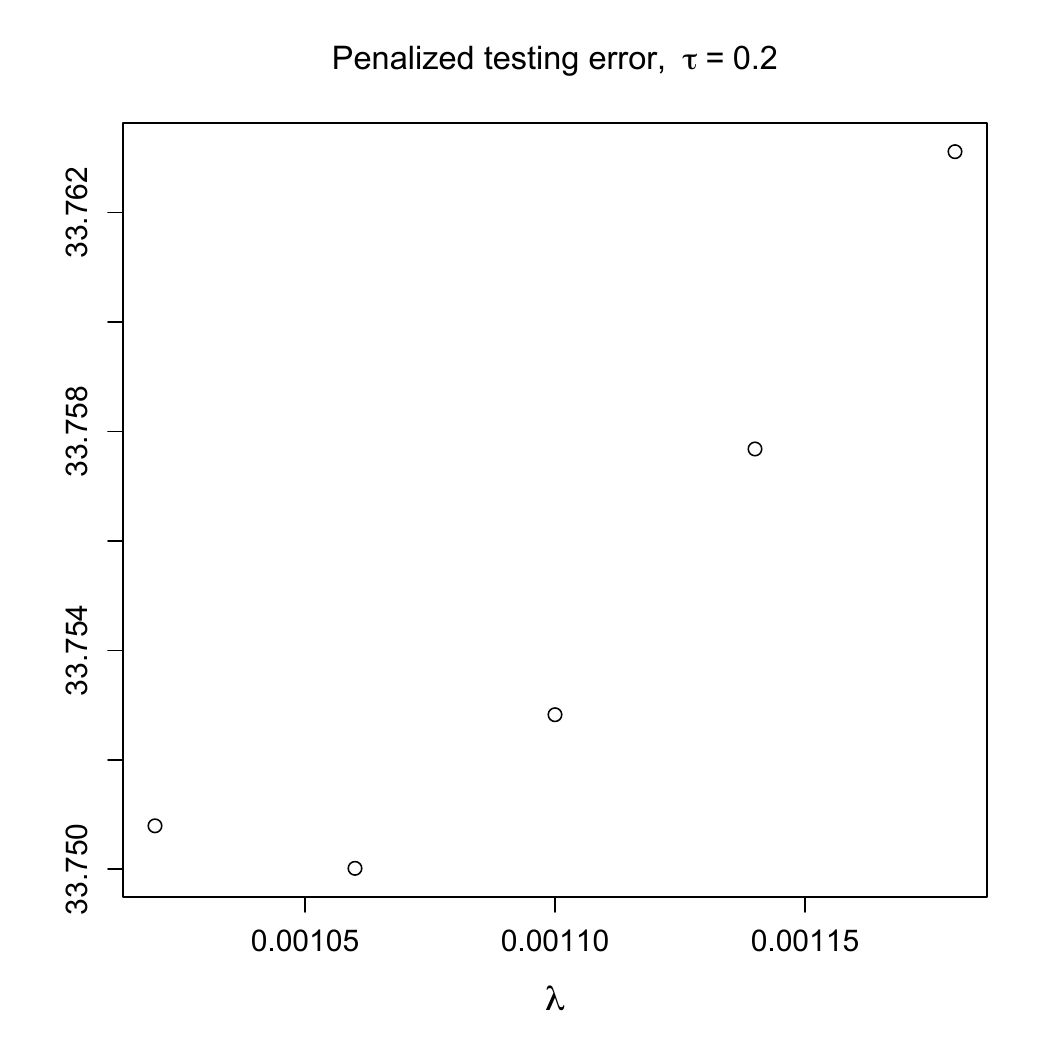}\\
\includegraphics[width=7cm, height = 7cm]{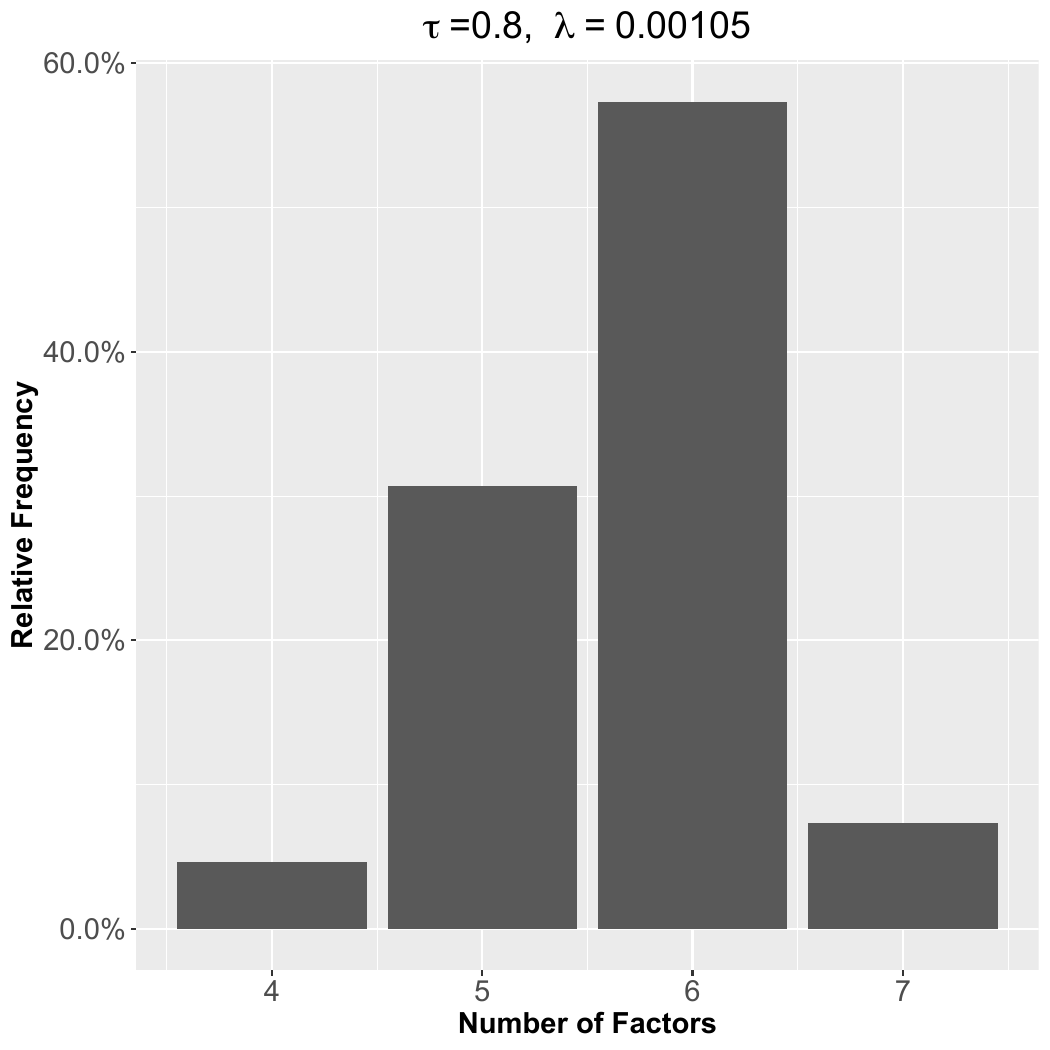}
\includegraphics[width=7cm, height = 7cm]{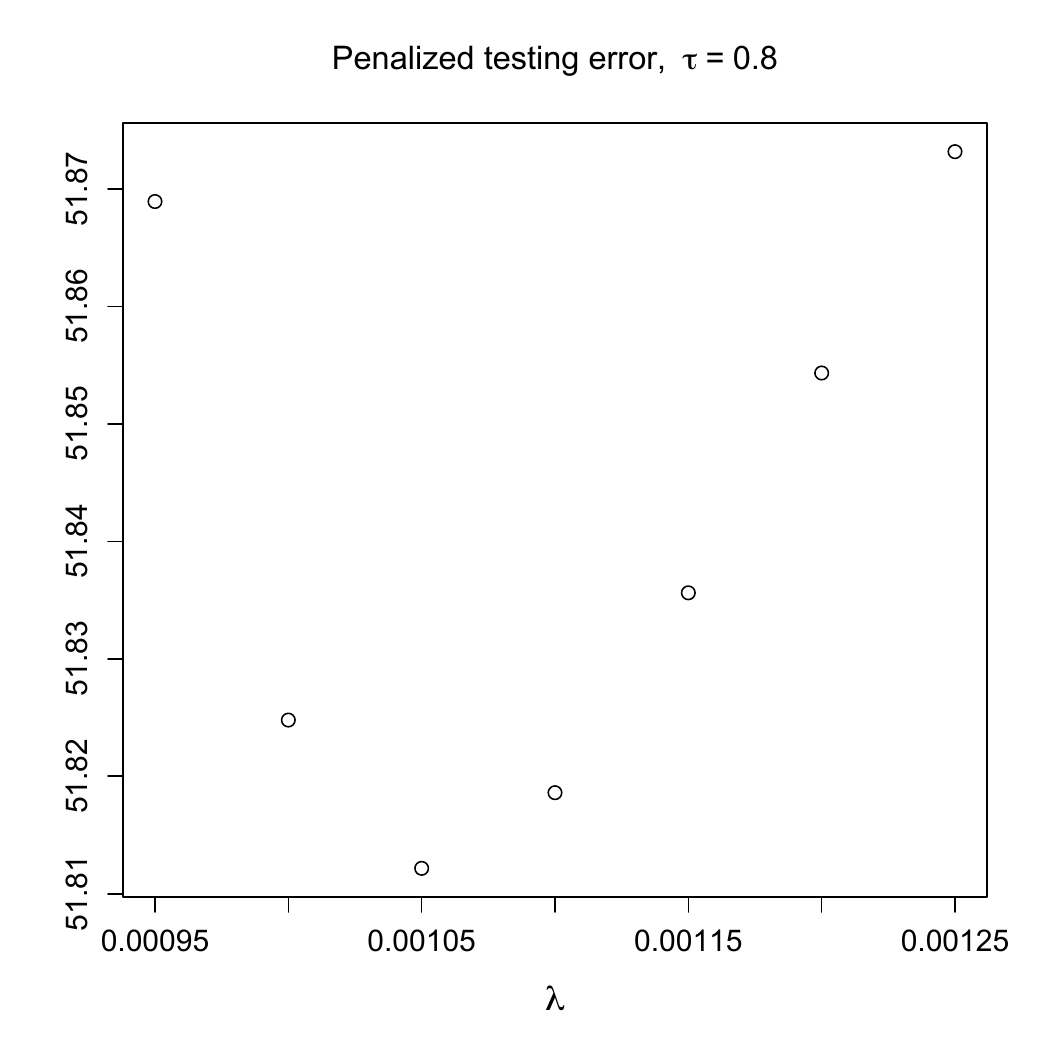}
\caption{Relative frequency of the estimated number of factors with 150 Monte Carlo repetitions, and the plot for the penalized testing error, $\tau=0.2$ and 0.8. Data are generated as \eqref{het.model} where $\bX_i$ is i.i.d. multivariate $U([0,1])$ with positive correlations. $r_{\tau=0.2}=2$ and $r_{\tau=0.8}=6$. $\kappa=6.66\times 10^{-7}$ for $\tau=0.2$ and $\kappa=4\times 10^{-7}$ for $\tau=0.8$. $T=4000$.}\label{facno_iid}
\end{figure}

\subsection{Estimating the Number of Factors for Dependent Data}\label{sec:facno_ts}

In this section, we consider a time dependent design
\begin{align}
	\bX_i = \Bb \bX_{i-1}+\bvep_i, \label{eq:ts_x}
\end{align}
where $\bvep_i$ is a multivariate Gaussian vector with mean zero and covariance $\Sigma_{\sbvep}$. The matrices $\Bb$ and $\Sigma_{\sbvep}$ are both $p\times p$, which are selected to imitate the temporal and cross-sectional dependent structure of the vector in \eqref{input.vari} in the empirical analysis Section \ref{sec.app}, where $p=\mbox{dim}(\bX)=460$ as in \eqref{input.vari}. $\Bb$ is set to be a sparse coefficient matrix, which is estimated by a vector autoregressive model with $\ell_1$ norm penalty \citep{DZZ16, NMB17}. Note that \eqref{eq:ts_x} is unstable, as four singular values of $\Bb$ are greater than one. The details for selecting $\Bb$ and $\Sigma_{\sbvep}$ are in Section S.5. 

The output $\bY$ is generated as \eqref{het.model} with $\bX_i$ in \eqref{eq:ts_x}. Here, we fix $\mbox{rank}(\Sb_2)=2$, while $\mbox{rank}(\Sb_1)=2$ or 3. Under this design, the number of quantile factors $r_\tau=\mbox{rank}(\Sb_1)$ for all $\tau\in(0,1)$ because the distribution of $\bX$ in \eqref{eq:ts_x} is symmetric about the origin. Similar to the situation in Section \ref{sec.app}, we set $n=2765$ and $m=230$. To selection $\lambda$, we adopt the simulation method in \eqref{qpivotal}, which yields $\lambda^*=4.38\times 10^{-3}$. Selecting $\lambda$ by minimizing penalized testing error in \eqref{eq:cv} is feasible, but it is more computationally demanding. As will be shown, $\lambda$ selected by \eqref{qpivotal} can yield accurate estimation in our setting. For all the numerical experiments in this section, $\kappa=1.5\times 10^{-4}$ and the number of iterations $T=3500$ for the algorithm.

Figure \ref{facno_ts_re} shows the performance of factor number estimation at $\tau=0.01$. In 150 Monte Carlo simulations, our method can identify the correct number of factors in over 60\% of the simulations. In the worst case scenario, the number of factors is misestimated by one. Note that the quantile level $\tau=0.01$ is more extreme than the that in Section \ref{sec:facno} and is more relevant for the empirical analysis in Section \ref{sec.app}. The results in Figure \ref{facno_ts_re} show that our method still has a good performance even for extreme quantile level and time series data. We remark that additional numerical experiments with a simpler AR(1) structure is in Section S.6.

As a remark, although numerical results suggest our method is useful even for complex time series data, our theory does not apply to this case. Extension of the theory in this direction is nontrivial, so it is left for future research.

\begin{figure}[!h]
	\centering
\includegraphics[width=7cm, height = 7cm]{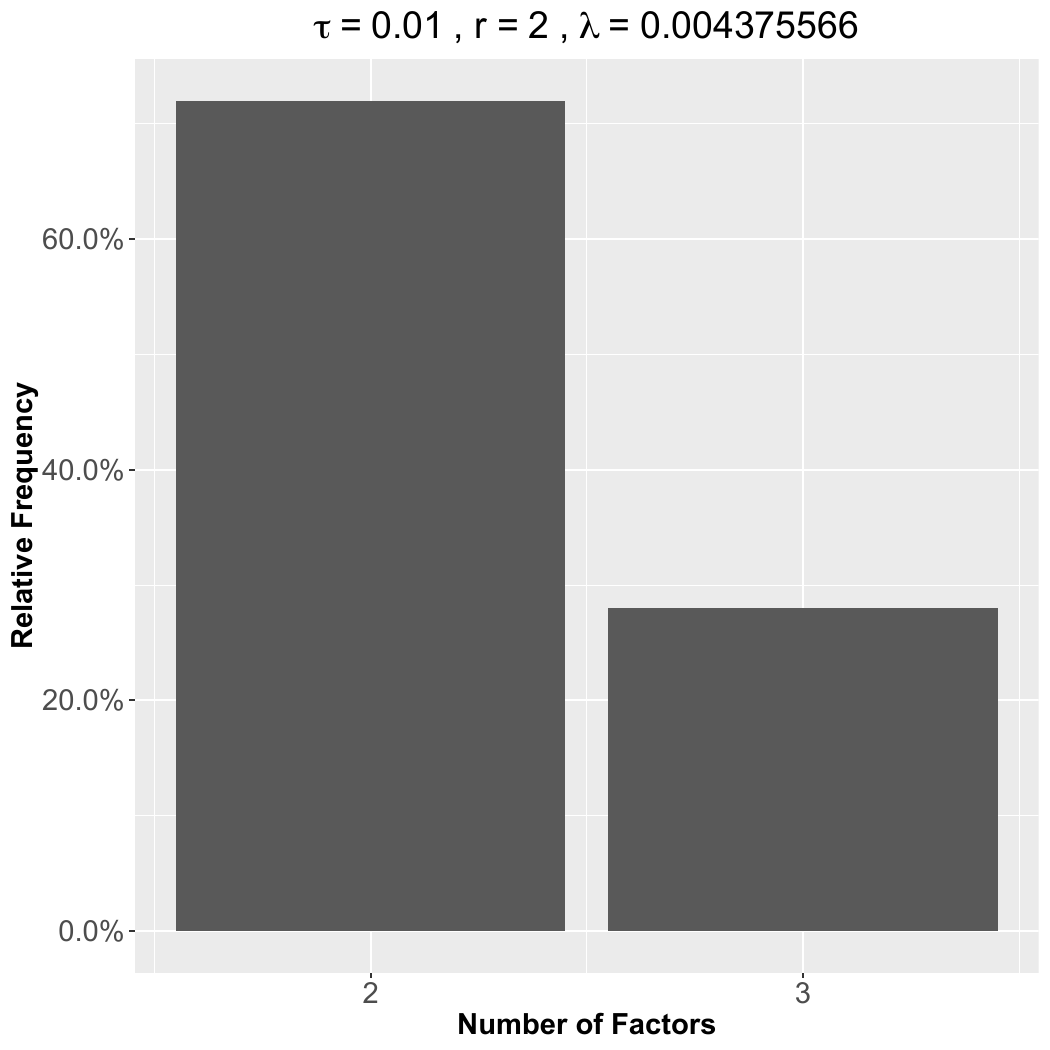}
\includegraphics[width=7cm, height = 7cm]{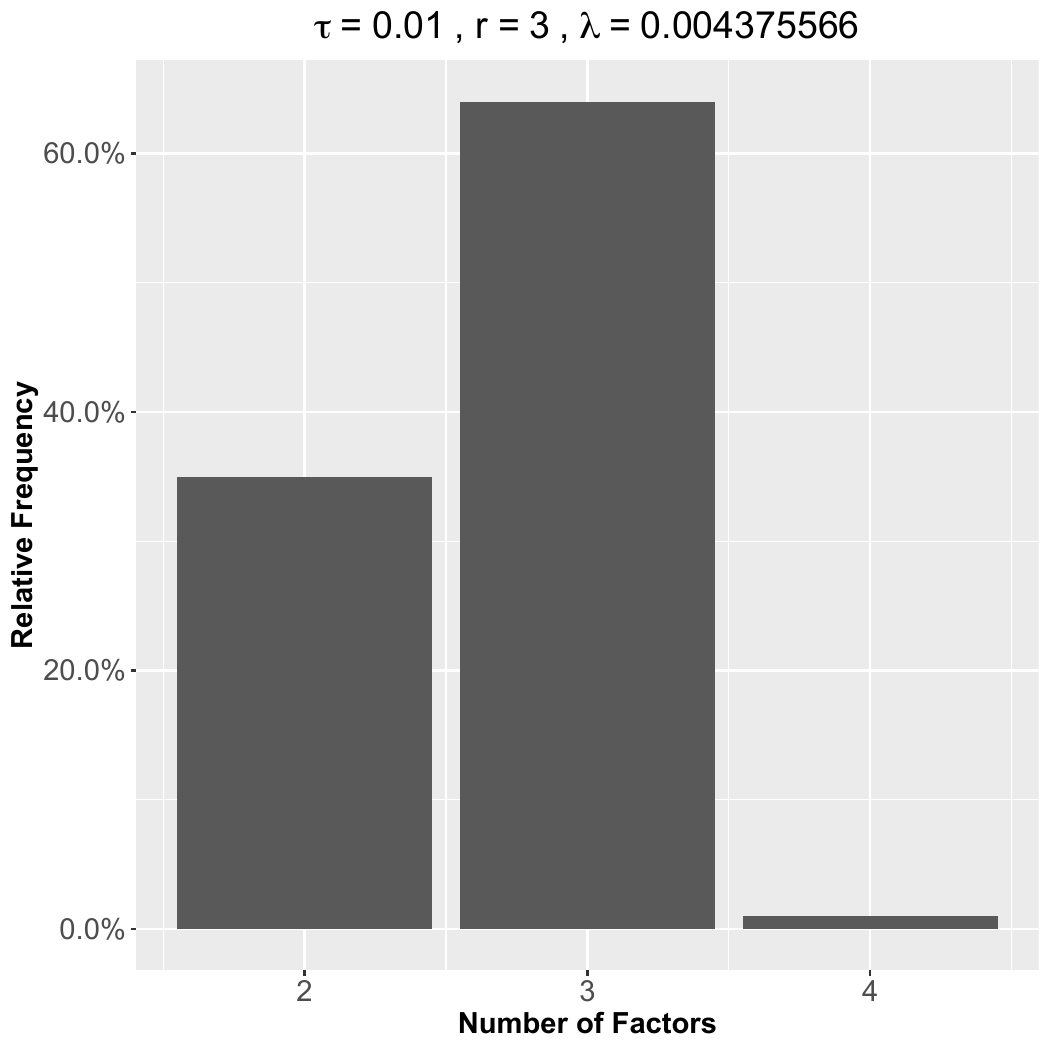}
\caption{The distribution of the estimated number of factors for $r=2$ (left) and $r=3$ (right) with $\tau=0.01$, with covariates simulated from \eqref{eq:ts_x}. The distribution is obtained from 150 Monte Carlo repetitions. $(\kappa,T)=(1.5\times 10^{-4},3500)$. The $\lambda$ is computed with the simulation method in \eqref{qpivotal}.}\label{facno_ts_re}
\end{figure}

\section{Empirical Analysis: Estimating the Systemic Risk}\label{sec.app}

In the aftermath of the financial crisis that started in 2007, governments and supervisory authorities have come to realize the need to quantify the impact of systemic risk on financial institutions. Numerous studies have been made in this direction; see \cite{BFLV12} or \cite{BO13} for a survey. Quantifying the impact of systemic risk is inherently a high dimensional problem, as hundreds or sometimes thousands of financial institutions have to be included in the model in order to make it realistic. Unfortunately, due to excessive computational cost, the models in the existing studies are low dimensional in nature. For example, \cite{AB16} estimate pairwise spillover effect between two institutions by conditioning on a set of variables; \cite{WKM10caviar2} estimate the impact of market shock by performing bivariate vector autoregression (VAR) between an institution and a pre-calculated proxy of the market shock. The proposed multitask quantile regression can fill this gap, as our method can easily scale up to hundreds of response variables and covariates. In addition, no proxy of the market shock needs to be pre-calculated, because the quantile factors obtained by our method summarize the market information that is most relevant to the downside risk. 


We analyze the same set of daily stock closing prices as \cite{WKM10caviar2}, with the same time frame from January 1, 2000 to August 6, 2010. The dataset is downloaded from Dr. Manganelli's personal website. See Table 1 of \cite{WKM10caviar2} for a detailed breakdown of the stocks by sector and country, as well as their averaged market value and leverage (the ratio of short and long term debt over common equity) over the data period. There are $m=230$ financial institutions. The daily log-returns of the stock closing prices are used, and this results in $n=2765$. 

Let $Y_{i,j}$ be the asset return for institution $j$ at time $i$, where $j=1,...,m$ and $i=1,...,n$. For $\tau\in(0,1)$, consider the quantile $q_{j}(\tau|\bX_i) = \bX_i^\top\bGamma_{\tau,\ast j}$ for $Y_j$, where 
\begin{align} 
	\bX_i = (|Y_{i-1,1}|,...,|Y_{i-1,m}|,Y_{i-1,1}^-,...,Y_{i-1,m}^-)^\top \in \R^{2m}, \label{input.vari}
\end{align}
and $Y^- \defeq \max\{-Y,0\}$. The covariate $\bX_i$ captures the fact that the positive or negative lag stock returns have different influence to the return today, which is motivated by \cite{EM04caviar}.

We estimate $\bGamma_\tau$ using the algorithm in Section \ref{sec:coef} (equivalently, Algorithm \ref{alg_qr}) with two quantile levels $\tau=0.01$ and $0.99$. The algorithm is performed with $\kappa=1.5\times 10^{-4}$ and we stop the algorithm when the change in the loss function is less than $10^{-6}$. The factors and loadings are estimated as \eqref{eq:facidenem} in Section \ref{sec.ex}. The tuning procedure in \eqref{qpivotal} yields $\lambda^*=8.53\times 10^{-3}$ for $\tau=0.01$. 
Left panel in Figure \ref{pcts} shows the estimated singular values for $\lambda\leq \lambda^*$. Even when $\lambda$ is smaller than $\lambda^*$ by ten folds, which corresponds to the case of increasing $n$ or $m$ by ten folds, the estimated number of factor for $\tau=0.01$ is still one, which shows the robustness of the estimated $r_\tau$. This result is similar for $\tau=0.99$. This suggests that the number of factor is one for both $\tau=0.01$ and 0.99. For later discussion, we set $\lambda^*$ for both $\tau=0.01$ and $\tau=0.99$ by symmetry.

Figure \ref{pcts} presents the estimated first factors at $\tau=0.01$ and 0.99. 
Both first factors $f_1^{0.01}(X_i)$ and $f_1^{0.99}(X_i)$ are volatile and moving away from 0 at the end of 2008 and in the first quarter of 2009, which corresponds to the periods of financial crisis.  

\begin{figure}[h!]
	\centering
  \includegraphics[width=6.2cm, height = 6cm]{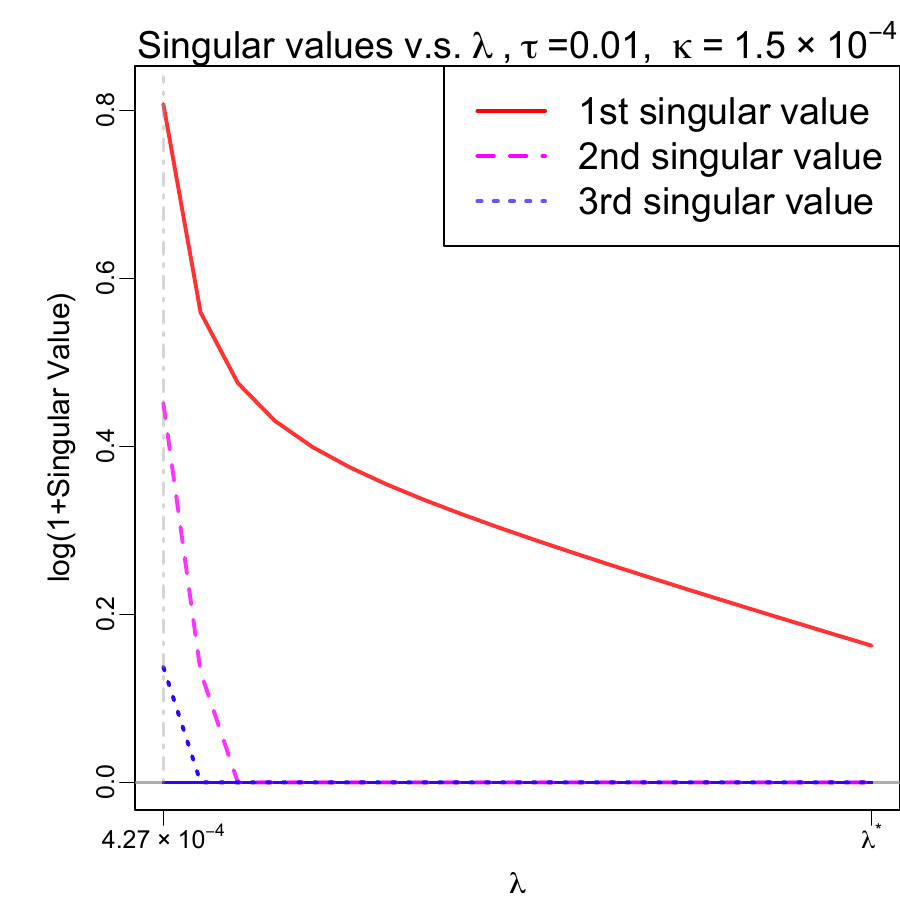}
 \includegraphics[width=6.2cm, height = 6cm]{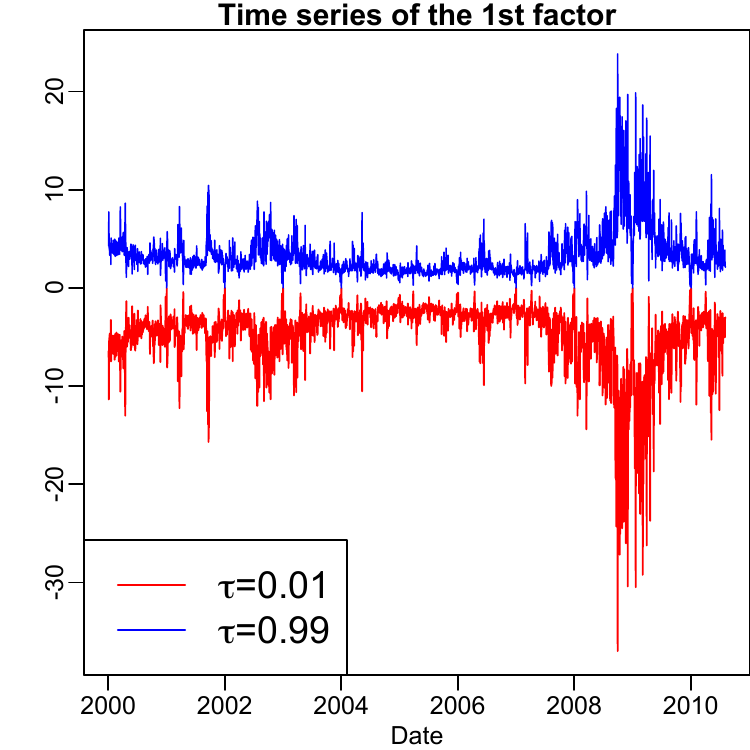}
 \caption{Left panel plots singular values versus $\lambda$. The right end point of x-axis is $\lambda^*=8.53\times 10^{-3}$ selected by the simulation method \eqref{qpivotal}. Right panel: time series plot of the first factor of $\tau=0.01$ (red) and $\tau=0.99$ (blue).}\label{pcts}
\end{figure}

The left panel of Figure \ref{T2T_fin} is the "tail to tail" plot with $\tau=0.01$ and 0.99, on which each point is a pair of loadings $((\hat\Psi_{0.01})_{1j},(\hat\Psi_{0.99})_{1j})$ defined in \eqref{eq:facidenem} for $j$th financial institution, $j=1,...,230$. The values $((\hat\Psi_{0.01})_{1j},(\hat\Psi_{0.99})_{1j})$ are all positive. The fact that they distribute around the 45 degree line suggests that the log-returns of these stocks are roughly equally associated to the two tail quantile factors, but the magnitude of their association to the factors varies dramatically. The points become more disperse and deviate from the 45 degree line in the northeast corner. 

The right panel of Figure \ref{T2T_fin} plots the institutions based on their averaged market value (x-axis) and leverage (y-axis), and the color represents the magnitude of the $\tau=0.01$ factor loading of the corresponding financial institution. It shows that financial institutions with large market value and high leverage tend to have high loadings of the first left tail factor $f_1^{0.01}(\bX)$, as most red and yellow color points are concentrating in the northeastern part of the figure. This shows that they are more vulnerable to the market shock, and this is in line with the conclusion of \cite{WKM10caviar2}. Interestingly, the institutions that are more vulnerable to market shock seem to form clusters. It is an interesting future research to study the geographical and financial properties of the financial institutions in the same cluster. 

\begin{figure}[h!]
	\centering
 \includegraphics[width=6cm, height = 7cm]{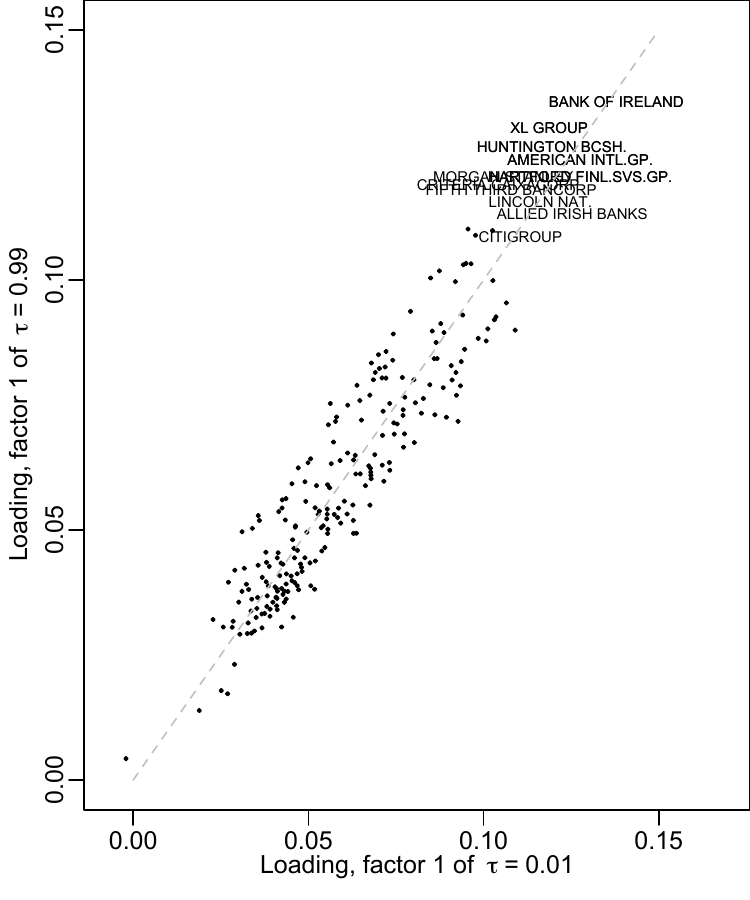}
 \includegraphics[width=9.5cm, height = 7cm]{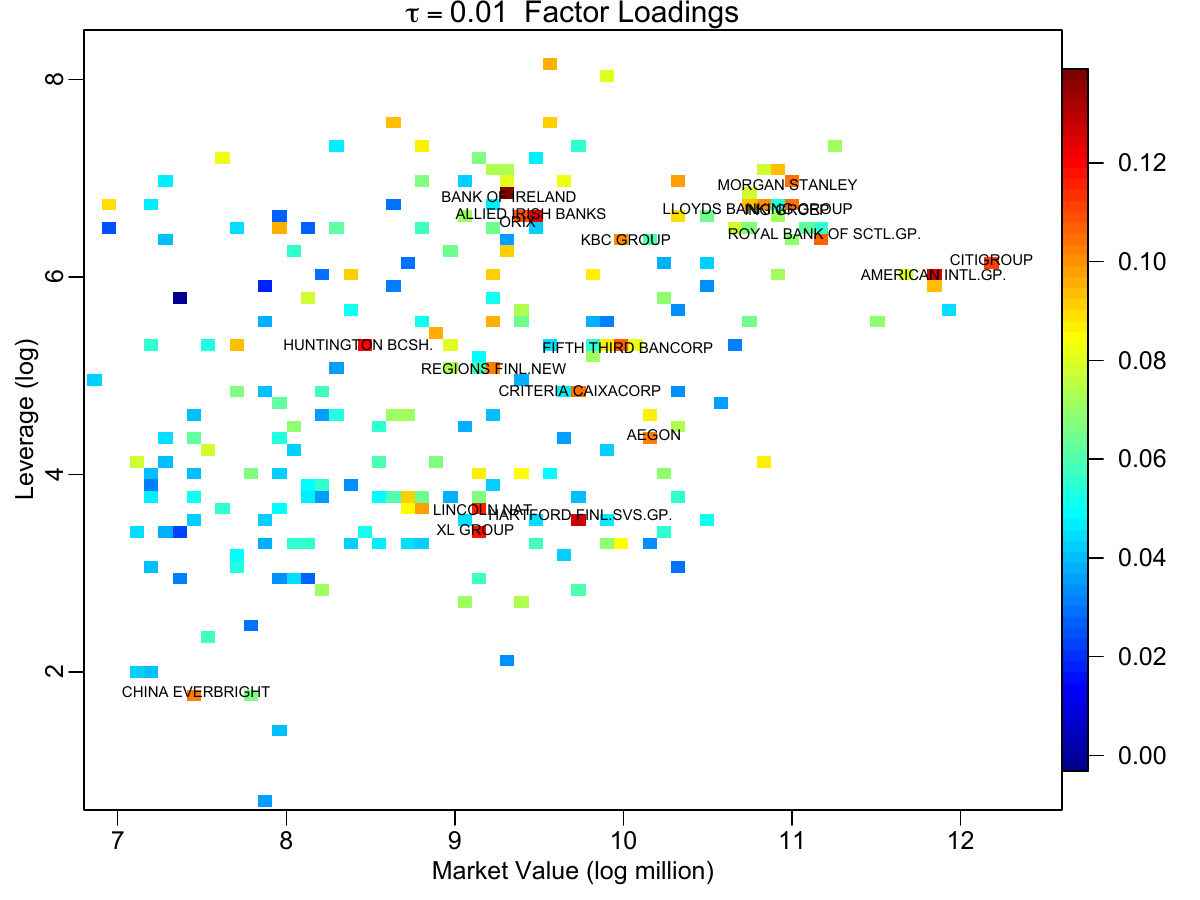}
 \caption{Left panel: tail to tail plot. Each point is a pair $((\hat\Psi_{0.01})_{1j},(\hat\Psi_{0.99})_{1j})$ for stocks $j=1,...,230$; Right panel: the plot of firms based on their averaged market value and leverage over the data period. The color scale corresponds to the magnitude of their $\tau=0.01$ factor loading, and the firms with loading greater than 0.1 are labeled by name.}\label{T2T_fin}
\end{figure}

\begin{remark}[Extreme quantiles]
As $\tau$ is close to zero or one, the non-asymptotic bounds \eqref{opred} and \eqref{eq:L2bdd} in Theorem \ref{thm:rec} become loose as $1/\fline$ increases, so the estimation may not be accurate. In the literature, extreme quantile is often characterized through the low \emph{extremal order} or \emph{extremal rank} $\tau n$ \citep{C99,C05,CFK17}. In particular, letting $\tau\to 0$ as $n\to\infty$ and $n\tau \to c$ for some $c\geq 0$ (respectively, $n(1-\tau)\to c$ for the right extreme quantile, by symmetry we only discuss left quantile in the following), classical asymptotic analysis breaks down if $c$ is small or equal to 0, and such scenario is regarded as "extreme" quantile. Simulation study of asymptotic distribution in \cite{C99} suggests that $n\tau \geq 15$ might be large enough to regard the quantile as "non-extreme". In our application, extreme quantile issue may be mild as $n \tau = 27.65$ with $\tau=1\%$. However, complete analysis for the extreme quantiles under multitask regression scenario is left for future research.
\end{remark}

\section{Conclusions and Future Works}\label{sec:conc}

In this paper, we consider a factor based multitask quantile regression model which allows the factors to vary with quantile levels, and the estimation of such model can be done with the nuclear penalization. Because the typical empirical risk minimizer cannot be efficiently computed due to non-smoothness of the loss function and the expensive subroutines such as singular value decomposition in the algorithm, a numerical procedure that approximately solves the empirical optimization problem is proposed and its theoretical guarantee is proved. Recommendations on how to tune the algorithm for provably accurate estimation are provided. Monte Carlo experiments show the performance of the numerical procedure and the ability to recover the number of factors, even for time series data and extreme quantiles. Potential application of our method is illustrated with a joint analysis on the financial risk of a large pool of stock returns of institutions with large market capitalization.

For future research, the readers may be aware that the model \eqref{eq:mintro} could be misspecified for some applications. To remedy this, nonparametric models may be applied to function $f_k^\tau(\bx)$ by regarding it as an element of a sieve space, and use the basis functions of the sieve space to represent $f_k^\tau(\bx)$ as a series. Methods for estimating this nonparametric model can be derived from adapting our algorithm in Section \ref{sec.al}. Illustrations of this idea using temperature data are presented in an earlier version of this paper \citep[Section 7]{CHY15}. Its theoretical analysis is left for future research. Other interesting research directions include showing that our bounds in Theorem \ref{thm:rec} are unimprovable, and extending our framework to extreme quantiles.

\vskip 2em \centerline{\Large \bf APPENDIX} \vskip -1em
\setcounter{subsection}{0}
\renewcommand{\thesubsection}{A.\arabic{subsection}}
\setcounter{equation}{0}
\renewcommand{\theequation}{A.\arabic{equation}}
\setcounter{theorem}{0}
\renewcommand{\thetheorem}{A.\arabic{theorem}}
\setcounter{algocf}{0}
\renewcommand{\thealgocf}{A.\arabic{algocf}}
\vskip 2em

Details on the algorithm in Section \ref{sec:coef} are provided in Section \ref{sec:alg_detail}. Section \ref{sec:pfcor} provides a proof of Corollary \ref{co:rec}.

\subsection{Details for the Numerical Procedure in Section \ref{sec:coef}}\label{sec:alg_detail}
The Fast Iterative Shrinkage-Thresholding Algorithm (FISTA) of \cite{BT:09} is a popular method for optimizing the loss function consisted of two convex functions. However, one of the major challenge here is that the subgradient of $\widehat Q_\tau(\Sb)$ is not Lipschitz, so the FISTA algorithm may not be stable. To resolve this problem, we apply the method of \cite{N:05} to find a "nice" surrogate for $\widehat Q_\tau(\Sb)$, as will be shown below. 

Recall from \eqref{peml} that the objective function to be minimized is
\begin{align}
L_\tau(\Sb) &=(mn)^{-1} \sum_{i=1}^n \sum_{j=1}^{m} \rho_\tau\big(Y_{ij}-\bX_i^\top\Sb_{\ast j}\big) +\lambda \|\Sb\|_* = \widehat Q_\tau(\Sb)+\lambda \|\Sb\|_*, \label{mqrloss}
\end{align}
We introduce the dual variables $\Theta_{ij}$:
\begin{align}
\widehat Q_\tau(\Sb)= \max_{\Theta_{ij} \in [\tau-1,\tau]} (mn)^{-1} \sum_{i=1}^n \sum_{j=1}^{m} \Theta_{ij} \big(Y_{ij}-\bX_i^\top\Sb_{\ast j}\big). \label{f}
\end{align}
See Section \ref{proof.f} in the supplementary material for a proof of \eqref{f}. 
To smooth this function, denote the matrix $\bTheta = (\Theta_{ij})$ for $i=1,...,n$, $j=1,...,m$, we consider a smooth approximation to $\widehat Q_\tau(\Sb)$ as in equation (2.5) of \cite{N:05}:
\begin{align}
  \widehat Q_{\tau,\kappa}(\Sb) \defeq \max_{\Theta_{ij} \in [\tau-1,\tau]} \Big\{(mn)^{-1} \tilde Q_\tau(\Sb,\bTheta)-\frac{\kappa}{2}\|\bTheta\|_{\rm F}^2 \Big\}, \label{qr_smooth}
\end{align}
where $\tilde Q_\tau(\Sb,\bTheta)\defeq\sum_{i=1}^n \sum_{j=1}^{m} \Theta_{ij}\big(Y_{ij}-\bX_i^\top\Sb_{\ast j}\big)$, and $\kappa>0$ is a smoothing regularization constant depending on $m,n$ and the desired accuracy. When $\kappa \to 0$, $\widehat Q_{\tau,\kappa}(\Sb)$ converges to $\widehat Q_{\tau}(\Sb)$. 
$\widehat Q_{\tau,\kappa}(\Sb)$ defined in \eqref{qr_smooth} has Lipschitz gradient
\begin{align}
\nabla \widehat Q_{\tau,\kappa}(\Sb) \defeq -(mn)^{-1}\Xb^\top [[(\kappa m n)^{-1} (\Yb-\Xb\Sb)]]_\tau, \label{sm.grad}
\end{align}
where $\Xb = [\bX_1\ \bX_2\ ...\ \bX_n]^\top$, $[[\Ab]]_\tau = ([[A_{ij}]]_\tau)$ performs component-wise truncation on a real matrix $\Ab$ to the interval $[\tau-1,\tau]$; in particular,
$$
[[A_{ij}]]_\tau = \left\{\begin{array}{ll}
\tau, &\mbox{ if }A_{ij} \geq \tau; \\
A_{ij}, &\mbox{ if } \tau-1<A_{ij}<\tau; \\
\tau-1, &\mbox{ if } A_{ij} \leq \tau-1.
\end{array}\right.
$$
Observe that \eqref{sm.grad} is similar to the subgradient $-\Xb \{\tau-\IF(\Yb-\Xb\Sb \leq 0)\}$ of $\widehat Q_\tau(\Sb)$, where the operator $\tau-\IF(\cdot \leq 0)$ applies component-wise to the matrix $\Yb-\Xb\Sb$. The major difference lies in the fact that \eqref{sm.grad} replaces the discrete non-Lipschitz $\tau-\IF(\Yb-\Xb\Sb \leq 0)$ with a Lipschitz function $[[\kappa^{-1} (\Yb-\Xb\Sb)]]_\tau$. 

Now, we replace the optimization problem involving $L_\tau(\Sb)$ in \eqref{mqrloss} by the one involving
	\begin{align}
	\widetilde L_\tau(\Sb) \defeq \widehat Q_{\tau,\kappa}(\Sb)+\lambda \|\Sb\|_*, \label{eq:tL}
	\end{align}
	where we recall the definition of $\widehat Q_{\tau,\kappa}(\Sb)$ in \eqref{qr_smooth}. Since the gradient of $\widehat Q_{\tau,\kappa}(\Sb)$ is Lipschitz, we may apply FISTA of \cite{BT:09} for minimizing \eqref{eq:tL}. Define $S_\lambda(\cdot)$ to be the proximity operator on $\R^{p\times m}$:
	\begin{align}
	S_{\lambda}(\Sb) \defeq \Ub (\Db-\lambda\Ib_{p\times m})_+ \Vb^\top, \label{eq:prox}
	\end{align}
	where $\Ib_{p\times m}$ is the $p \times m$ rectangular identity matrix with the main diagonal elements equal to 1, and the SVD $\Sb=\Ub\Db\Vb^\top$. See Theorem \ref{prox_nu} in the supplementary material for more detail for the proximity operator. 
	
	Specific steps are summarized in Algorithm \ref{alg_qr}.

\RestyleAlgo{boxruled}
\LinesNumbered
\begin{algorithm}[h!]
\textbf{Input:} $\Yb$, $\Xb$, $0<\tau<1$, $\lambda>0$, $\kappa>0$\; 
\textbf{Initialization:} $\bGamma_{\tau,0}=0,\bOmega_{\tau,1}=0$, step size $\delta_1 = 1$, $M = \frac{1}{\kappa m^2 n^2} \|\Xb\|^2$\; 
\For{$t=1,2,...,T$}{
$\bGamma_{\tau,t}=\prox\big(\bOmega_{\tau,t}-\frac{1}{M} \nabla \widehat Q_{\tau,\kappa}(\bOmega_{\tau,t})\big)$\;
$\delta_{t+1} = \frac{1+\sqrt{1+4 \delta_t^2}}{2}$\;
$\bOmega_{\tau,t+1} = \bGamma_{\tau,t} + \frac{\delta_t-1}{\delta_{t+1}}(\bGamma_{\tau,t}-\bGamma_{\tau,t-1})$\;
}
\textbf{Output:} $\bGamma_{\tau,T}$
\caption{}\label{alg_qr}
\end{algorithm}



\subsection{Proof of Corollary \ref{co:rec}}\label{sec:pfcor}

To apply Theorem \ref{thm:rec}, it is enough to find a bound for the optimization error $\delta$ of $\Gamat$ that holds with high probability. Suppose the initial estimator is $\bGamma_{\tau,t=0}=0$. The bound in \eqref{qr.bound} suggests that
\begin{align}
	L_\tau(\Gamat)-L_\tau(\Gamah) \leq \underbrace{\frac{3\kappa mn (\tau \vee \{1-\tau\})^2}{2}}_{\mbox{\scriptsize(I)}} + \underbrace{\frac{4 (\|\bGamma_\tau\|_{\rm F}^2+\|\bGamma_{\tau,\infty}-\bGamma_\tau\|_{\rm F}^2)}{(T+1)^2} \frac{\|\Xb\|^2}{\kappa m^2n^2}}_{\mbox{\scriptsize(II)}}, \label{eq:opterr2}
\end{align} 
where we apply the bound $\|\bGamma_{\tau,\infty}\|_{\rm F}^2\leq 2\|\bGamma_{\tau,\infty}-\bGamma_\tau\|_{\rm F}^2+2\|\bGamma_\tau\|_{\rm F}^2$. It is sufficient to show that (I)+(II)$\leq \lambda (m/n)^{1/2}$, and the desired conclusion will follow from Theorem \ref{thm:rec}. Using the bound on $\kappa$ in \eqref{eq:kas}, elementary calculation verifies that (I) in \eqref{eq:opterr2} is less than $\lambda (m/n)^{1/2}/2$. Under the event that $\|\bGamma_{\tau,\infty}-\bGamma_\tau\|_{\rm F}^2=o(1)$, it follows by elementary calculation that (II)$\leq \lambda (m/n)^{1/2}/2$, and the corollary is proved.

\bigskip
It is left to show that $\|\bGamma_{\tau,\infty}-\bGamma_\tau\|_{\rm F}^2=o(1)$ with high probability. Recall that $\Gamai=\lim_{T\to\infty}\Gamat=\arg\,\min_{\Sb}\{\tilde L_\tau(\Sb)=\widehat Q_{\tau,\kappa}(\Sb)+\lambda \|\Sb\|_*\}$. The optimization error $L_\tau(\Gamai)-L_\tau(\Gamah)$ of $\Gamai$ can be estimated by 
\begin{align}
	L_\tau(\Gamai)-L_\tau(\Gamah) &\leq (L_\tau(\Gamai)-\widetilde L_\tau(\bGamma_{\tau,\infty}))+(\underbrace{\widetilde L_\tau(\bGamma_{\tau,\infty})-\widetilde L_\tau(\Gamah)}_{\mbox{\scriptsize $\leq 0$}})+(\underbrace{\widetilde L_\tau(\Gamah)-L_\tau(\Gamah)}_{\mbox{\scriptsize $\leq 0$ by \eqref{eq:smerrbd}}}) \notag\\
	&\leq L_\tau(\Gamai)-\widetilde L_\tau(\bGamma_{\tau,\infty})\notag\\
	&\hspace{-0.25cm}\stackrel{\eqref{eq:smerrbd}}{\leq} \kappa (\tau \vee \{1-\tau\})^2 \frac{nm}{2}\notag\\
	&\leq \lambda(m/n)^{1/2}/6.\label{eq:opterr1}
\end{align}
Therefore, $\Gamai$ is an approximate optimizer. The growth condition \eqref{eq:gr} with $C=1$ in the hypothesis of this corollary ensures \eqref{eq:gr} holds with $C=1/6$. Hence, \eqref{eq:L2bdd} in Theorem \ref{thm:rec} yields
\begin{align}
	\|\Gamai-\bGamma_\tau\|_{\rm F}\leq u\bigg(\frac{m}{\sigma_{\min}(\bSigma_X)}\bigg)^{1/2} \epsilon_{n,\tau,r}, \label{eq:gamaibdd} 
\end{align}
with probability at least $1-\eta-\gamma_n-16(pm)^{1-u^2}-3\exp\{-(p+m)\log 8\}$. 

\bibliographystyle{dcu}
\bibliography{bibgqr}

\newpage

\setcounter{page}{1}
\vskip 1em \centerline{\Large \bf SUPPLEMENTARY MATERIAL: FACTORISABLE} \centerline{\Large \bf MUITITASK QUANTILE REGRESSION} \vskip 2em

Section S.1 presents the convergence analysis for the algorithm. Section S.2 provides details on the non-asymptotic risk analysis of $\tG$. Section S.3 discusses technical detail and remarks. Section S.4 lists some auxiliary results.

\bigskip
\noindent{\bf Additional notations. }For any two matrices $\Ab,\Bb \in \R^{p\times m}$, $\langle \cdot,\cdot \rangle: \R^{n \times m} \times \R^{n \times m} \to \R$ denotes the trace inner product given by $\langle \Ab,\Bb \rangle = \tr(\Ab \Bb^\top)$. Define the empirical measure of $(\bY_i,\bX_i)$ by $\PP_n$. For a function $f:\R^p \to \R$, and $\bZ_i \in \R^p$, define the \emph{empirical process} $\GG_n(f) = n^{-1/2} \sum_{i=1}^n \{f(\bZ_i)-\E[f(\bZ_i)]\}$. The subgradient for $\widehat Q_\tau(\Sb)$ is the matrix
\begin{align}
\nabla \widehat Q_\tau(\Sb) \defeq (nm)^{-1} \sum_{i=1}^n \bX_i\bW_{\tau,i\ast}(\Sb)^\top \defeq (nm)^{-1} \Xb^\top \Wb_\tau(\Sb) \in \R^{p \times m},\label{eq:subgra}
\end{align}
where $$
\bW_{\tau,i\ast}(\Sb) \defeq \left(\IF(Y_{ij}-\bX_i^\top \Sb_{\ast j}\leq 0)-\tau\right)_{1 \leq j \leq m}, \quad \Wb_\tau(\Sb) = [\bW_{\tau,1}(\Sb)\ ...\ \bW_{\tau,n}(\Sb)]^\top \in \R^{n \times m}.
$$
For the true coefficient matrix $\bGamma_\tau$, $\bW_{\tau,i\ast}(\bGamma_\tau) \defeq \bW_{\tau,i\ast}$ and $\Wb_\tau \defeq \Wb_\tau(\bGamma_\tau)$.


\vskip 2em \centerline{\Large \bf S.1: Proofs for Algorithmic Convergence Analysis} \vskip -1em
\setcounter{subsection}{0}
\renewcommand{\thesubsection}{S.1.\arabic{subsection}}
\setcounter{equation}{0}
\renewcommand{\theequation}{S.1.\arabic{equation}}
\setcounter{theorem}{0}
\renewcommand{\thetheorem}{S.1.\arabic{theorem}}

\subsection{Proof of \eqref{f}}\label{proof.f}
To see that this equation holds, note that for each pair of $i,j$, when $Y_{ij}-\bX_i^\top\Sb_{\ast j} > 0$, $\Theta_{ij} = \tau$, since $\tau$ is the largest "positive" value in the interval $[\tau-1,\tau]$. When $Y_{ij}-\bX_i^\top\Sb_{\ast j} \leq 0$, $\Theta_{ij}= \tau-1$ since $\tau$ is the smallest "negative" value in the interval $[\tau-1,\tau]$. This verifies the equation. \hfill$\qed$

\begin{remark}
It is necessary to choose $[\tau-1,\tau]$ rather than $\{\tau-1,\tau\}$ for the support of $\Theta_{ij}$ in \eqref{f} (though both choices fulfill the equation). The previous choice is an interval and is therefore a convex set, and the conditions given in \cite{N:05} is fulfilled.
\end{remark}

\subsection{Proof of Theorem \ref{thm.qr}}\label{proof.thmqr}
Recall the definition of $L_\tau(\Sb)$ and $\hQ_\tau(\Sb)$ in \eqref{mqrloss}, $\tilde L_\tau(\Sb)$ and $\hQ_{\tau,\kappa}(\Sb)$ in \eqref{eq:tL} and \eqref{qr_smooth}. We note a comparison property in (2.7) of \cite{N:05}, for an arbitrary $\Sb\in\R^{p\times m}$,
\begin{align}
\hQ_{\tau,\kappa}(\Sb) \leq \hQ_\tau(\Sb) \leq \hQ_{\tau,\kappa}(\Sb)+\kappa \max_{\bTheta \in[\tau-1,\tau]^{n \times m}} \frac{\|\bTheta\|_{\rm F}^2}{2} \label{eq:smerrbd}
\end{align}
where 
\begin{align*}
\max_{\bTheta \in[\tau-1,\tau]^{n \times m}} \|\bTheta\|_{\rm F}^2 = \max_{\bTheta \in[\tau-1,\tau]^{n \times m}} \sum_{i \leq n,j \leq m} \Theta_{ij}^2 \leq (\tau \vee \{1-\tau\})^2 nm.
\end{align*}
Recall that $\Gamah$ is a minimizer of $L_\tau(\Sb)$ defined in \eqref{mqrloss}. It follows by \eqref{eq:smerrbd} that for an arbitrary $\Sb\in\R^{p\times m}$,
\begin{align}
	\widetilde L_\tau(\Gamah) \leq L_\tau(\Gamah) \leq L_\tau(\Sb) \leq \tilde L_\tau(\Sb)+\kappa (\tau \vee \{1-\tau\})^2 \frac{nm}{2},\label{eq:fact2}
\end{align}
where the first inequality is from the first inequality of \eqref{eq:smerrbd}, the second is the definition of the minimizer $\Gamah$, and the third inequality is from the second inequality of \eqref{eq:smerrbd}. Recall that $\bGamma_{\tau,\infty}=\lim_{t\to\infty}\bGamma_{\tau,t}$ is a minimizer of $\widetilde L_\tau(\Sb)$, then \eqref{eq:fact2} gives
\begin{align}
	\widetilde L_\tau(\bGamma_{\tau,\infty}) \leq \widetilde L_\tau(\Gamah) \leq \widetilde L_\tau(\bGamma_{\tau,\infty})+\kappa (\tau \vee \{1-\tau\})^2 \frac{nm}{2},\label{eq:smgh}
\end{align}
where the first inequality is from the definition of $\bGamma_{\tau,\infty}$ as a minimizer of $\widetilde L_\tau(\Sb)$ and the second inequality is from \eqref{eq:fact2}, which holds for an arbitrary matrix $\Sb\in\R^{p\times m}$.

Now from triangle inequality,
\begin{align}
\big|L_\tau(\Gamat)-L_\tau(\Gamah)\big|\leq &\big|L_\tau(\Gamat)-\widetilde L_\tau(\Gamat)\big| + \big|\widetilde L_\tau(\Gamat)-\tilde L_\tau(\bGamma_{\tau,\infty})\big|+\big|\tilde L_\tau(\bGamma_{\tau,\infty})-\widetilde L_\tau(\Gamah)\big| \notag\\
&+ \big|L_\tau(\Gamah)-\widetilde L_\tau(\Gamah)\big|. \label{temp0}
\end{align}
The third term on the right-hand side of \eqref{temp0} is bounded by \eqref{eq:smgh}. For any matrix $\Sb$, 
we have from \eqref{eq:smerrbd} that
\begin{align}
\big|L_\tau(\Sb)-\widetilde L_\tau(\Sb)\big| \leq \kappa \frac{nm (\tau \vee \{1-\tau\})^2}{2}. \label{temp1} 
\end{align}
Hence, both $\big|L_\tau(\Gamat)-\widetilde L_\tau(\Gamat)\big|$ and $\big|L_\tau(\Gamah)-\widetilde L_\tau(\Gamah)\big|$ satisfy \eqref{temp1}. 

Lemma \ref{thm.smooth} implies that the gradient of $\widehat Q_{\tau,\kappa}(\Sb)$ is Lipschitz continuous with Lipschitz constant $M$. By Theorem 4.1 of \cite{JY:09} or Theorem 4.4 of \cite{BT:09} (applied in general real Hilbert space, see their Remark 2.1), we have
\begin{align}
\big|\widetilde L_\tau(\Gamat) - \widetilde L_\tau(\bGamma_{\tau,\infty}) \big| \leq \frac{2 M \|\bGamma_{\tau,0}-\bGamma_{\tau,\infty}\|_{\rm F}^2}{(t+1)^2}, \label{temp2}
\end{align}
where $M= (\kappa m^2 n^2)^{-1}\|\Xb\|^2$ as given in Lemma \ref{thm.smooth}. 

\subsection{Technical Details for Theorem \ref{thm.qr}}

\begin{lemma}\label{lem.smooth}
For any $\Sb,\bTheta\in\R^{p\times m}$, $\tilde Q_\tau(\Sb,\bTheta)$ can be expressed as $\tilde Q_\tau(\Sb,\bTheta) =  \langle -\Xb\Sb,\bTheta \rangle + \langle \Yb,\bTheta \rangle$.
\end{lemma}
\begin{proof}[Proof of Lemma \ref{lem.smooth}] One can show by elementary matrix algebra that
\begin{align*}
	\tilde Q_\tau(\Sb,\bTheta)&=\sum_{i=1}^n \sum_{j=1}^{m} \Theta_{ij}\left(Y_{ij}-\bX_i^\top\Sb_{\ast j}\right) 
	= \sum_{i=1}^n \sum_{j=1}^{m} \Theta_{ij} Y_{ij} - \sum_{i=1}^n \sum_{j=1}^{m} \Theta_{ij} \bX_i^\top\Sb_{\ast j} \\
	&=  \langle \Yb,\bTheta \rangle + \langle -\Xb\Sb,\bTheta \rangle.
\end{align*}
The proof is therefore completed.
\end{proof}

\begin{lemma}\label{thm.smooth}
For any $\kappa>0$, $\widehat Q_{\tau,\kappa}(\Sb)$ is a well-defined, convex and continuously differentiable function in $\Sb$ with the gradient $\nabla \widehat Q_{\tau,\kappa}(\Sb) = -(mn)^{-1} \Xb^\top \bTheta^*(\Sb)\in \R^{p \times m}$, where $\bTheta^*(\Sb)$ is the optimal solution to \eqref{qr_smooth}, namely
\begin{align}
\bTheta^*(\Sb) = [[(\kappa m n)^{-1} (\Yb-\Xb\Sb)]]_\tau. \label{opt.xi}
\end{align}
The gradient $\nabla \widehat Q_{\tau,\kappa}(\Sb)$ is Lipschitz continuous with the Lipschitz constant $M= (\kappa m^2 n^2)^{-1}\|\Xb\|^2$.
\end{lemma}
\begin{proof}[Proof of Lemma \ref{thm.smooth}]
In view of Lemma \ref{lem.smooth}, we have from \eqref{qr_smooth} that
\begin{align}
  \widehat Q_{\tau,\kappa}(\Sb) = \max_{\Theta_{ij} \in [\tau-1,\tau]} \left\{(mn)^{-1}\langle \Yb,\bTheta \rangle + (mn)^{-1}\langle -\Xb\Sb,\bTheta \rangle-\frac{\kappa}{2}\|\bTheta\|_{\rm F}^2 \right\}.
\end{align}
$\widehat Q_{\tau,\kappa}(\Sb)$ matches the form in (2.5) on page 131 of \cite{N:05}, with their $\hat \phi(\bTheta)=(mn)^{-1}\langle \Yb,\bTheta \rangle$ which is a continuous convex function, and their $A = -(mn)^{-1}\Xb$ which maps from the vector space $\R^{p\times m}$ to the space $\R^{n \times m}$ (the model setting described below (2.2) on page 129 of \cite{N:05}), and their $d_2(\bTheta) = \frac{\kappa}{2}\|\bTheta\|_{\rm F}^2$. 
Therefore, applying Theorem 1 of \cite{N:05}, with $\sigma_2=1$, $d(\bTheta)=\|\bTheta\|_{\rm F}^2/2$, the gradient $\nabla \widehat Q_{\tau,\kappa}(\Sb) = -(mn)^{-1} \Xb^\top \bTheta^*(\Sb)\in \R^{p \times m}$, where $\bTheta^*(\Sb)$ is the optimal solution to \eqref{qr_smooth}:
\begin{align*}
\bTheta^*(\Sb) = [[(\kappa m n)^{-1} (\Yb-\Xb\Sb)]]_\tau,
\end{align*}
and the Lipschitz constant of $\nabla \widehat Q_{\tau,\kappa}(\Sb)$ is $\|\Xb\|/(\kappa n^2m^2)$, where $\|\Xb\|$ is the spectral norm of $\Xb$ (see line 8 on page 129 of \cite{N:05}). Hence, the proof is completed.
\end{proof}

\vskip 2em \centerline{\Large \bf S.2: Proofs for Non-Asymptotic Bounds}
\setcounter{subsection}{0}
\renewcommand{\thesubsection}{S.2.\arabic{subsection}}
\setcounter{equation}{0}
\renewcommand{\theequation}{S.2.\arabic{equation}}
\setcounter{theorem}{0}
\renewcommand{\thetheorem}{S.2.\arabic{theorem}}
\vskip 1em

\begin{remark}\label{rem:normeq}
	For any $\bDelta\in\R^{p\times m}$, from \hyperref[A2]{(A2)},
	\begin{align}
		\|\bDelta\|_{L_2(P_X)}^2 = m^{-1} \E\big[\|\bDelta^\top \bX_i\|_2^2\big] = m^{-1} \sum_{j=1}^m \bDelta_{*j}^\top \E[\bX_i\bX_i^\top]\bDelta_{*j} \geq m^{-1}\sigma_{\min}(\bSigma_X) \|\bDelta\|_{\rm F}^2.\label{eq:frobeq}
	\end{align}
	Moreover, by $\|\p_{\bGamma_\tau}(\bDelta)\|_{\rm F} \leq \|\bDelta\|_{\rm F}$, we have a bound  
	\begin{align}
		\|\bDelta\|_{L_2(P_X)}  \geq \Big(\frac{\sigma_{\min}(\bSigma_X)}{m}\Big)^{1/2} \|\bDelta\|_{\rm F} \geq \Big(\frac{\sigma_{\min}(\bSigma_X)}{m}\Big)^{1/2} \|\p_{\bGamma_\tau}(\bDelta)\|_{\rm F}. \label{eq:pfrobeq}
	\end{align}
\end{remark}

\subsection{Proof for Lemma \ref{lem.rate}}\label{sec:proof_lem_rate}
To prove the first statement, applying the same $\mathcal E$-net argument on the unit Euclidean sphere $\cS^{m-1} = \{\ub \in \R^m: \|\ub\|_2=1\}$ as in the first part of the proof of Lemma 3 in \cite{NW:11} (page 6 to the beginning of page 7 in their supplemental materials), we obtain
\begin{align}
\P\bigg(\frac{1}{n}\|\Xb^\top \Wb_\tau\| \geq 4 s \bigg) = \P\bigg(\sup_{\vb \in S^{p-1}\atop\ub \in \cS^{m-1}} \frac{1}{n} \big|\vb^\top\Xb^\top \Wb_\tau\ub\big| \geq 4 s \bigg) \leq 8^{p+m} \sup_{\vb\in S^{p-1}, \ub\in \cS^{m-1}\atop\|\ub\|=\|\vb\|=1} \P\bigg(\frac{|\langle \Xb \vb,\Wb_\tau \ub\rangle|}{n} \geq s \bigg). \label{NWlem3}	
\end{align}
To bound $n^{-1} \langle \Xb \vb,\Wb_\tau \ub\rangle = n^{-1} \sum_{i=1}^n \langle \vb,\bX_i\rangle\langle \ub,\bW_{\tau,i\ast}\rangle$, first we show the sub-Gaussianity of $\langle \ub,\bW_{\tau,i\ast}\rangle$. Theorem 3.1 of \cite{BM13} suggests that the sub-Gaussian norm of the $j$th component of $\bW_{\tau,i\ast}$ is
\begin{align*}
		\|W_{\tau,ij}\|_{\psi_2} =  \left\{\begin{array}{ll}
					0, &\  \tau=0,1;\\
					\frac{2\tau-1}{2\{\log \tau-\log (1-\tau)\}}, &\  \tau \in (0,1)-\{1/2\};\\
					1/4, &\ \tau=1/2,
					\end{array}
				\right.
	\end{align*} 
	where $\|\cdot\|_{\psi_2}$ denotes the sub-Gaussian norm. It follows by Lemma \ref{lem.hoef} (Hoeffding's inequality) that 
\begin{align*}
	\P\big(\langle \ub,\bW_{\tau,i\ast}\rangle \geq s\big) \leq \exp\bigg(1-\frac{C's^2}{K(\tau)\|\ub\|_2^2}\bigg) = \exp\bigg(1-\frac{C's^2}{K(\tau)}\bigg).
\end{align*} 

We apply Lemma \ref{lem.hoef} again to bound $n^{-1} \sum_{i=1}^n \langle \vb,\bX_i\rangle\langle \ub,\bW_{\tau,i\ast}\rangle$. Conditioning on $\bX_i$, we have
\begin{align*}
	\P\bigg(\bigg|n^{-1} \sum_{i=1}^n \langle \vb,\bX_i\rangle\langle \ub,\bW_{\tau,i\ast}\rangle\bigg| \geq s\bigg) &\leq \exp\bigg(1-\frac{C'n s^2}{K(\tau)n^{-1}\sum_{i=1}^n \langle \vb,\bX_i\rangle^2}\bigg) \\
	&\leq \exp\bigg(1-\frac{C'n s^2}{K(\tau)c_2 \|\bSigma_X\| }\bigg).
\end{align*} 
where the second inequality follows from the fact that $\|\vb\|_2=1$ and $n^{-1}\sum_{i=1}^n \langle \vb,\bX_i\rangle^2 \leq \|\Xb^\top\Xb/n\| \leq c_2 \|\bSigma_X\|$ on the event that \hyperref[A2]{(A2)} holds. 

To summarize, on the event that \hyperref[A2]{(A2)} holds,
\begin{align*}
\P\bigg(\frac{1}{n}\|\Xb^\top \Wb_\tau\| \geq 4 s \bigg) &\leq 8^{p+m} \exp\bigg(1-\frac{C'n s^2}{K(\tau)c_2 \|\bSigma_X\| }\bigg)\\
&\leq \exp\bigg(1-\frac{C'ns^2}{K(\tau)c_2 \|\bSigma_X\|}+(p+m)\log 8\bigg).
\end{align*}
Therefore, for arbitrary $u>1$, the event
\begin{align}
\frac{1}{n}\|\Xb^\top \Wb_\tau\| \geq 4 \cdot \sqrt{u (\log 8)\frac{K(\tau)c_2 \|\bSigma_X\|}{C'}}\sqrt{\frac{p+m}{n}}, \label{eq:temp1bound}
\end{align}
has probability smaller than $3e^{-(u-1)(p+m)\log 8}+\gamma_n$, as $e < 3$.

\bigskip
To prove the second statement, we note that the event in \eqref{eq:temp1bound} has probability less than $\eta$ by setting $k=1-(\eta-\gamma_n)/(3 (p+m)\log 8)$. 

 \hfill $\qed$

\subsection{Proof for Theorem \ref{thm:rec}}\label{sec:proof_lem_ora}

Before we prove Theorem \ref{thm:rec}, we first define the "support" of matrices by projections. 

\begin{defin}\label{def:proj}
	For $\Ab \in \R^{p \times m}$ with rank $r$, the singular value decomposition of $\Ab$ is $\Ab = \sum_{j=1}^r \sigma(\Ab) \ub_j \vb_j^\top$. The \emph{support} of $\Ab$ is defined by $(S_1,S_2)$ in which $S_1 = \mbox{span}\{\ub_1,...,\ub_r\}$ and $S_2 =  \mbox{span}\{\vb_1,...,\vb_r\}$. Define the projection matrix on $S_1$:
	$
	\Pb_1 \defeq \Ub_{[1:r]}\Ub_{[1:r]}^\top, 
	$
	in which $\Ub_{[1:r]} = [\ub_1\,...\,\ub_r]\in\R^{p\times r}$; $\Pb_2 \defeq \Vb_{[1:r]} \Vb_{[1:r]}^\top$, where $\Vb_{[1:r]}=[\vb_1\,...\,\vb_r]\in\R^{m\times r}$. Denote $\Pb_1^\perp=\Ib_{p\times r}-\Pb_1$ and $\Pb_2^\perp=\Ib_{m\times r}-\Pb_2$. For any matrix $\Sb \in \R^{p \times m}$, define
$$
\mathcal P_{\Ab}(\Sb) \defeq \Pb_1 \Sb \Pb_2; \quad \mathcal P_{\Ab}^\perp(\Sb) \defeq \Pb_1^\perp \Sb \Pb_2^\perp.
$$
Define for any $a\geq 0$,
\begin{align}
	\K(\bGamma_\tau;a) &\defeq \left\{\Sb \in \R^{p \times m}: \|\p_{\bGamma_\tau}^\perp(\Sb)\|_* \leq 3 \|\p_{\bGamma_\tau}(\Sb)\|_* + a\right\} \label{eq:conea}.
\end{align}
\end{defin}
	We note that nuclear norm is \emph{decomposable} under the projection: for any $\Sb,\Ab\in\R^{p\times m}$, $\|\Sb\|_* = \|\p_{\Ab}(\Sb)\|_*+\|\p_{\Ab}^\perp(\Sb)\|_*$. This is analogous to the $\ell_1$ norm for vectors: for any vector $\bv$ and support $S$, $\|\bv\|_1=\|\bv_S\|_1+\|\bv_{S^c}\|_1$; see Definition 1 on page 541 of \cite{NRWY:12}. Moreover, the rank of $\p_{\Ab}(\Sb)$ is at most $\rank(\Ab)$.

The shape of $\K(\bGamma_\tau;a)$ is not a cone when $a>0$, but is still a star-shaped set. This set has a similar shape as the set defined in equation (17) on page 544 in \cite{NRWY:12}. See also their Figure 1 on page 544.

To simplify the notations in this proof, let
\begin{align}
	\tdel &= \tG-\bGamma_\tau,\\
	\alpha_r&= 4\sqrt{r/\sigma_{\min}(\bSigma_X)},\\ 
	\alpha_{r,m}&= m^{1/2} \alpha_r,\label{eq:alpharm}\\
c_n &= 16 \sqrt{2} m^{-1/2} \delta \lambda^{-1} \sqrt{c_2\sigma_{\max}(\bSigma_X)+B_p} \sqrt{\log m+\log p},\label{eq:cn}\\
d_n &= 8 \sqrt{2}\alpha_r \sqrt{c_2\sigma_{\max}(\bSigma_X)+B_p} \sqrt{\log m+\log p},\label{eq:dn}
\end{align}
Let the events
\begin{align}
&\Omega_1: \mbox{Assumption \hyperref[A2]{(A2)} holds};\notag\\
&\Omega_2: \cA(t) \leq u (t d_n + c_n) \mbox{ for }u>1, \mbox{ where } \notag\\
&\cA(t) \defeq \sup_{\|\bDelta\|_{L_2(P_X)} \leq t, \bDelta \in \cK(\bGamma_\tau;2\delta/\lambda)} \bigg|\GG_n\bigg[m^{-1}\sum_{j=1}^m\big(\rho_\tau\{Y_{ij}-\bX_i^\top(\bGamma_{\tau,\ast j}+\bDelta_{\ast j})\}-\rho_\tau\{Y_{ij}-\bX_i^\top \bGamma_{\tau,\ast j}\}\big)\bigg]\bigg|.
\label{eq:cA}\\
&\Omega_3: \frac{1}{n}\|\Xb^\top \Wb\| \leq C^* \sqrt{\sigma_{\max}(\bSigma_X)K(\tau)}\sqrt{\frac{p+m}{n}},\notag 
\end{align}
where $C^* = 4\sqrt{2\frac{c_2}{C'}\log 8}$,

The probability of event $\P(\Omega_1 \cap \Omega_2 \cap \Omega_3) \geq 1-\gamma_n-16(pm)^{1-u^2}-3e^{-(p+m)\log 8}$ by Assumption \hyperref[A2]{(A2)}, Lemma \ref{lem.rate} and Lemma \ref{lem.Q}. 

Recall that $\alpha_{r,m}$, $c_n$ and $d_n$ are defined in \eqref{eq:alpharm}, \eqref{eq:cn} and \eqref{eq:dn}. Set 
	\begin{align}
	t = \sqrt{n^{-1/2} u c_n \frac{4}{\fline}+\frac{8}{\fline}\delta}+\frac{4}{\fline}(n^{-1/2}u d_n + \lambda \alpha_{r,m}).\label{eq:oproof_u}
	\end{align}

\bigskip
	We will prove by contradiction. Suppose to the contrary that $\|\tdel\|_{L_2(P_X)} \geq t$. Since $\Gamah$ minimizes $L_\tau(\Sb)=\hQ_\tau(\Sb)+\lambda\|\Sb\|_*$ (defined in \eqref{peml}) and $L_\tau(\Gamah)-L_\tau(\bGamma_\tau)<0$, we have
		\begin{align}
			&\hQ_\tau(\bGamma_\tau+\tdel)-\hQ_\tau(\bGamma_\tau)+\lambda (\|\bGamma_\tau+\tdel\|_*-\|\bGamma_\tau\|_*) \notag\\
			&= L_\tau(\Gamah)-L_\tau(\bGamma_\tau) + L_\tau(\bGamma_\tau+\tdel)-L_\tau(\Gamah) \notag\\
			&\leq \delta, \label{eq:obs0}
		\end{align}
		where we recall \eqref{eq:appest}.
		
	 Observe that $\tdel=\tG-\bGamma_\tau\in\K(\bGamma_\tau;0)\subset\K(\bGamma_\tau;2\delta/\lambda)$ with probability $1-\eta$ by applying \eqref{lambda} and Lemma \ref{lem.nsp}. 
		Hence, from \eqref{eq:appest},
\begin{align}
	\delta > \inf_{\|\bDelta\|_{L_2(P_X)} \geq t, \bDelta \in \cK(\bGamma_\tau;2\delta/\lambda)} \hQ_\tau(\bGamma_\tau+\bDelta)-\hQ_\tau(\bGamma_\tau)+\lambda (\|\bGamma_\tau+\bDelta\|_*-\|\bGamma_\tau\|_*). \label{minorize1}
\end{align}
Note the facts that 
\begin{enumerate}
	\item $\hQ_\tau(\cdot)+\lambda\|\cdot\|_*$ is convex (unique optimum);
	\item $\K(\bGamma_\tau;2\delta/\lambda)$ is star-shaped (see Figure 1 of \cite{NRWY:12}).
\end{enumerate}
Hence, $\|\tdel\|_{L_2(P_X)} \geq t$ can be replaced by $\|\tdel\|_{L_2(P_X)} = t$ and the strict inequality in \eqref{minorize1} is maintained
\begin{align*}
	\delta \geq \inf_{\|\bDelta\|_{L_2(P_X)} = t, \bDelta \in \cK(\bGamma_\tau;2\delta/\lambda)} \hQ_\tau(\bGamma_\tau+\bDelta)-\hQ_\tau(\bGamma_\tau)+\lambda (\|\bGamma_\tau+\bDelta\|_*-\|\bGamma_\tau\|_*).
\end{align*}
It can be deducted from the last display that
\begin{align*}
	\delta \geq \inf_{\|\bDelta\|_{L_2(P_X)} = t, \bDelta \in \cK(\bGamma_\tau;2\delta/\lambda)} Q_\tau(\bGamma_\tau+\bDelta)-Q_\tau(\bGamma_\tau)-n^{-1/2} \cA(t) + \lambda (\|\bGamma_\tau+\bDelta\|_*-\|\bGamma_\tau\|_*),
\end{align*}
By triangle inequality, $\big|\|\bGamma_\tau+\bDelta\|_*-\|\bGamma_\tau\|_*\big| \leq \|\bDelta\|_* \leq \alpha_{r,m} t+2\delta/\lambda$ on the set $\{\|\bDelta\|_{L_2(P_X)} = t, \bDelta \in \cK(\bGamma_\tau;2\delta/\lambda)\}$ by Lemma \ref{lem.mis}(ii). Applying the bound in $\Omega_2$ obtains
\begin{align*}
	\delta \geq \inf_{\|\bDelta\|_{L_2(P_X)} = t, \bDelta \in \cK(\bGamma_\tau;2\delta/\lambda)} Q_\tau(\bGamma_\tau+\bDelta)-Q_\tau(\bGamma_\tau)-n^{-1/2} u (d_n t+c_n)- \lambda (\alpha_{r,m} t+2 \delta/\lambda).
\end{align*}
Since $\delta \leq C \lambda \sqrt{m/n}$, by Remark \ref{rmk:nu},
$$
\nu_\tau(2\delta/\lambda) \geq \nu_\tau(2C\sqrt{m/n}) > u \epsilon_{n,\tau,r} \geq t/4
$$ 
(where the second inequality is from \eqref{eq:gr}; the last inequality will be shown in \eqref{eq:terr} below), invoking Lemma \ref{lem.mis} (i) to get the minorization
\begin{align}
	\delta\geq \inf_{\|\bDelta\|_{L_2(P_X)} = t, \bDelta \in \cK(\bGamma_\tau;2\delta/\lambda)} \frac{1}{4}\fline t^2-n^{-1/2} u (d_n t+c_n)- \lambda (\alpha_{r,m} t+2 \delta/\lambda). \label{thm.finali1}
\end{align}
Rearranging terms to get
\begin{align}
	0 \geq \inf_{\|\bDelta\|_{L_2(P_X)} = t, \bDelta \in \cK(\bGamma_\tau;2\delta/\lambda)} \frac{1}{4}\fline t^2-n^{-1/2} u (d_n t+c_n)- \lambda \alpha_{r,m} t-3 \delta. \label{thm.finali}
\end{align}
However, the right-hand side of \eqref{thm.finali} is strictly greater than 0 whenever
\begin{align}
	t > \frac{2}{\fline}(n^{-1/2}u d_n + \lambda \alpha_{r,m}) + \frac{2}{\fline}\sqrt{(n^{-1/2}u d_n + \lambda \alpha_{r,m})^2 + \fline(n^{-1/2} u c_n +3\delta)}.
\end{align}
The right hand side of the last display is upper bounded by (by $\sqrt{a+b}<\sqrt{a}+\sqrt{b}$ for all $a,b>0$)
\begin{align*}
	 t = \frac{2}{\fline}(n^{-1/2}u d_n + \lambda \alpha_{r,m}) +  \frac{2}{\fline}(n^{-1/2}u d_n + \lambda \alpha_{r,m})+\sqrt{\frac{4}{\fline} n^{-1/2} u c_n +\frac{12}{\fline}\delta},
\end{align*}
which leads to the $t$ in \eqref{eq:oproof_u}. We get a contradiction, so  $\|\tdel\|_{L_2(P_X)} \geq t$ does not hold. Namely, $\|\tdel\|_{L_2(P_X)} < t$.

\bigskip
To show \eqref{opred}, we will prove
\begin{align}
	t \leq u \epsilon_{n,\tau,r}, \label{eq:terr}
\end{align}
where $\epsilon_{n,\tau,r}$ is defined in \eqref{eq:gr}. To see this, first note that,
\begin{align}
	\lambda \stackrel{\eqref{lambda}}{\leq} 2\bar\lambda &\stackrel{\eqref{eq:bddquant}}{\leq} 2\frac{C^*}{m} \sqrt{\Big(1-\frac{\eta-\gamma_n}{3 (p+m)\log 8}\Big)\sigma_{\max}(\bSigma_X)K(\tau)}\sqrt{\frac{p+m}{n}} \notag\\
	&\leq \frac{2C^*}{m} \sqrt{\sigma_{\max}(\bSigma_X)K(\tau)}\sqrt{\frac{p+m}{n}} \label{eq:lamest}
\end{align}
since $0<\eta< 1$ and $\gamma_n\to 0$. 

Elementary calculation shows that for $u\geq 1$,
\begin{align} 
	&\max\big\{2\lambda\alpha_{r,m}/\fline, 2n^{-1/2}ud_n/\fline\big\}\notag\\
	&\quad\quad\leq  \frac{2(32\sqrt{2}+8C^*) u}{\fline\wedge 1}\sqrt{\frac{\sigma_{\max}(\bSigma_X) \vee 1}{\sigma_{\min}(\bSigma_X)\wedge 1}}  \sqrt{\frac{r(m+p \vee B_p)(\log p+\log m)}{m n}}.  \label{eq:recbd1}
\end{align}

Under the condition that $\delta < \lambda m^{1/2} n^{-1/2}$, $r\geq 1$,
\begin{align}
 \sqrt{\frac{1}{\fline} n^{-1/2} u c_n} 
 &\leq \sqrt{\frac{u}{\fline} \frac{d_n}{\alpha_r n}} 
 \leq \alpha_r^{-1/2} \frac{u d_n}{(\fline \wedge 1)\sqrt{n}}  
 \leq \frac{(\sigma_{\min}(\bSigma_X)^{1/2} \vee 1) u d_n}{(\fline \wedge 1)\sqrt{n}} \notag \\
 &\leq \frac{(32\sqrt{2}+8C^*)u}{\fline \wedge 1}\sqrt{\frac{\sigma_{\max}(\bSigma_X) \vee 1}{\sigma_{\min}(\bSigma_X)\wedge 1}}  \sqrt{\frac{r(m+p \vee B_p)(\log p+\log m)}{m n}} \label{eq:recbd2}
\end{align}
since $u\geq 1$, $d_n\geq 1$ (as $m,p\to\infty$).

Lastly, again from $\delta < \lambda m^{1/2} n^{-1/2}$ and \eqref{eq:lamest},
\begin{align}
	\delta &\leq \lambda m^{1/2} n^{-1/2} \leq 2 C^*\sqrt{\sigma_{\max}(\bSigma_X) K(\tau)\frac{p+m}{n^2 m}} \leq C^*n^{-1}\sqrt{\sigma_{\max}(\bSigma_X) \frac{p+m}{m}} \notag\\
	&\leq C^* \frac{p+m}{nm} \sqrt{\sigma_{\max}(\bSigma_X)}, \label{eq:recbd3_1}
\end{align}
where in the third inequality the fact $\sup_\tau|K(\tau)| \leq 1/4$ (noted below Lemma \ref{lem.rate}, or in (K4) of Lemma 2.1 on p.35 of \cite{BM13}) is applied, where $K(\tau)$ is defined in \eqref{eq:Ktau}; in the last inequality, the fact $\sqrt{1+p/m}\leq 1+p/m$ is applied. 
Hence,
\begin{align}
	\sqrt{\frac{1}{\fline}\delta} &\leq \frac{1}{\fline \wedge 1}\sqrt{C^*\frac{p+m}{nm}} \sigma_{\max}(\bSigma_X)^{1/4} \notag\\
	&\leq \frac{(32\sqrt{2}+8C^*)u}{\fline \wedge 1}\sqrt{\frac{\sigma_{\max}(\bSigma_X) \vee 1}{\sigma_{\min}(\bSigma_X)\wedge 1}}  \sqrt{\frac{r(m+p \vee B_p)(\log p+\log m)}{m n}} \label{eq:recbd3}
\end{align}
where the inequality follows by the facts:
\begin{align*}
	\bullet\ &\frac{\sqrt{C^*}}{\fline \wedge 1} \leq \frac{C^*}{\fline \wedge 1} \leq \frac{(32\sqrt{2}+8C^*)u}{\fline \wedge 1}, \quad \mbox{($u> 1$ from the hypothesis of the Theorem,}\\ 
	&\mbox{and $C^*\geq 1$ from Lemma \ref{lem.rate})}\\
	\bullet\ &\sigma_{\max}(\bSigma_X)^{1/4} \leq (\sigma_{\max}(\bSigma_X)\vee 1)^{1/4} \leq (\sigma_{\max}(\bSigma_X)\vee 1)^{1/2} \leq \sqrt{\frac{\sigma_{\max}(\bSigma_X) \vee 1}{\sigma_{\min}(\bSigma_X)\wedge 1}}\\
\bullet\ &\sqrt{\frac{p+m}{nm}} \leq \sqrt{\frac{r(m+p \vee B_p)(\log p+\log m)}{m n}}, \quad \mbox{$B_p\geq 1$ by \hyperref[A2]{(A2)}, $r\geq 1$, $p,m\geq 3$ in \ref{A1}}.
	\end{align*}
Note that if $r=\rank(\bGamma_\tau)=0$, then the matrix $\bGamma_\tau=0$ and this case is excluded.

Combining \eqref{eq:recbd1}, \eqref{eq:recbd2} and \eqref{eq:recbd3} gives \eqref{eq:terr}.

\hfill $\qed$

\subsection{Technical Details for Theorem \ref{thm:rec}}\label{sec:tech_oracle}

The following lemma asserts that the empirical error $\tG-\bGamma_\tau$ lies in the cone $\K(\bGamma_\tau;2\delta/\lambda)$. 
\begin{lemma}\label{lem.nsp}
	Suppose $\lambda \geq 2\|\nabla \hQ(\bGamma_\tau)\|$ and $\tdel = \tG-\bGamma_\tau$, where $\nabla \hQ(\bGamma_\tau)$ is the subgradient of $\hQ(\bGamma_\tau)$ defined in \eqref{eq:subgra}. Then $\|\cP_{\bGamma_\tau}^\perp(\tdel)\|_* \leq 3 \|\cP_{\bGamma_\tau}(\tdel)\|_*+2\delta'/\lambda$ for all $\delta'\geq \delta$. That is, $\tdel \in \K(\bGamma_\tau;2\delta'/\lambda)$ for all $\delta'\geq \delta$.
\end{lemma}
\begin{proof}[Proof for Lemma \ref{lem.nsp}]\label{proof.lem.nsp}
\begin{align}
	0 &\leq \hQ_\tau(\bGamma_\tau)-\hQ_\tau(\Gamah)+\lambda (\|\bGamma_\tau\|_*-\|\Gamah\|_*) \quad (\mbox{$\widehat\bGamma_\tau$ is the minimizer of $\hQ_\tau(\Sb)+\lambda\|\Sb\|_*$})\notag\\
	&\leq \hQ_\tau(\bGamma_\tau)-\hQ_\tau(\tG)+\lambda (\|\bGamma_\tau\|_*-\|\tG\|_*)+\delta\quad (\mbox{by \eqref{eq:appest}})\notag\\
	&\leq \|\nabla \hQ_\tau(\bGamma_\tau) \| \|\tdel\|_* + \lambda (\|\bGamma_\tau\|_*-\|\tG\|_*) +\delta\notag\\
	&\leq \|\nabla \hQ_\tau(\bGamma_\tau) \| \big(\|\cP_{\bGamma_\tau}(\tdel)\|_*+\|\cP_{\bGamma_\tau}^\perp(\tdel)\|_*\big) + \lambda (\|\cP_{\bGamma_\tau}(\bGamma_\tau)\|_*-\|\cP_{\bGamma_\tau}^\perp(\tG)\|_*-\|\cP_{\bGamma_\tau}(\tG)\|_*)+\delta\notag\\
&\leq \|\nabla \hQ_\tau(\bGamma_\tau) \| \big(\|\cP_{\bGamma_\tau}(\tdel)\|_*+\|\cP_{\bGamma_\tau}^\perp(\tdel)\|_*\big) + \lambda (\|\cP_{\bGamma_\tau}(\tdel)\|_*-\|\cP_{\bGamma_\tau}^\perp(\tdel)\|_*)+\delta, \label{eq:subgradient}
\end{align}
where the second inequality follows from the definition of subgradient:
\begin{align*}
	\hQ_\tau(\Gamah)-\hQ_\tau(\bGamma_\tau) \geq \langle \nabla\hQ_\tau(\bGamma_\tau),\Gamah-\bGamma_\tau \rangle,
\end{align*}
and H\"older's inequality; the third inequality is from the fact that $\cP_{\bGamma_\tau}^\perp(\bGamma_\tau)=0$ and for any $\Sb$, $\|\Sb\|_* = \|\cP_{\bGamma_\tau}(\Sb)\|_*+\|\cP_{\bGamma_\tau}^\perp(\Sb)\|_*$ (the discussion after Definition \ref{def:proj})
; the fourth inequality is from the triangle inequality. 

Rearrange expression \eqref{eq:subgradient} to get,
$$
(\lambda-\|\nabla \hQ_\tau(\bGamma_\tau) \|) \|\cP_{\bGamma_\tau}^\perp(\tdel)\|_* \leq (\lambda+\|\nabla \hQ_\tau(\bGamma_\tau) \|) \|\cP_{\bGamma_\tau}(\tdel)\|_*+\delta.
$$	
Choose $\lambda \geq 2 \|\nabla \hQ_\tau(\bGamma_\tau) \|$,
$$
\frac{1}{2} \lambda \|\cP_{\bGamma_\tau}^\perp(\tdel)\|_* \leq \frac{3}{2} \lambda \|\cP_{\bGamma_\tau}(\tdel)\|_*+\delta.
$$
Hence, $\|\cP_{\bGamma_\tau}^\perp(\tdel)\|_* \leq 3 \|\cP_{\bGamma_\tau}(\tdel)\|_*+2\delta/\lambda \leq 3 \|\cP_{\bGamma_\tau}(\tdel)\|_*+2\delta'/\lambda$ for all $\delta'\geq\delta$. 
\end{proof}

\begin{lemma}\label{lem.mis} Under assumptions \hyperref[A2]{(A2)}, \hyperref[A3]{(A3)}, we have for all $\delta>0$,
	\begin{enumerate}
		\item[(i)] If $\|\bDelta\|_{L_2(P_X)} \leq 4 \nu_\tau(\delta)$, and $\bDelta \in \cK(\bGamma_\tau;2\delta/\lambda)$, then $Q_\tau(\bGamma_\tau+\bDelta)-Q_\tau(\bGamma_\tau) \geq \frac{1}{4}\fline \|\bDelta\|_{L_2(P_X)}^2$;
		\item[(ii)] If $\bDelta \in \cK(\bGamma_\tau;2\delta/\lambda)$, $\|\bDelta\|_* \leq 4\sqrt{\frac{rm}{\sigma_{\min}(\bSigma_X)}} \|\bDelta\|_{L_2(P_X)}+2\delta/\lambda$, where $r = \rank(\bGamma_\tau)$. 
	\end{enumerate}
\end{lemma}

\begin{proof}[Proof for Lemma \ref{lem.mis}]
\begin{enumerate}
	\item Let $Q_{\tau,j}(\bGamma_{\tau,\ast j}) = \E[\rho_\tau(Y_{ij}-\bX_i^\top \bGamma_{\tau,*j})]$. From Knight's identity \citep{K98}, for any $v,u\in\R$,
	\begin{align}
	\rho_\tau(u-v)-\rho_\tau(u) = -v \psi_\tau(u)+ \int_0^v \big(\IF\{u\leq z\}-\IF\{u\leq 0\}\big)dz.\label{eq:knight}
	\end{align}
	where $\psi_\tau(u)\defeq \tau-\IF(u \leq 0)$. Putting $u = Y_{ij}-\bX_i^\top \bGamma_{\tau,*j}$ in \eqref{eq:knight}, and $v = \bX_i^\top \bDelta_{*j}$, $\E[-v \psi_\tau(u)]=0$ for all $j$ and $i$, by the definition of $\bGamma_\tau=\arg\min_{\Sb}\E[\hQ_\tau(\Sb)]$. Therefore, using law of iterative expectation and mean value theorem, we have by \hyperref[A3]{(A3)} that
	\begin{align}
		&\hspace{-0.5cm}Q_{\tau,j}(\bGamma_{\tau,\ast j}+\bDelta_{\ast j})-Q_{\tau,j}(\bGamma_{\tau,\ast j})\notag\\
		&= \E\bigg[\int_0^{\bX_i^\top \bDelta_{\ast j}} F_{Y_j|\bX_i}(\bX_i^\top \bGamma_{\tau,\ast j}+z|\bX_i)-F_{Y_j|\bX_i}(\bX_i^\top \bGamma_{\tau,\ast j}|\bX_i)dz\bigg]\notag\\
		&= \E\bigg[ \int_0^{\bX_i^\top \bDelta_{\ast j}} z f_{Y_j|\bX_i}(\bX_i^\top \bGamma_{\tau,\ast j}|\bX_i)+\frac{z^2}{2}f_{Y_j|\bX_i}'(\bX_i^\top \bGamma_{\tau,\ast j}+z^\dag|\bX_i) dz\bigg] \notag\\
		&\geq \fline\frac{\E\big[(\bX_i^\top \bDelta_{\ast j})^2\big]}{4}+\fline\frac{\E\big[(\bX_i^\top \bDelta_{\ast j})^2\big]}{4}-\frac{1}{6}\bar f'\E[|\bX_i^\top \bDelta_{*j}|^3]\label{eq:minorize}
		\end{align}
	for $z^\dag \in [0,z]$. Now, for $\bDelta \in \cK(\bGamma_\tau;2\delta/\lambda)$, the condition 
	$$
	\|\bDelta\|_{L_2(P_X)} \leq 4\nu_\tau(\delta) = \frac{3}{2}\frac{\fline}{\bar f'} \inf_{\bDelta \in \cK(\bGamma_\tau;2\delta/\lambda)\atop\bDelta \neq 0} \frac{\big(\sum_{j=1}^m \E[|\bX_i^\top\bDelta_{*j}|^2]\big)^{3/2}}{\sum_{j=1}^m \E[|\bX_i^\top\bDelta_{*j}|^3]}
	$$ 
	implies
	\begin{align*}
	\fline m^{-1}\sum_{j=1}^m\frac{\E\big[(\bX_i^\top \bDelta_{\ast j})^2\big]}{4} \geq \frac{1}{6}\bar f'm^{-1}\sum_{j=1}^m\E[|\bX_i^\top \bDelta_{*j}|^3]
	\end{align*}
	Therefore, 
	$$
	Q_\tau(\bGamma_\tau+\bDelta)-Q_\tau(\bGamma_\tau) \geq \fline m^{-1} \sum_{j=1}^m \frac{\E(\bX_i^\top \bDelta_{\ast j})^2}{4} = \frac{1}{4} \fline \|\bDelta\|_{L_2(P_X)}^2.
	$$
	\item By the decomposability of nuclear norm, $\bDelta \in \cK(\bGamma_\tau;2\delta/\lambda)$ and \eqref{eq:pfrobeq} in Remark \ref{rem:normeq}, we can estimate
	\begin{align*}
		\|\bDelta\|_* &= \|\cP_{\bGamma_\tau}(\bDelta)\|_* + \|\cP_{\bGamma_\tau}^\perp(\bDelta)\|_* \leq 4 \|\cP_{\bGamma_\tau}(\bDelta)\|_* +2\delta/\lambda\leq 4 \sqrt{r} \|\cP_{\bGamma_\tau}(\bDelta)\|_{\rm F}+2\delta/\lambda \\
		&\leq 4\sqrt{\frac{rm}{\sigma_{\min}(\bSigma_X)}} \|\bDelta\|_{L_2(P_X)}+2\delta/\lambda.
	\end{align*}
\end{enumerate}
\end{proof}

\begin{lemma}\label{lem.Q}
Under Assumptions \hyperref[A1]{(A1)}-\hyperref[A3]{(A3)}, recall that $\cA(t)$ is defined in \eqref{eq:cA}, then for an arbitrary $u>1$,
$$
\P\Big\{\cA(t) \leq 8 \sqrt{2} u (\alpha_r t+2m^{-1/2} \delta/\lambda) \sqrt{(c_2\sigma_{\max}(\bSigma_X)+B_p)} \sqrt{\log m+\log p}\Big\} \geq 1-16(pm)^{1-u^2}-\gamma_n,
$$
where $\alpha_r= 4\sqrt{r/\sigma_{\min}(\bSigma_X)}$ and $r=\rank(\bGamma_\tau)$. 
\end{lemma}

\begin{proof}[Proof for Lemma \ref{lem.Q}]
To simplify notations, let 
\begin{align}
\alpha_r\defeq 4\sqrt{r/\sigma_{\min}(\bSigma_X)}\label{eq:alphar}	
\end{align}
Let $\{\vep_{ij}\}_{i\leq n,j\leq m}$ be independent Rademacher random variables independent from $Y_{ij}$ and $\bX_i$ for all $i,j$. Denote $\P_{\vep}$ and $\E_{\vep}$ as the conditional probability and the conditional expectation with respect to $\{\vep_{ij}\}_{i\leq n,j\leq m}$, given $Y_{ij}$ and $\bX_i$. Denote
\begin{align}
	\chi_{ij}^\tau(\cdot)\defeq \rho_\tau\{Y_{ij}-\bX_i^\top\bGamma_{\tau,\ast j}-\cdot\}-\rho_\tau\{Y_{ij}-\bX_i^\top \bGamma_{\tau,\ast j}\}. \label{eq:defchi}
\end{align}
$\chi_{ij}^\tau(\cdot)$ is a contraction in the sense that $\chi_{ij}^\tau(0)=0$, and for all $a,b\in\R$, 
\begin{align}
	\big|\chi_{ij}^\tau(a)-\chi_{ij}^\tau(b)\big| \leq |a-b|. \quad \forall i=1,...,n, \ j=1,...,m.\label{eq:contra}
\end{align}

First, we note that for any $\bDelta$ satisfying $\bDelta\in \cK(\bGamma_\tau;2\delta/\lambda)$ and $\|\bDelta\|_{L_2(P_X)}\leq t$,
\begin{align}
	&\var\bigg(\GG_n\bigg(m^{-1}\sum_{j=1}^m \chi_{ij}^\tau(\bX_i^\top\bDelta_{*j})\bigg)\bigg)\notag\\
	&= \var\bigg(m^{-1}\sum_{j=1}^m \chi_{ij}^\tau(\bX_i^\top\bDelta_{*j})\bigg) \leq m^{-1} \sum_{j=1}^m \E\big[(\chi_{ij}^\tau(\bX_i^\top\bDelta_{*j}))^2\big]\notag\\ 
	&\leq m^{-1} \sum_{j=1}^m \E\big[(\bX_i^\top\bDelta_{*j})^2\big]
	\leq t^2, \label{eq:var}
\end{align}
where the first equality and the second inequality follows from elementary computations and i.i.d. assumption \hyperref[A1]{(A1)}, the third inequality is a result of \eqref{eq:contra}, and the last inequality applies \eqref{eq:frobeq} in Remark \ref{rem:normeq}. 

To apply Lemma 2.3.7 of \cite{VW:1996}, we observe from Chebyshev's inequality that for any $s>0$,
\begin{align*}
	&\inf_{\|\bDelta\|_{L_2(P_X)} \leq t, \bDelta \in \cK(\bGamma_\tau;2\delta/\lambda)} \P\bigg(\bigg|\GG_n\bigg(m^{-1}\sum_{j=1}^m \chi_{ij}^\tau(\bX_i^\top\bDelta_{*j})\bigg)\bigg|< \frac{s}{2}\bigg) \\
	&= 1-\sup_{\|\bDelta\|_{L_2(P_X)} \leq t, \bDelta \in \cK(\bGamma_\tau;2\delta/\lambda)} \P\bigg(\bigg|\GG_n\bigg(m^{-1}\sum_{j=1}^m \chi_{ij}^\tau(\bX_i^\top\bDelta_{*j})\bigg)\bigg| \geq \frac{s}{2}\bigg)\geq 1-4\frac{t^2}{s^2}.
\end{align*}
Taking $s \geq \sqrt{8} t$, we have 
$$
\frac{1}{2} \leq \inf_{\|\bDelta\|_{L_2(P_X)} \leq t, \bDelta \in \cK(\bGamma_\tau;2\delta/\lambda)} \P\bigg(\bigg|\GG_n\bigg(m^{-1}\sum_{j=1}^m \chi_{ij}^\tau(\bX_i^\top\bDelta_{*j})\bigg)\bigg|< \frac{s}{2}\bigg).
$$
Thus, applying Lemma 2.3.7 of \cite{VW:1996}, we have
\begin{align}
	\P\{\cA(t) > s\} \leq 4 \P\bigg(\sup_{\|\bDelta\|_{L_2(P_X)} \leq t \atop \bDelta \in \cK(\bGamma_\tau;2\delta/\lambda)}\bigg|n^{-1/2}m^{-1}\sum_{i=1}^n\sum_{j=1}^m \vep_{ij} \chi_{ij}^\tau(\bX_i^\top\bDelta_{*j})\bigg|>\frac{s}{4}\bigg). \label{eq:A1}
\end{align}
Now we restrict the $\cA(t)$ on the event $\Omega$ on which \eqref{cond.cov} in \hyperref[A2]{(A2)} holds, with $\P(\Omega) \geq 1-\gamma_n$. Applying Markov's inequality, for an arbitrary constant $\mu>0$, the right-hand side of \eqref{eq:A1} can be bounded by
\begin{align}
	&\P\{\cA(t) > s|\Omega\} \notag\\
	&\leq 4 \exp\bigg(\frac{-\mu s}{4}\bigg) \E\bigg[\E_{\vep}\bigg[\exp\bigg\{\mu\sup_{\|\bDelta\|_{L_2(P_X)} \leq t \atop \bDelta \in \cK(\bGamma_\tau;2\delta/\lambda)}\bigg|n^{-1/2}m^{-1}\sum_{i=1}^n\sum_{j=1}^m \vep_{ij} \chi_{ij}^\tau(\bX_i^\top\bDelta_{*j})\bigg|\bigg\}\bigg]\bigg|\Omega\bigg]. \label{eq:A2}
\end{align}
Now recall \eqref{eq:contra}, the comparison theorem for Rademacher processes (Lemma 4.12 in \cite{LT:1991}) implies the right-hand side of \eqref{eq:A2} is bounded by
\begin{align}
	&\P\{\cA(t) > s|\Omega\} \notag\\
	&\leq 4 \exp\bigg(\frac{-\mu s}{4}\bigg) \E\bigg[\E_{\vep}\bigg[\exp\bigg\{2\mu\sup_{\|\bDelta\|_{L_2(P_X)} \leq t \atop \bDelta \in \cK(\bGamma_\tau;2\delta/\lambda)}\bigg|n^{-1/2}m^{-1}\sum_{i=1}^n\sum_{j=1}^m \vep_{ij} \bX_i^\top\bDelta_{*j}\bigg|\bigg\}\bigg]\bigg|\Omega\bigg]. \label{eq:A3}
\end{align}
To obtain a bound for the right-hand side of \eqref{eq:A3}, we note that 
\begin{align}
	\bigg|\sum_{i=1}^n\sum_{j=1}^m \vep_{ij} \bX_i^\top\bDelta_{*j}\bigg| &= \bigg|\mbox{tr}\Big(\Big[\sum_{i=1}^n \vep_{i1} \bX_i \ \sum_{i=1}^n \vep_{i2} \bX_i\ ...\ \sum_{i=1}^n \vep_{im} \bX_i\Big]^\top \bDelta\Big)\bigg|\notag\\
	&\leq \|\bDelta\|_* \sup_{\ba\in\cS^{p-1}}\bigg|\sum_{j=1}^m \Big(\sum_{i=1}^n \vep_{ij}\bX_i^\top\ba\Big)^2\bigg|^{1/2}\notag\\
	&\leq m^{1/2}\|\bDelta\|_* \max_{j\leq m}\bigg\|\sum_{i=1}^n \vep_{ij}\bX_i\bigg\|,\label{eq:bd1}
\end{align}
where the first inequality is from H\"older's inequality, and the second inequality is elementary. 

Now we apply random matrix theory to bound the right-hand side of \eqref{eq:A3}. Using matrix dilations (see, for example Section 2.6 of \cite{T:11}), we have
\begin{align}
	\bigg\|\sum_{i=1}^n \vep_{ij}\bX_i\bigg\| =\bigg\| \sum_{i=1}^n \vep_{ij}\begin{pmatrix}
	\IO_p &\bX_i\\
	\bX_i^\top &0 
 \end{pmatrix}\bigg\|. \label{eq:dil}
\end{align}
Notice that the random matrix $\vep_{ij}\begin{pmatrix}
	\IO_p &\bX_i\\
	\bX_i^\top &0 
 \end{pmatrix}$ is self adjoint and symmetrically distributed conditional on $\bX_i$. We now obtain
 \begin{align}
 	&\E_{\vep}\bigg[\exp\bigg\{2\mu\sup_{\|\bDelta\|_{L_2(P_X)} \leq t \atop \bDelta \in \cK(\bGamma_\tau;2\delta/\lambda)}\bigg|n^{-1/2}m^{-1}\sum_{i=1}^n\sum_{j=1}^m \vep_{ij} \bX_i^\top\bDelta_{*j}\bigg|\bigg\}\bigg]\notag\\
	&\leq \E_{\vep}\bigg[\exp\bigg\{2\mu (\alpha_r t + m^{-1/2}2\delta/\lambda) \max_{j \leq m} \bigg\|n^{-1/2}\sum_{i=1}^n \vep_{ij} \bX_i^\top\bigg\|\bigg\}\bigg]\notag\\
	&\leq m \max_{j \leq m} \E_{\vep}\bigg[\exp\bigg\{2\mu (\alpha_r t + m^{-1/2}2\delta/\lambda) \bigg\|n^{-1/2}\sum_{i=1}^n \vep_{ij} \begin{pmatrix}
	\IO_p &\bX_i\\
	\bX_i^\top &0 
 \end{pmatrix}\bigg\|\bigg\}\bigg]\notag\\
 	&\leq m 2 (p+1) \max_{j \leq m}  \exp\bigg\{\sigma_{\max}\bigg(\sum_{i=1}^n \log \E_{\vep} \bigg[\exp\bigg\{2\mu (\alpha_r t + m^{-1/2}2\delta/\lambda) n^{-1/2}\vep_{ij} \begin{pmatrix}
	\IO_p &\bX_i\\
	\bX_i^\top &0 
 \end{pmatrix}\bigg\}\bigg]\bigg)\bigg\}\label{eq:bd2}
 \end{align}
 where the first inequality is from Lemma \ref{lem.mis}(ii) and \eqref{eq:bd1} and recall $\alpha_r$ in \eqref{eq:alphar}, the second inequality follows from \eqref{eq:dil}, Lemma \ref{lem.mis} (ii) ($\bDelta\in\cK(\bGamma_\tau;2\delta/\lambda)$), and the fact that 
 \begin{align*}
 	\E[\max_{j\leq m} \exp(|Z_j|)] &\leq m \max_{j \leq m}\E[\exp(|Z_j|)], \ \mbox{ for any random variable } Z_j\in\R.
\end{align*}
 The third inequality is by Theorem 3(ii) of \cite{MP13} by the symmetric distribution of $\vep_{ij}$, where for a self adjoint matrix $\Ab$,
 \begin{align*}
 	\exp(\Ab) &\defeq \Ib + \sum_{j=1}^\infty \frac{\Ab^j}{j!}\\
	\log(\exp(\Ab)) &\defeq \Ab.
 \end{align*}
From equation (2.4) on page 399 of \cite{T:11}, for any $j$ and $c>0$,
\begin{align*}
	\E_{\vep} \bigg[\exp\bigg\{c\ \vep_{ij} \begin{pmatrix}
		\IO_p &\bX_i\\
		\bX_i^\top &0 
	 \end{pmatrix}\bigg\}\bigg] &= \frac{1}{2}\bigg(\exp\bigg\{c\begin{pmatrix}
		\IO_p &\bX_i\\
		\bX_i^\top &0 
	 \end{pmatrix}\bigg\}+\exp\bigg\{-c\begin{pmatrix}
		\IO_p &\bX_i\\
		\bX_i^\top &0 
	 \end{pmatrix}\bigg\}\bigg)\\
	 &\preccurlyeq \exp\bigg\{\frac{c^2}{2}\begin{pmatrix}
		\bX_i\bX_i^\top &\IO_p\\
		 0 &\bX_i^\top\bX_i 
	 \end{pmatrix}\bigg\},
\end{align*}
where "$\Ab \preccurlyeq \Bb$" means the $\Bb-\Ab$ is positive semidefinite for two matrices $\Ab,\Bb$. From equation (2.8) on page 399 of \cite{T:11}, the logarithm defined above preserves the order $\preccurlyeq$. Hence, \eqref{eq:bd2} is bounded by
\begin{align}
	&2m (p+1)  \exp\bigg\{2 \mu^2 (\alpha_r t + m^{-1/2}2\delta/\lambda)^2 \sigma_{\max}\bigg(n^{-1}\sum_{i=1}^n \begin{pmatrix}
		\bX_i\bX_i^\top &\IO_p\\
		 0 &\bX_i^\top\bX_i 
	 \end{pmatrix}\bigg)\bigg\}\notag\\
	 &\leq 2m(p+1)\exp\big\{2 \mu^2 (\alpha_r t + m^{-1/2}2\delta/\lambda)^2 (\sigma_{\max}(\hat\bSigma_X)+B_p)\big\},\label{eq:bd3}
\end{align}
where the last inequality follows from a bound for the spectral norm for block matrices in equation (2) of Theorem 1 in \cite{BK90} (with Shatten-$\infty$ norm), and Assumption \hyperref[A2]{(A2)}. 

Putting \eqref{eq:bd3} into \eqref{eq:A2}, we obtain
\begin{align}
	\P\{\cA(t) > s|\Omega\} 
	&\leq 8m(p+1) \exp\bigg(\frac{-\mu s}{4}\bigg) \E\big[ \exp\big\{2 \mu^2 (\alpha_r t + m^{-1/2}2\delta/\lambda)^2 (\sigma_{\max}(\hat\bSigma_X)+B_p)\big\}\big|\Omega\big]\notag\\
	&\leq 8m(p+1) \exp\bigg(\frac{-\mu s}{4}\bigg) \exp\big\{2 \mu^2 (\alpha_r t + m^{-1/2}2\delta/\lambda)^2 (c_2\sigma_{\max}(\bSigma_X)+B_p)\big\}.\label{eq:bd4}
\end{align}
Minimizing the expression \eqref{eq:bd4} with respect to $\mu$ gives
\begin{align}
	\P\{\cA(t) > s|\Omega\} 
	&\leq 8m(p+1) \exp\bigg\{-\frac{s^2}{128 (\alpha_r t + m^{-1/2}2\delta/\lambda)^2 (c_2\sigma_{\max}(\bSigma_X)+B_p)}\bigg\}.
\end{align}
Taking 
\begin{align}
s &=8 \sqrt{2} u (\alpha_r t + m^{-1/2}2\delta/\lambda)\sqrt{(c_2\sigma_{\max}(\bSigma_X)+B_p)} \sqrt{\log m+\log p} \label{eq:ss}
\end{align}
Notice that $s \geq \sqrt{8}t$ for large enough $p,m$, so the symmetrization \eqref{eq:A1} is valid. Recall that $\P(\Omega) \geq 1-\gamma_n$. The proof is then completed.
\end{proof}

\begin{remark}\label{rem:noniid}
The Lemma 2.3.7 of \cite{VW:1996} and Lemma 4.12 of \cite{LT:1991} applied in the proof of Lemma \ref{lem.Q} require only independence in the random variables $(Y_{ij},\bX_i)$, without needing identical distribution. The random matrix theory applied in the proof may also be generalized to matrix martingales; see Section 7 of \cite{T:11} for more details. 
\end{remark}

\begin{remark}\label{rem:uniftau} It can be observed that Lemma \ref{lem.Q} is valid uniformly for any $0<\tau < 1$.
\end{remark}

\subsection{Proof of Theorem \ref{th:facloa}}\label{sec:facloa}
In this proof, we abbreviate $\sigma_k(\bGamma_\tau)$, $\sigma_k(\tG)$, $(\tilde\Vb_\tau)_{*k}$ and $(\Vb_\tau)_{*k}$, $(\tilde\Ub_\tau)_{*k}$ and $(\Ub_\tau)_{*k}$ by $\sigma_k$, $\tilde\sigma_k$, $\tilde\Vb_{*k}$ and $\Vb_{*k}$, $\tilde\Ub_{*k}$ and $\Ub_{*k}$.

To prove \eqref{eq:vbound}, since $\Psi_\tau=\Vb_\tau$ and $\hat \Psi_\tau = \tilde\Vb_\tau$, by Theorem 3 of \cite{YWS15}, 
\begin{align}
	\sin \cos^{-1}(|\tilde\Vb_{*j}^\top \Vb_{*j}|)\leq \frac{2(2\|\bGamma_\tau\|+\|\tG-\bGamma_\tau\|_{\rm F})\|\tG-\bGamma_\tau\|_{\rm F}}{\min\{\sigma_{j-1}^2(\bGamma_\tau)-\sigma_j^2(\bGamma_\tau),\sigma_{j}^2(\bGamma_\tau)-\sigma_{j+1}^2(\bGamma_\tau)\}}\label{eq:vbound.temp}	
\end{align}
where by the fact that $|\tilde\Vb_{*j}^\top \Vb_{*j}| \leq 1$,
\begin{align*}
	\sin \cos^{-1}(|\tilde\Vb_{*j}^\top \Vb_{*j}|) &= \sqrt{1-(\tilde\Vb_{*j}^\top \Vb_{*j})^2} = \sqrt{(1-\tilde\Vb_{*j}^\top \Vb_{*j})(1+\tilde\Vb_{*j}^\top \Vb_{*j})}\\
	&\geq \sqrt{(1-|\tilde\Vb_{*j}^\top \Vb_{*j}|)^2}=1-\big|\tilde\Vb_{*j}^\top \Vb_{*j}\big|.
\end{align*}
Similar bound like \eqref{eq:vbound} also holds for $\tilde\Ub_{*j}$, by the discussion below Theorem 3 of \cite{YWS15}.

For a proof for inequality \eqref{eq:bfac}, by direct calculation,
\begin{align}
	\big|\hat f_k^\tau(\bX_i)-f_k^\tau(\bX_i)\big| &= \big|\tilde\sigma_k \tilde\Ub_{*k}^\top \bX_i-\sigma_k \Ub_{*k}^\top \bX_i\big|\notag\\
	&\leq \big\|\tilde\sigma_k \tilde\Ub_{*k}^\top-\sigma_k \Ub_{*k}^\top\big\|\|\bX_i\|\notag\\
	&\leq \big(\big|\tilde\sigma_k-\sigma_k\big|\big\|\tilde\Ub_{*k}\big\|+\sigma_k \big\|\tilde\Ub_{*k}-\Ub_{*k}\big\|\big)\|\bX_i\|\notag\\
	&\leq \Big(\big|\tilde\sigma_k-\sigma_k\big|+\sigma_k \sqrt{(\tilde\Ub_{*k}-\Ub_{*k})^\top(\tilde\Ub_{*k}-\Ub_{*k})}\Big)\|\bX_i\|\notag\\
	&\leq \Big(\big|\tilde\sigma_k-\sigma_k\big|+\sigma_k \sqrt{2(1-\tilde\Ub_{*k}^\top\Ub_{*k})}\Big)\|\bX_i\| \label{eq:facbd}
\end{align}
where we apply the fact that $\|\tilde\Ub_{*k}\big\|=1$. By assumption $\tilde\Ub_{*k}^\top\Ub_{*k}\geq 0$, $\tilde\Ub_{*k}^\top\Ub_{*k}=|\tilde\Ub_{*k}^\top\Ub_{*k}|$. Apply Lemma \ref{lem:mirsky} and the bound \eqref{eq:vbound.temp} with $\Vb$ being replaced by $\Ub$ to \eqref{eq:facbd}, then \eqref{eq:bfac} is proved. Thus, the proof for this theorem is completed. \hfill$\qed$


\vskip 2em \centerline{\Large \bf S.3: Miscellaneous Technical Details}
\setcounter{subsection}{0}
\renewcommand{\thesubsection}{S.3.\arabic{subsection}}
\setcounter{equation}{0}
\renewcommand{\theequation}{S.3.\arabic{equation}}
\setcounter{theorem}{0}
\renewcommand{\thetheorem}{S.3.\arabic{theorem}}
\vskip 1em

\subsection{Detail on Remark \ref{rmk:nu}}\label{sec:nu}

For \eqref{eq:gr} to hold, it is enough to have $\E[|\bX_i^\top\bDelta_{*j}|^3]\leq C \E[|\bX_i^\top\bDelta_{*j}|^2]^{3/2}$ for all $j=1,2,...,m$, where $C>0$ is a constant independent of $j$, because
\begin{align}
	\Big(\sum_{j=1}^m \E[|\bX_i^\top\bDelta_{*j}|^2]^{3/2}\Big)^{2/3} \leq \sum_{j=1}^m \E[|\bX_i^\top\bDelta_{*j}|^2]
\end{align}
by the inequality $\|\ba\|_{3/2} \leq \|\ba\|_1$ for an arbitrary $\ba=(a_1,a_2,...,a_m)$ with $a_j\geq 0,\ \forall j$. If $\bX_i$ is i.i.d. sampled from a \emph{log-concave} density, then Theorem 5.22 of \cite{LV07} implies $\E[|\bX_i^\top\bDelta_{*j}|^3]\leq 3^{3/2} \E[|\bX_i^\top\bDelta_{*j}|^2]^{3/2}$ for any $\bDelta$. See also Design 1 on p.2 of the supplemental materials of \cite{BC:2011}. This implies \eqref{eq:gr} as $\epsilon_{n,\tau,r}$ is small as $n \gtrsim B_p r(p+m)(\log p+\log m)$.

\subsection{Detail on Remark \ref{rmk:notlowrank}}\label{sec:notlowrank}
We need some extra notations. Let $\cV\subset\R^m$ and $\cU\subset\R^p$ be two subspaces with dimension $r$, let $\cM = \{\bDelta\in\R^{p\times m}:\mbox{row space of }\bDelta\subset\cV,\mbox{ column space of }\bDelta\subset\cU\}$; $\overline\cM^\perp = \{\bDelta\in\R^{p\times m}:\mbox{row space of }\bDelta\subset\cV^\perp,\mbox{ column space of }\bDelta\subset\cU^\perp\}$ (defined similarly as in Example 3 on page 542 of \cite{NRWY:12}). For any matrix $\Sb\in\R^{p\times m}$,
	\begin{align*}
		\p_{\cM}(\Sb) &=  \Pb_{\cU}\Sb \Pb_{\cV},\quad \p_{\overline\cM}^\perp(\Sb) = \Pb_{\cU}^\top\Sb\Pb_{\cV}^\top,
	\end{align*}
	where $\Pb_{\cV} = \Vb\Vb^\top$, $\Pb_{\cV}^\perp = \Ib_{m\times r}-\Pb_{\cV}$, $\Vb=[\vb_1\,...\,\vb_r]$, and $\{\vb_j\}_{j=1}^r$ is a set of orthonormal basis for $\cV$; analogously, $\Pb_{\cU} = \Ub\Ub^\top$, $\Pb_{\cU}^\perp = \Ib_{p\times r}-\Pb_{\cU}$, $\Ub=[\ub_1\,...\,\ub_r]$, and $\{\ub_j\}_{j=1}^r$ is a set of orthonormal basis for $\cU$. Moreover, for any $\Sb\in\R^{p\times m}$, $\|\Sb\|_*=\|\p_{\cM}(\Sb)\|_*+\|\p_{\overline\cM}^\top(\Sb)\|_*$.

	It can be shown that when $\lambda \geq 2\|\nabla \hQ(\bGamma_\tau)\|$, the difference $\tdel=\tG-\bGamma_\tau$ lies in the set
	\begin{align}
		&\cK(\overline\cM,4 \|\p_{\overline\cM}^\perp(\bGamma_\tau)\|+2 \delta'/\lambda)\notag\\
		&\quad\quad\quad\defeq\bigg\{\bDelta\in\R^{p\times m}: \|\p_{\overline\cM}^\perp(\bDelta)\|\leq 3\|\p_{\overline\cM}(\bDelta)\|+4 \|\p_{\cM}^\perp(\bGamma_\tau)\|+\frac{2 \delta'}{\lambda}\bigg\},\label{eq:notlowrank_nsp}
	\end{align}
	where $\delta'\geq\delta$. Under this situation, the recovery property of $\tG$ can be shown via similar argument as for Theorem \ref{thm:rec} (possibly under more restrictive conditions), and we leave out the details. 

To show \eqref{eq:notlowrank_nsp}, we first note an inequality
\begin{align}
	\|\tG\|_*-\|\bGamma_\tau\|_* \leq 2 \|\p_{\cM}^\perp(\bGamma_\tau)\|_* + \|\p_{\overline\cM}(\tdel)\|_* - \|\p_{\overline\cM}^\perp(\tdel)\|_*,\label{eq:notlowrank52}
\end{align}
which can be shown by exactly the same argument for showing inequality (52) in Lemma 3 on page 27 in the supplementary material of \cite{NRWY:12}, because the nuclear norm is decomposable with respect to $(\cM,\overline\cM^\perp)$. 

It can be seen that from similar argument as \eqref{eq:subgradient},
	\begin{align}
		0 &\leq \hQ_{\tau}(\bGamma_\tau)-\hQ_{\tau}(\Gamat)+\lambda\|\bGamma_\tau\|_*-\lambda\|\Gamat\|_* + \delta\notag\\
		&\leq \|\nabla \hQ_\tau(\bGamma_\tau) \| \big(\|\cP_{\overline\cM}(\tdel)\|_*+\|\cP_{\overline\cM}^\perp(\tdel)\|_*\big) \notag\\
		&\quad\quad\quad\quad\quad\quad+ \lambda (2 \|\p_{\cM}^\perp(\bGamma_\tau)\|_* + \|\p_{\overline\cM}(\tdel)\|_* - \|\p_{\overline\cM}^\perp(\tdel)\|_*)+\delta, \label{eq:subgra_notlowrank}
	\end{align}
	where the second inequality is from \eqref{eq:notlowrank52}. Rearrange expression \eqref{eq:subgra_notlowrank} to get,
	$$
	(\lambda-\|\nabla \hQ_\tau(\bGamma_\tau) \|) \|\cP_{\overline\cM}^\perp(\hDelta)\|_* \leq (\lambda+\|\nabla \hQ_\tau(\bGamma_\tau) \|) \|\cP_{\overline\cM}(\hDelta)\|_*+2\lambda\|\p_{\cM}^\perp(\bGamma_\tau)\|_*+\delta.
	$$	
	By $\lambda \geq 2 \|\nabla \hQ_\tau(\bGamma_\tau) \|$,
	$$
	\frac{1}{2} \lambda \|\cP_{\overline\cM}^\perp(\hDelta)\|_* \leq \frac{3}{2} \lambda \|\cP_{\overline\cM}(\hDelta)\|_*+2\lambda\|\p_{\cM}^\perp(\bGamma_\tau)\|_*+\delta.
	$$

\subsection{Details for Generating matrices $\Sb_1$ and $\Sb_2$ in Section \ref{sec.si}}\label{sec:s1s2}
Given $(r_1,r_2)$, $\Sb_1$ and $\Sb_2$ are selected with the following procedure:
\begin{enumerate}
	\item Generate vectors $\{\ba_1,...,\ba_{r_1}\}$ and $\{\bb_1,...,\bb_{r_2}\}$, where $\ba_{j_1},\bb_{j_2}\in\R^p$, and $a_{j_1 k_1}, b_{j_2 k_2} \sim U(0,1)$ i.i.d. for $j_1=1,...,r_1$, $j_2=1,...,r_2$, $k_1,k_2=1,...,p$;
	\item Set the columns of $\Sb_1$ and $\Sb_2$ by $(\Sb_1)_{\ast j} = \sum_{k=1}^{r_1} \alpha_{k,j} \ba_k$ and $(\Sb_2)_{\ast j}=\sum_{k=1}^{r_2} \beta_{k,j} \bb_{k}$ for $j=1,...,m$, where $\alpha_{k,j}$, $\beta_{k,j}$ are independent random variables in $U[0,1]$ for $k=1,...,p$ and $j=1,...,m$. 
\end{enumerate}

In our simulation, the first two nonzero singular values for $\Sb_1$ are $(\sigma_1(\Sb_1),\sigma_2(\Sb_1))=(179.91,26.51)$ and the rest singular value is 0. For $\Sb_2^{Sym}$, the first two nonzero singular values are $(\sigma_1(\Sb_2^{Sym}),\sigma_2(\Sb_2^{Sym}))=(175.48,25.74)$ and the rest is 0. For $\Sb_2^{Sym}$, the first six nonzero singular values are $(\sigma_1(\Sb_2^{Asym}),...,\sigma_{6}(\Sb_2^{Asym}))=(473.40,29.87,25.66,23.89,23.58,22.16)$ and the rest is 0.

	
\vskip 2em \centerline{\Large \bf S.4: Auxiliary Lemmas} \vskip -1em
\setcounter{subsection}{0}
\renewcommand{\thesubsection}{S.4.\arabic{subsection}}
\setcounter{equation}{0}
\renewcommand{\theequation}{S.4.\arabic{equation}}
\setcounter{theorem}{0}
\renewcommand{\thetheorem}{S.4.\arabic{theorem}}

\begin{defin}\label{def_prox}
Let $\mathcal X=\R^{p \times n}$ with inner product $\langle \Ab,\Bb
\rangle = \mbox{tr}(\Ab^\top \Bb)$ and $\|\cdot\|$ be the induced norm.
$f: \mathcal X \rightarrow \R$ a lower semicontinuous convex
function. The \emph{proximity operator of $f$}, $S_f: \mathcal X
\rightarrow \mathcal X$:
\begin{align*}
&S_f(\Yb) \defeq \operatorname{arg}\,\underset{\Xb \in \mathcal X}{\operatorname{min}} \bigg\{f(\Xb)+\frac{1}{2} \|\Xb-\Yb\|^2\bigg\}, \forall \Yb \in \mathcal X.
\end{align*}
\end{defin}
\begin{theorem}[Theorem 2.1 of \cite{SVT:2010}]\label{prox_nu}
Suppose the singular decomposition of $\Yb=\Ub\Db \Vb^\top \in \R^{p \times m}$, where $\Db$ is a $p \times m$ rectangular diagonal matrix and $\Ub$ and $\Vb$ are unitary matrices. The proximity operator $S_\lambda(\cdot)$ associated with $\lambda \|\cdot\|_{*}$ is
\begin{align}
S_{\lambda}(\Yb) \defeq \Ub (\Db-\lambda\Ib_{pm})_+ \Vb^\top,
\end{align}
where $\Ib_{pm}$ is the $p \times m$ rectangular identity matrix with diagonal elements equal to 1.
\end{theorem}


\begin{lemma}[Hoeffding's Inequality, Proposition 5.10 of \cite{V:2012}]\label{lem.hoef}
Let $X_1,...,X_n$ be independent centered sub-gaussian random variables, and let $K = \max_i \|X_i\|_{\psi_2}$. Then for every $\ab=(a_1,...,a_n)^\top \in \R^n$ and every $t \geq 0$, we have
\begin{align*}
	\P\bigg(\bigg|\sum_{i=1}^n a_i X_i\bigg|\geq t\bigg) \leq e \cdot \exp\bigg(-\frac{C't^2}{K^2 \|\ab\|_2^2}\bigg),
\end{align*}
where $C'>0$ is a universal constant.
\end{lemma}

\begin{lemma}[Hoeffding's Inequality: classical form]\label{lem.hoefc}
	Let $X_1,...,X_n$ be independent random variables such that $X_i \in [a_i,b_i]$ almost surely, then 
	\begin{align*}
		\P\bigg(\bigg|\sum_{i=1}^n X_i\bigg|\geq t\bigg) \leq 2 \exp\bigg(-\frac{2t^2}{\sum_{i=1}^n (b_i-a_i)^2}\bigg).
	\end{align*}
\end{lemma}

\vskip 2em \centerline{\Large \bf S.5: Selecting the Matrix $\Bb$ in Section 4.3} \vskip -1em
\setcounter{subsection}{0}
\renewcommand{\thesubsection}{S.5.\arabic{subsection}}
\setcounter{equation}{0}
\renewcommand{\theequation}{S.5.\arabic{equation}}
\setcounter{theorem}{0}
\renewcommand{\thetheorem}{S.5.\arabic{theorem}}
\setcounter{figure}{0}
\renewcommand{\thefigure}{S.5.\arabic{figure}}
\setcounter{table}{0}
\renewcommand{\thetable}{S.5.\arabic{table}}
\vskip 1em

The $\Bb$ in \eqref{eq:ts_x} is the coefficient estimator obtained by fitting a VAR(1) model \citep{L05} to the $\bX_i$ in \eqref{input.vari}, and $\Sigma_{\sbvep}$ is the sample covariance matrix from the residuals. Due to the high dimensionality ($460$), the VAR model may be over-parameterized especially when the order is high, and straightforwardly estimating the VAR may yield unreliable estimates. Therefore, as suggested by multiple authors (e.g. \cite{DZZ16, NMB17} and the references therein), we estimate the VAR model with the $\ell_1$ norm penalty, or Lasso \citep{T96}, to alleviate the problem of over-parameterization. Henceforth, the VAR model estimated with the Lasso penalty will be called Lasso-VAR. The computation can be carried out with the \texttt{R} package \texttt{BigVAR} \citep{NMB17}. The Lasso tuning parameter is selected optimally by the cross-validation procedure provided in the package. 

To evaluate the adequacy of the VAR(1) model for the real data in \eqref{input.vari}, Table \ref{tab:var} provides the 1-step-ahead mean square forecasting error (MSFE) of Lasso-VAR (see Eq. (12) of \cite{NMB17}) with different lags. As it requires excessive computational time and resource for model estimation and cross-validation, the maximal order under consideration here is three. Lasso-VAR(3) has the smallest MSFE, but the difference between the models seems small, so we take Lasso-VAR(1). The MSFE of VAR with order selected by AIC or BIC \citep{L05,NMB17} is 2805 with optimal order of the both being 0, which is higher than that of Lasso-VAR as shown in Table \ref{tab:var}. 

For a simple diagnosis of Lasso-VAR(1), we check the autocorrelation and partial autocorrelation function of each individual residual series. Autocorrelation and partial autocorrelation functions of some series are significant. However, increasing the order of the VAR model does not improve the situation. To our best knowledge, we are not aware of any literature on vector ARIMA models for high dimensional time series, which might provide a better fit of our data. Fitting a very high dimensional VAR like ours is very subtle. As the Lasso-VAR(1) has demonstrated competent forecasting performance as shown in Table \ref{tab:var}, we adopt Lasso-VAR(1). A full exploration of the time series structure of the data is left for future research.

\begin{table}[htb]        
	\begin{center}        
		\begin{tabular}{lrrr}
			\hline\hline
			Order &1 &2 &3 \\
			\hline
			Lasso-VAR MSFE & 2364.82 & 2353.075 & 2341.046\\
			\% of active coef. &3.9 &2.836 &2.381 \\
			\hline
			\multicolumn{4}{c}{MSFE of VAR-AIC/BIC (optimal order = 0):  2805}\\
			\hline\hline
		\end{tabular}        
	\end{center}		  
	\caption{The mean square forecasting error (MSFE) and the percentage of active coefficients (total number of coefficients = $460\times (1 + 460 \times $ order)) with different orders, where ``1'' is from the intercept. For the matrix $\B$, we do not include the intercept.}\label{tab:var}        
\end{table}

\vskip 2em \centerline{\Large \bf S.6: Additional Numerical Results: AR(1) Model} \vskip -1em
\setcounter{subsection}{0}
\renewcommand{\thesubsection}{S.6.\arabic{subsection}}
\setcounter{equation}{0}
\renewcommand{\theequation}{S.6.\arabic{equation}}
\setcounter{theorem}{0}
\renewcommand{\thetheorem}{S.6.\arabic{theorem}}
\setcounter{figure}{0}
\renewcommand{\thefigure}{S.6.\arabic{figure}}
\vskip 1em

In this section, we consider the same data generating model as \eqref{het.model} in Section \ref{sec:esterr}, but now the regressor $\bX_i$ follows an AR(1) model
\begin{align}
	\bX_i = 0.5 \bX_{i-1} + \bu_i, \label{eq:dep_regressor}
\end{align}
where $\bu_i$ follows the multivariate $U([0,1])$ distribution with covariance matrix $\bSigma$ in which $\bSigma_{ij} = 0.1*0.8^{|i-j|}$. Because $\bY_i$ is generated as \eqref{het.model}, the true number of factors is 2 for $\tau=0.2$ and 6 for $\tau=0.8$ as in the i.i.d. case. The computational setting is the same as the i.i.d. case.

Figure \ref{facno_ts} shows the relative frequency of the estimated number of factors and the estimated penalized testing error when the regressors follow \eqref{eq:dep_regressor}. It appears that the presense of time dependency slightly decreases the recovery accuracy, but the pattern of the penalized testing error and the estimation performance of the number of factors remain similar to the i.i.d. case in Section \ref{sec:facno}. However, for $\tau=0.8$, smaller $\kappa$ and greater $T$ than than those for $\tau=0.2$ are selected to ensure estimation accuracy, which is due to the fact that the true number of factors for $\tau=0.8$ is greater than that of $\tau=0.2$.
\begin{figure}[!h] 
	\centering
\includegraphics[width=7cm, height = 7cm]{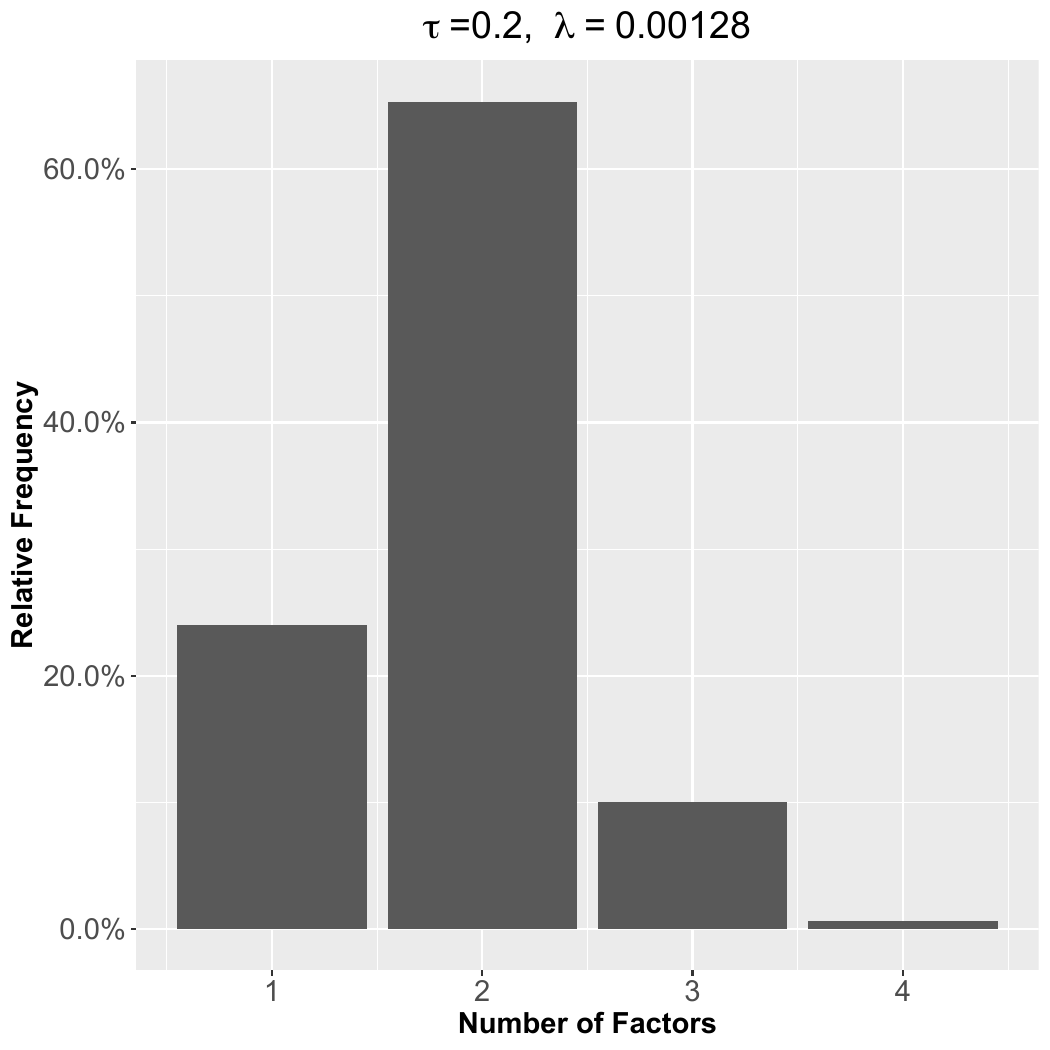}
\includegraphics[width=7cm, height = 7cm]{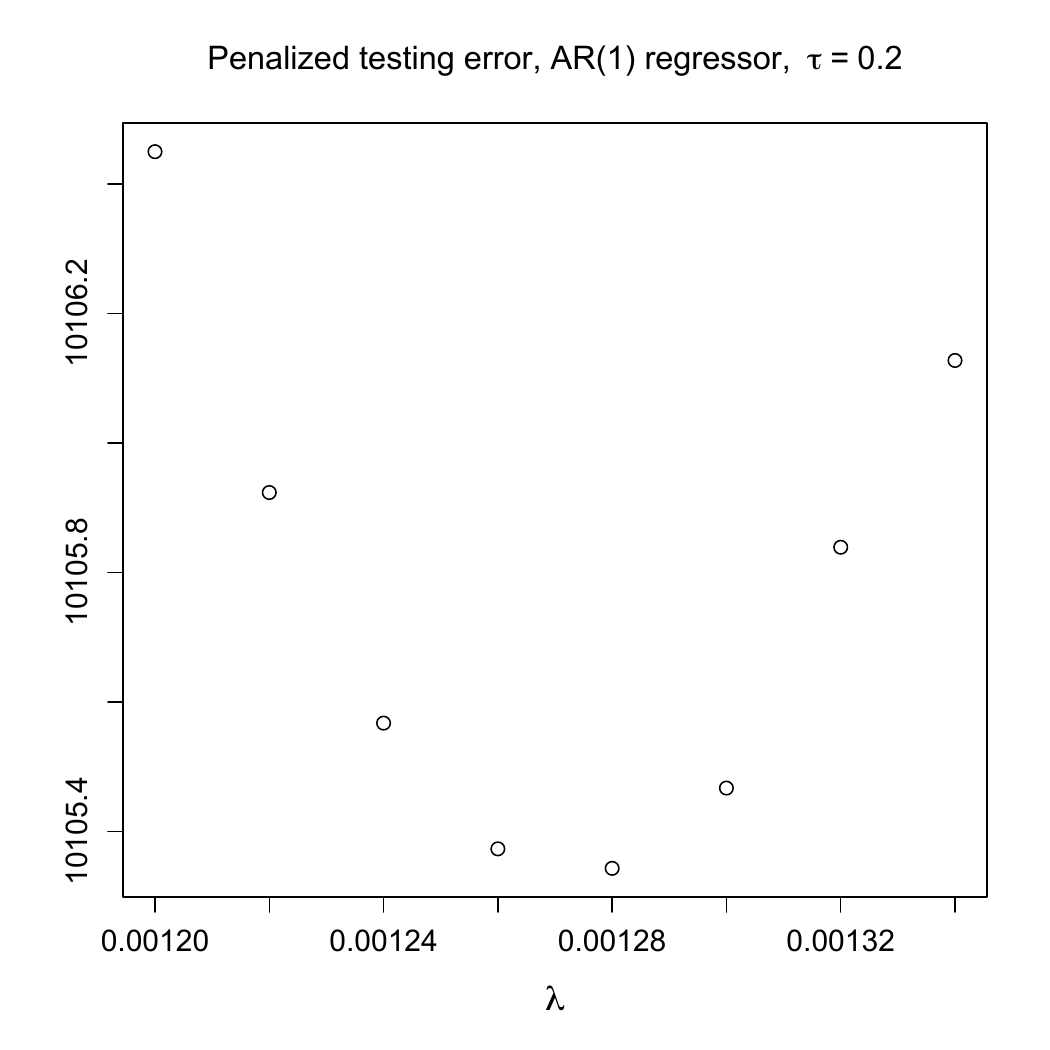}\\
\includegraphics[width=7cm, height = 7cm]{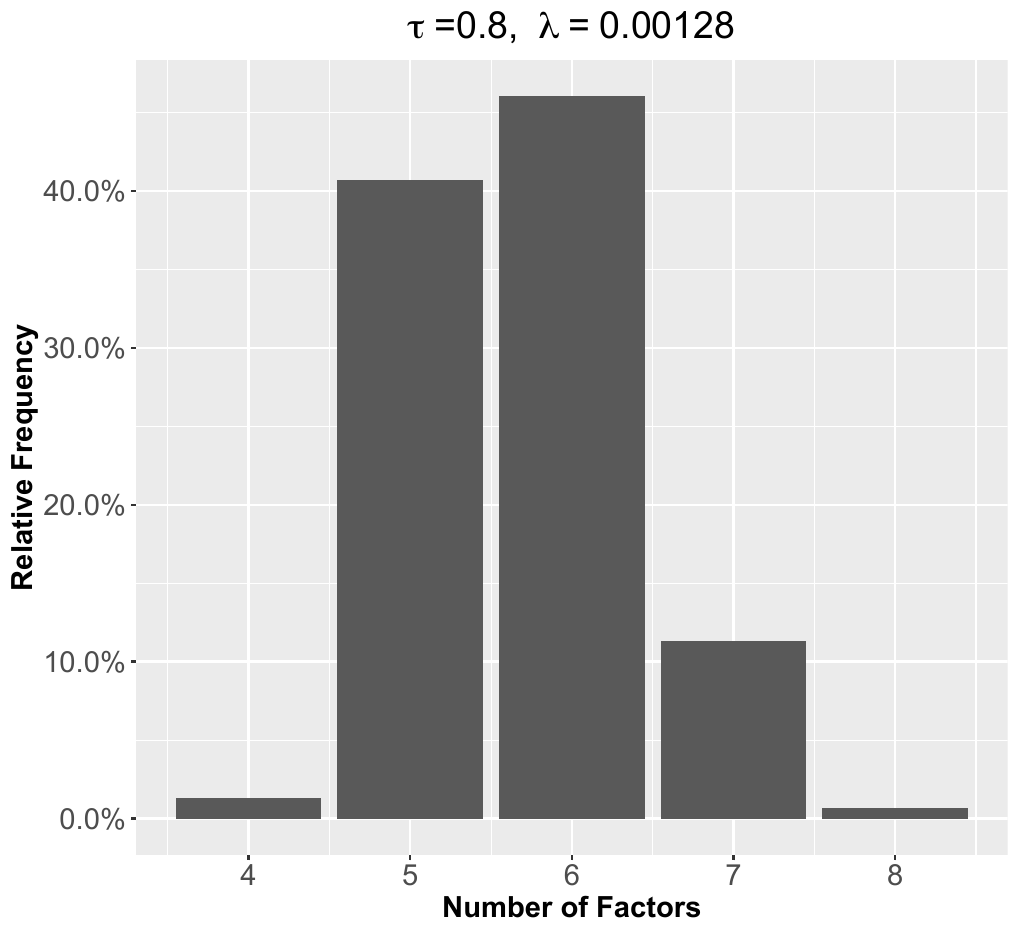}
\includegraphics[width=7cm, height = 7cm]{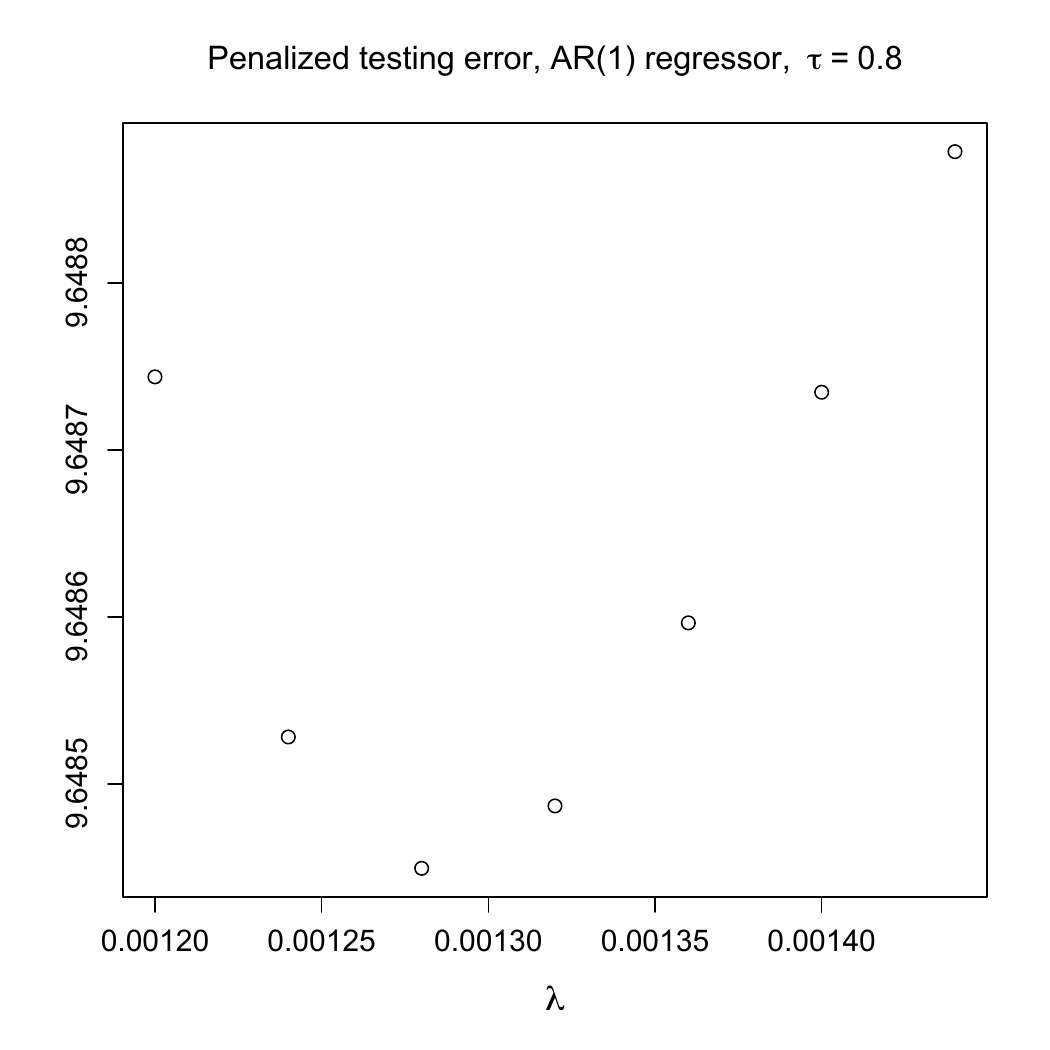}
\caption{The histogram of the estimated number of factors and the plot for the penalized testing error computed by the average of 150 Monte Carlo repetitions, $\tau=0.2$ and 0.8. Data are generated as \eqref{het.model}, with AR(1) regressor $\bX_i$ generated as \eqref{eq:dep_regressor}. The true number of factors is 2 for $\tau=0.2$ and 6 for $\tau=0.8$. $(\kappa,T)=(6.66*10^{-6},3500)$ for $\tau=0.2$ and $(\kappa,T)=(8*10^{-7},4000)$ for $\tau=0.8$.}\label{facno_ts}
\end{figure}

\end{document}